\documentclass[letterpaper]{article}
\usepackage[totalwidth=480pt, totalheight=680pt]{geometry}

\usepackage[T1]{fontenc}
\usepackage{amsfonts,amsmath,amssymb,amsthm}
\usepackage{algorithm}
\usepackage{algorithmicx}
\usepackage{algpseudocode}
\usepackage{booktabs}
\usepackage{tabularx}
\usepackage{multirow}
\usepackage{url}
\usepackage[round]{natbib}
\usepackage{xcolor}
\definecolor{darkgreen}{rgb}{0,0.5,0}
\usepackage{hyperref}
\hypersetup{
    colorlinks=true,
    linkcolor=red,
    citecolor=darkgreen,
    filecolor=magenta,
    urlcolor=blue
}
\usepackage{authblk}

\newtheorem{theorem}{Theorem}[section]

\newtheorem{lemma}[theorem]{Lemma}
\newtheorem{corollary}[theorem]{Corollary}
\theoremstyle{definition}
\newtheorem{definition}[theorem]{Definition}

\theoremstyle{remark}

\makeatletter
\providecommand*{\theHALG@line}{}
\renewcommand*{\theHALG@line}{\thealgorithm.\arabic{ALG@line}}
\makeatother

\newcommand{\MMS}{\textsf{MMS}}
\newcommand{\USW}{\mathrm{USW}}
\newcommand{\OPTUSW}{\mathrm{OPT}_{\mathrm{USW}}}

\title{\bf Online Fair Division with Upgrade Recourse}
\author[1]{Yichen Huang}
\author[2]{Nicholas Teh}
\affil[1]{Harvard University, USA}
\affil[2]{University of Oxford, UK}
\date{}

\begin{document}

\maketitle

\begin{abstract}
We study the online fair division of indivisible goods. In the standard model, where assignments are irrevocable, no algorithm can guarantee any positive fraction of the maximin share (MMS). Reassigning earlier goods can restore fairness, but it may take value away from agents who already hold them. We ask how much fairness can be maintained if every reassignment must leave each affected agent at least as well off. We consider two forms of such \emph{upgrade recourse}. Under \emph{item-level upgrades}, each good an agent holds must be kept or replaced by a distinct good she values at least as much; under \emph{bundle-level upgrades}, only her total value must not decrease. For item-level upgrades, we show that the optimal MMS factor is essentially $1/H_n \approx 1/\ln n$, where $H_n$ is the $n$-th harmonic number, and that fixed picking sequences nearly attain it. Since these sequences use only the agents' rankings of individual goods, knowing the actual values cannot improve the optimal factor for up to 33 agents, and improves it by at most a vanishing margin as $n$ grows. Bundle-level upgrades allow substantially stronger guarantees. By permitting several-for-one exchanges, they turn this vanishing factor into a constant: we achieve exact MMS for two agents and more than half of MMS for any number of agents, in both cases together with relaxations of envy-freeness. Finally, we show that predictions about future goods let online algorithms inherit any offline MMS guarantee, losing only an additive term bounded by the prediction error. Under item-level upgrades, however, insisting on exact MMS when predictions are correct leaves no positive guarantee when they are wrong.
\end{abstract}

\section{Introduction}
Consider teams sharing a growing pool of accelerators for independent training experiments or deployed inference services. Each team values a device by the capacity it provides for its workload, so the teams need not agree on which device is best. Fair access and differences in performance across hardware and workloads are central concerns in machine learning cluster scheduling~\citep{mahajan2020themis,narayanan2020gavel}. As new devices become available, reassignment can improve fairness. 
However, taking a device from one team may reduce the capacity available for its planned experiments or its deployed service. We therefore ask that each update preserve the value every team already receives.

This requirement connects fair division to a standard concern in safe deployment: an update should preserve an existing performance guarantee. Safe policy improvement in reinforcement learning, for example, compares a new policy with a baseline policy~\citep{laroche2019safe}; here, we compare successive allocations, separately for every recipient. We model the teams as $n$ agents and the devices as indivisible goods. Each good's values are revealed on arrival; it must be assigned immediately and remains available for reassignment. Since the workloads are independent, we use additive valuations.

The key benchmark we consider is the widely studied \emph{maximin share} (MMS)~\citep{budish2011combinatorial,amanatidis2023fair}: the value an agent could secure by partitioning the available goods into $n$ bundles and receiving a least-valued bundle. An $\alpha$-MMS allocation gives every agent at least an $\alpha$ fraction of her own benchmark. We also study \emph{envy-freeness up to one good} (EF1), which requires that any envy toward another agent disappear after removing some good from that agent's bundle, and its stronger variant \emph{EFX}, which requires this after removing any positively valued good. With irrevocable assignments, no positive approximation to MMS or EF1 is possible for general additive valuations, even when the number of arrivals is known~\citep{neoh2026closinggaps}. Randomization does not ensure MMS with high probability either: against an oblivious adversary, which fixes the input in advance, for every $\alpha>0$ every randomized algorithm has an input on which it returns an $\alpha$-MMS allocation with probability at most $5/6$~\citep{cohen2026oblivious}. Allowing reassignment, or \emph{recourse}, can restore fairness~\citep{HeProcacciaPsomasZeng2019}. The question is whether it can do so while preserving each recipient's earlier value.

We introduce two forms of \emph{upgrade recourse}. Under \emph{item-level upgrades}, each previously assigned good must be retained or replaced by a distinct good that its owner values at least as much. Under \emph{bundle-level upgrades}, only the owner's total value must be nondecreasing. Replacing two devices worth three units each with one worth seven increases a team's capacity and frees the old devices for other teams. Bundle-level upgrades permit this exchange; item-level upgrades require two distinct replacements. Thus, under item-level upgrades, an agent never holds fewer goods than before, and none of her goods is replaced by one she values less. Neither rule bounds the number of reassignments. Our question is: \emph{To what extent can fairness be maintained as goods arrive if every revision must leave each affected agent at least as well off?}

\paragraph{Scheduling and predictions.}
For identical additive valuations, goods are jobs and agents are machines, and MMS is the optimal minimum load in \emph{machine covering}~\citep{galvez2020symmetry}; advance information about jobs gives a semi-online model~\citep{tan2007semi}. Learning-augmented algorithms instead use predictions that may be wrong, seeking better guarantees when they are accurate while limiting the loss when they are not~\citep{lykouris2018caching,cohen2023general}; this approach has been applied to online Nash social welfare and fair division~\citep{banerjee2022online,neoh2026online}. In the accelerator example, procurement plans may predict the capacity of future devices and justify giving a team fewer resources now. If those devices do not arrive, restoring fairness may require revising earlier assignments. Thus the value of predictions depends on which past decisions can be changed.

\subsection{Our Results}
All algorithms are deterministic, and valuations are additive unless stated otherwise. Table~\ref{tab:main-results} summarizes the main MMS bounds.

\paragraph{Fairness under item-level upgrades.}
A simple insertion procedure maintains, after every arrival, the allocation that any fixed picking sequence produces on the goods observed so far; this relies on the monotonicity of picking sequences when goods are added~\citep{chakraborty2021picking}. We show that no algorithm using only the agents' rankings of individual goods can guarantee a higher MMS factor than the best fixed picking sequence. Let $\alpha_n^{\mathrm{ord}}$ denote the supremum MMS factor for these \emph{ordinal} algorithms, and let $\alpha_n$ denote the supremum factor when the algorithm may also use the revealed cardinal values. We show that $\alpha_n^{\mathrm{ord}}=\alpha_n$ for $2\le n\le33$ and determine these factors exactly, including $\alpha_2=2/3$ and $\alpha_3=1/2$. For every $n$, we characterize $\alpha_n^{\mathrm{ord}}$ as the limit of a sequence of finite optimization problems. For $n\ge3$, both factors are strictly below $1/H_n$, where $H_n=\sum_{r=1}^n1/r$, yet their gaps from $1/H_n$ are $o(n^{-k})$ for every fixed $k>0$. In contrast, if EF1 is also required, the optimal MMS factor drops to $1/n$.

\paragraph{Bundle-level upgrades admit stronger guarantees.}
We show that bundle-level upgrades permit exact MMS, EFX, and Pareto optimality after every arrival for two additive agents, without predictions. For $n\ge3$, we obtain EF1, $1/2$-EFX, and $(n-1)/(2n-3)$-MMS. These MMS factors exceed the item-level upper bound for every $n$. For identical valuations, we maintain the classical longest-processing-time-first guarantee of $(3n-1)/(4n-2)$-MMS, together with EFX and Pareto optimality. Yet no factor above $10/11$ is possible under bundle-level upgrades for $n\ge3$, even with identical valuations. Thus several-for-one exchanges yield a constant MMS guarantee without predictions, although preserving every agent's value still prevents exact MMS for $n\ge3$.

\paragraph{Learning-augmented algorithms.}
Predictions of each agent's ranked values (her singleton values in nonincreasing order) determine a picking sequence that transfers any offline $\alpha$-MMS guarantee to the online setting. Agent $i$ receives at least $\alpha\MMS_i(G)-E_i$, where $E_i$ sums absolute differences between her predicted and realized ranked values after zero padding. This error measure also accounts for an incorrect number of goods. Nevertheless, under item-level upgrades, exact MMS with correct predictions rules out every positive MMS factor for arbitrary predictions, even when the number of goods is known and valuations are identical. For bundle-level upgrades, two agents can attain exact MMS regardless of the prediction. For $n\ge3$, however, no factor above $10/11$ is possible for arbitrary predictions, even when the number of goods is known. Predictions can improve fairness, but their consequences depend on which earlier assignments may be revised.

\begin{table}[t]
\caption{MMS guarantees with an unknown number of goods and no predictions, for deterministic algorithms and additive valuations. Bold entries are optimal. Here $\alpha_n$ is the optimal MMS factor under item-level upgrades. Except in the EF1 row, the upper bounds impose no additional fairness requirement.}
\label{tab:main-results}
\begin{center}
\small
\setlength{\tabcolsep}{3.5pt}
\renewcommand{\arraystretch}{1.2}
\begin{tabularx}{\textwidth}{@{}p{.205\textwidth} c c c X@{}}
\toprule
Recourse & $n$ & Achievable MMS & Upper bound & Implication \\
\midrule
None~\citep{neoh2026closinggaps} & $\ge2$ & $\mathbf{0}$ & $\mathbf{0}$ & Reassignment is necessary. \\
\midrule
Item-level upgrades & $2$ & $\mathbf{2/3}$ & $\mathbf{2/3}$ & Exact optimal factor. \\
 & $3$ & $\mathbf{1/2}$ & $\mathbf{1/2}$ & Exact optimal factor. \\
 & $4$--$33$ & $\boldsymbol{\alpha_n}$ & $\boldsymbol{\alpha_n}$ & Cardinal values do not improve on fixed picking sequences; see Table~\ref{tab:exact-factors}. \\
 & $\ge34$ & $1/H_n-o(n^{-k})$ & $<1/H_n$ & Every fixed $k>0$; a fixed sequence suffices. Exact factors remain open. \\
Item-level upgrades + EF1 & $\ge2$ & $\mathbf{1/n}$ & $\mathbf{1/n}$ & Better MMS requires giving up EF1. \\
\midrule
Bundle-level upgrades & $2$ & $\mathbf{1}$ & $\mathbf{1}$ & Exact MMS, EFX, and PO together. \\
 & $\ge3$ & $(n-1)/(2n-3)$ & $10/11$ & More than half of MMS, with EF1 and $1/2$-EFX. \\
Identical, bundle-level upgrades & $\ge3$ & $(3n-1)/(4n-2)$ & $10/11$ & More than three quarters of MMS, with EFX and PO. \\
\bottomrule
\end{tabularx}
\end{center}
\smallskip
\begin{minipage}{\textwidth}\footnotesize
Item-level proofs are in Appendix~\ref{app:item-proofs}, with exact factors in Table~\ref{tab:exact-factors} and an explicit guarantee for every $n$ in Corollary~\ref{cor:general-deadline-mms}. Bundle-level proofs are in Appendix~\ref{app:value-proofs}; their upper bounds hold for identical valuations. We do not claim polynomial running time for the two-agent exact MMS algorithm.
\end{minipage}
\end{table}

\subsection{Related Work}
We discuss the most closely related work here. Appendix~\ref{app:further_related_work} gives a more detailed discussion, including temporal voting, recourse in online algorithms, offline MMS guarantees, and connections to scheduling.

Dynamic fair division permits reassignment, often with a bound on disruptions~\citep{friedman2015dynamic,HeProcacciaPsomasZeng2019}. In online house allocation, where each agent receives exactly one house, recourse is inherent, and envy-freeability with subsidies can be maintained using bounded recourse~\citep{teh2026house}. We instead constrain the effect of reassignment on each recipient: bundle-level upgrades impose resource monotonicity, while item-level upgrades require distinct, weakly better replacements. A small number of reassignments can still substantially reduce an agent's value; conversely, many reassignments may leave every agent better off. Temporal fair division also requires fairness after every round, but keeps assignments irrevocable: even when all future goods are known, an allocation that is EF1 after every round need not exist~\citep{elkind2025temporal}. Under item-level upgrades, balanced picking sequences maintain EF1 after every arrival without any information about future goods (Theorem~\ref{thm:balanced-serial-insertion}). Earlier ordinal MMS guarantees use deadlines for picking turns~\citep{amanatidis2016truthfulMMS,halpern2021partial}; our upper bounds hold even for algorithms that use revealed cardinal values. For bundle-level upgrades, some EF1 and EFX guarantees follow from existing procedures for extending partial allocations~\citep{bu2025welfare}; our MMS guarantees do not.

Predictions have been studied for online allocation~\citep{banerjee2022online,cohen2023general} and, more specifically, for online MMS, proportionality up to one good, and EFX~\citep{neoh2026online,choo2026approximate,melissourgos2026onlineefx}, as well as for budget-constrained online fair division and online house allocation~\citep{cohen2026budget,teh2026house}. \citet{neoh2026online} already use value-frequency predictions and analyze prediction errors through a matching between predicted and arriving goods. With upgrade recourse, our error bound depends only on the sorted predicted and realized values and also allows an incorrect number of goods. We further compare guarantees for correct and arbitrary predictions. Unlike in bandit fair division~\citep{procaccia2024honor}, each good's values are observed on arrival, and the uncertainty concerns only future goods.

\section{Preliminaries}
\label{sec:preliminaries}
Let $[k]=\{1,\dots,k\}$, with $[0]=\varnothing$, and $\mathbb N_+=\{1,2,\dots\}$. For a real number $z$, let $[z]_+=\max\{z,0\}$, and let $H_n=\sum_{r=1}^n1/r$. There are $n\ge2$ \emph{agents} and $m\ge1$ indivisible \emph{goods} $G=\{g_1,\dots,g_m\}$, arriving in this order. 
Let $G^{(t)}=\{g_1,\dots,g_t\}$ and $G^{(0)}=\varnothing$. Each good must be allocated on arrival and remains available for reassignment. Unless a prediction is supplied, the algorithm knows nothing about future goods, including their number $m$ (which we call the \emph{horizon}). 

Agent $i$ has a fixed nonnegative valuation $v_i:2^G\to\mathbb R_{\ge0}$. Unless stated otherwise, valuations are \emph{additive}: $v_i(S)=\sum_{g\in S}v_i(g)$, where $v_i(g)=v_i(\{g\})$. When a good $g$ arrives, all values $v_i(g)$ are revealed. For nonadditive results, valuations are \emph{normalized}, $v_i(\varnothing)=0$, and \emph{monotone}, $v_i(S)\le v_i(T)$ whenever $S\subseteq T$. Where stated, they are also \emph{subadditive}: $v_i(S\cup T)\le v_i(S)+v_i(T)$. We allow value queries on subsets of observed goods. Normalization refers only to $v_i(\varnothing)=0$; in particular, $v_i(G)$ is not known in advance.

An allocation $\mathcal A=(A_1,\dots,A_n)$ partitions the goods into possibly empty bundles. After arrival $t$, the algorithm returns a complete allocation $\mathcal A^{(t)}$ of $G^{(t)}$; set $A_i^{(0)}=\varnothing$ and $\mathcal A=\mathcal A^{(m)}$. Fairness and efficiency refer to the final allocation. Without predictions, they must hold after every arrival, since arrivals could stop there. Guarantees conditional on a correct prediction need hold only for the realized final instance. Recourse constrains \emph{every} update, including when predictions are incorrect.
All algorithms are deterministic, and guarantees are worst-case over arrival sequences.

\subsection{Fairness and Efficiency}
An allocation is \emph{envy-free} (EF) if $v_i(A_i)\ge v_i(A_j)$ for all $i,j$. Since an envy-free allocation need not exist, even offline with two agents and one good that both value positively, we consider the following standard relaxations of EF. Approximation factors lie in $[0,1]$, and larger factors are better.

\begin{definition}[EF1 and EFX]
For $\alpha\in[0,1]$, an allocation is \emph{$\alpha$-EF1} if, for every $i,j\in[n]$ with $A_j\ne\varnothing$, some $g\in A_j$ satisfies $v_i(A_i)\ge\alpha v_i(A_j\setminus\{g\})$. It is \emph{$\alpha$-EFX} if this inequality holds for every $g\in A_j$ with $v_i(g)>0$. When $\alpha=1$, it is simply EF1 and EFX, respectively.
\end{definition}
For a set $H$ and $k\in\mathbb N_+$, let $\Pi_k(H)$ be the set of its ordered $k$-partitions, allowing empty parts.
\begin{definition}[MMS]
Let agent $i$'s $k$-maximin share be $\MMS_i^k(H)=\max_{(P_1,\dots,P_k)\in\Pi_k(H)}\min_{\ell\in[k]}v_i(P_\ell)$, and write $\MMS_i(H)=\MMS_i^n(H)$. For $\alpha\in[0,1]$, an allocation of $H$ is \emph{$\alpha$-MMS} if $v_i(A_i)\ge\alpha\MMS_i(H)$ for every $i\in[n]$.
\end{definition}
Removing one good cannot decrease MMS when the number of parts is also reduced by one (Lemma~\ref{lem:mms-deletion}); we refer to this fact as \emph{MMS deletion}. This lets us account for earlier picks and set aside singleton bundles in the proofs.

We also consider a commonly studied notion of economic efficiency in fair division.
An allocation is \emph{Pareto optimal} (PO) if no allocation $\mathcal B$ of the same goods has $v_i(B_i)\ge v_i(A_i)$ for all $i$, with at least one strict inequality. This comparison includes allocations that recourse may prevent the algorithm from reaching.

\subsection{Recourse}
We place no bound on the number of reassignments. Instead, we consider two forms of \emph{upgrade recourse}, which compare individual goods or entire bundles.
\begin{definition}[Item-level upgrade recourse]
For bundles $X,Y$, write $Y\succeq_i^{\mathrm{up}}X$ if there is an injection $\phi:X\to Y$ with $v_i(\phi(g))\ge v_i(g)$ for every $g\in X$. An algorithm has \emph{item-level upgrade recourse} if $A_i^{(t)}\succeq_i^{\mathrm{up}}A_i^{(t-1)}$ for every agent $i$ and arrival $t$.
\end{definition}
The injection requires distinct replacements, with a retained good allowed to replace itself.
In particular, an agent's number of goods can never decrease. This is appropriate when each good serves a separate need, but it can be restrictive when the agent cares only about her total value. An exchange of several goods for one may increase her value and free goods for other agents, yet item-level upgrade recourse excludes it. Our second recourse rule permits such exchanges.
\begin{definition}[Bundle-level upgrade recourse]
An algorithm has \emph{bundle-level upgrade recourse} if $v_i(A_i^{(t)})\ge v_i(A_i^{(t-1)})$ for every agent $i$ and arrival $t$.
\end{definition}
For additive valuations, item-level upgrade recourse implies bundle-level upgrade recourse by summing the individual replacement inequalities. For nonadditive valuations, the implication need not hold, even under subadditivity; Appendix~\ref{app:recourse-comparison} gives an example.

\subsection{Picking Sequences and Ranked Values}
Picking sequences are a standard way to allocate goods using only the agents' rankings of individual goods: agents take turns choosing their most valuable remaining good. Round-robin is a special case, and the MMS guarantee of a picking sequence depends on the positions of each agent's turns~\citep{amanatidis2016truthfulMMS,halpern2021partial}. Their monotonicity when goods are added makes them particularly useful for upgrade recourse~\citep{chakraborty2021picking}.

Fix an infinite \emph{picking sequence} $\sigma=(\sigma_1,\sigma_2,\dots)\in[n]^{\mathbb N_+}$. Its picking allocation on $G^{(t)}$ lets agent $\sigma_k$ choose her most valuable remaining good at turn $k\in[t]$. Picking considers all goods in $G^{(t)}$, so the good chosen at turn $k$ need not be $g_k$. Ties are broken by a fixed deterministic order over goods, and the order among old goods does not change when new goods arrive. Write $x\succ_i^{\mathrm{tb}}y$ when agent $i$ ranks $x$ above $y$, comparing by $v_i$ and breaking ties by this order; in particular, $v_i(x)\ge v_i(y)$.

A sequence is \emph{balanced} if each agent's number of turns differs from every other's by at most one in every prefix. Equivalently, each successive block of $n$ turns is a permutation of the agents. A \emph{round-robin} sequence repeats the same permutation.

For any sequence $\sigma$, let $P_i^\sigma(k)=\{\ell\in[k]:\sigma_\ell=i\}$ be the set of agent $i$'s turns among the first $k$, and $N_i^\sigma(k)=|P_i^\sigma(k)|$. For a finite set $H$ of $m$ goods, list $i$'s singleton values in nonincreasing order as $x_1^i(H)\ge\cdots\ge x_m^i(H)$, which we call her \emph{ranked values}, and define the \emph{rank value} $R_i^\sigma(H)=\sum_{k\in P_i^\sigma(m)}x_k^i(H)$, where $R_i^\sigma(\varnothing)=0$.
Before turn $k$, at most $k-1$ goods have been taken, so at least one of the choosing agent's $k$ most valuable goods remains. Hence agent $i$'s value in the picking allocation on $H$ is at least $R_i^\sigma(H)$. The \emph{ordered instance} of $H$ has goods $e_1,\dots,e_m$ and additive values $v_i^{\mathrm{ord}}(e_k)=x_k^i(H)$. It preserves each agent's multiset of singleton values, and hence her MMS, while making all agents rank the goods in the same order.

\section{Item-Level Upgrade Recourse}
\label{sec:goods-item}
Item-level upgrades leave room to revise an allocation, but not to take away an earlier good without replacing it. The challenge is therefore to choose fair allocations that can still be upgraded to fair allocations as the set of goods grows. Recomputing a fair allocation from scratch does not ensure this: an agent might have to give up a good for several goods that are individually less valuable.

Our approach fixes who owns each picking turn, rather than which good that turn receives. When a more desirable good arrives, an earlier turn can take it and pass its old choice to a later turn. The monotonicity argument of \citet{chakraborty2021picking} ensures that no earlier choice becomes worse. We implement these changes by insertion, maintaining the picking allocation on all observed goods. This separates two questions: how to preserve earlier choices, and how to choose the turns to obtain fairness. We first recover EF1 using balanced turns, then choose different orders to obtain stronger MMS guarantees.

\subsection{Insertion and EF1}
After $g_t$ arrives, call $g_t$ the unallocated good, and let each earlier picking turn, in order, compare its current good with the unallocated good. If the turn's owner prefers the unallocated good, the turn takes it, and its old good becomes the unallocated good. Otherwise, the turn keeps its current good. The final unallocated good goes to the new turn $t$. We call this procedure \emph{insertion}; Algorithm~\ref{alg:serial-insertion} in Appendix~\ref{app:serial-insertion} gives the pseudocode.
\begin{lemma}
\label{lem:serial-insertion}
Fix any infinite picking sequence $\sigma$.
Algorithm~\ref{alg:serial-insertion} run with $\sigma$ satisfies item-level upgrade recourse. Moreover, after each arrival $t$, the maintained allocation $\mathcal A^{(t)}$ is exactly
the picking allocation induced by $(\sigma_1,\dots,\sigma_t)$ on
$G^{(t)}$.
\end{lemma}

For recourse, each old good is mapped to the good that the same turn holds after the update; these replacements are distinct because different turns hold different goods. An update takes $O(n+t)$ time, for $O(nm+m^2)$ overall, and there are at most $m(m-1)/2$ reassignments.

Balanced turns now give the familiar EF1 guarantee in the online setting. To see why, compare agents $i$ and $j$: the good chosen at $j$'s $(\ell+1)$st turn is still available at $i$'s $\ell$th turn. Agent $i$ therefore values her choice at least as much. Pairing these goods leaves only $j$'s first good unmatched, which is precisely the good removed for EF1. The implication from EF1 to $1/n$-MMS is standard~\citep{amanatidis2018relaxations}.
\begin{theorem}\label{thm:balanced-serial-insertion}
For every $n\ge2$ and every balanced sequence $\sigma$, Algorithm~\ref{alg:serial-insertion} run with $\sigma$ satisfies item-level upgrade recourse, and its final allocation is EF1 and $1/n$-MMS.
\end{theorem}

\subsection{Stronger MMS Guarantees}
The $1/n$-MMS guarantee gives every agent a positive fraction of her benchmark, but that fraction becomes small as the number of agents grows. Can we do better while preserving every earlier choice? Balanced turns are natural for controlling envy, but MMS compares an agent with her own partition benchmark. An early turn may secure much of the available value; an agent who starts later may instead need more subsequent turns. This suggests departing from balanced orders. For two agents, the order $(1,2,2,1,2,2,\dots)$ compensates the second agent with more later turns. Its offline $2/3$-MMS guarantee is due to \citet{amanatidis2016truthfulMMS}; insertion makes it feasible online.
The factor $2/3$ is optimal even for arbitrary item-level upgrade algorithms (Theorem~\ref{thm:tight-two-agent-mms}).

For general $n$, we first give each agent one turn, in the order $1,\ldots,n$, and compensate later starters with more subsequent turns. After the first $i-1$ picks, MMS deletion leaves $n-i+1$ parts without lowering agent $i$'s benchmark. To obtain $\alpha$-MMS, it suffices to give her an $\alpha/(n-i+1)$ fraction of the total value of the remaining goods. Picking guarantees this value fraction when she receives at least the same fraction of turns in every prefix starting at her first turn. We enforce this by assigning deadlines to her successive turns, spaced roughly $(n-i+1)/\alpha$ positions apart, and ordering all agents' turns by deadline. The requested fractions sum to $\alpha H_n$, which must be at most one, explaining the target $1/H_n$. The difficulty is meeting every deadline after rounding to integer positions. We show that the loss needed for this rounding decreases faster than any inverse polynomial in $n$.
\begin{theorem}\label{thm:near-harmonic}
There is a constant $c>0$ such that, for every sufficiently large $n$, insertion with a fixed picking sequence guarantees $(1/H_n-\exp[-c(\log n)^{3/2}])$-MMS. Conversely, for every $n\ge3$, the supremum MMS factor of deterministic algorithms with item-level upgrade recourse is strictly below $1/H_n$, even when they use the revealed cardinal values.
\end{theorem}

For every fixed $k>0$, both optimal factors are therefore within $o(n^{-k})$ of $1/H_n$; in particular, using cardinal values instead of rankings can improve the MMS factor by at most $o(n^{-k})$. Appendix~\ref{app:uniform-mms} proves the theorem; Corollary~\ref{cor:general-deadline-mms} gives a simpler explicit factor of at least $1/(H_n+1)$ for every $n\ge2$.

For small numbers of agents, we determine the optimal factors exactly. Do fixed picking sequences remain optimal, or can an unrestricted algorithm use cardinal values to decide who should receive an early good? For three agents, we show that this extra information does not improve the factor $1/2$. The lower bound uses a fixed sequence; the upper bound chooses later values depending on who receives the first two goods. Appendix~\ref{app:three-agent-tight} gives the six-good example. The same comparison between achievable picking guarantees and upper bounds for unrestricted algorithms yields exact factors through $33$ agents.

\begin{theorem}\label{thm:finite-exact}
For every $2\le n\le33$, $\alpha_n=\alpha_n^{\mathrm{ord}}$, and insertion with a fixed picking sequence attains this factor. In particular, $\alpha_2=2/3$ and $\alpha_3=1/2$; Table~\ref{tab:exact-factors} gives the remaining exact factors. The upper bounds apply to arbitrary item-level upgrade algorithms, not only ordinal ones.
\end{theorem}

\subsection{The Role of Rankings}
The guarantees above use fixed picking sequences, even though new rankings become available after every arrival. Could an ordinal algorithm use them to choose who picks next, or to reassign goods in ways that fixed picking cannot reproduce? Theorem~\ref{thm:ordinal-optimality} shows that neither kind of adaptation improves the worst-case MMS factor.

To see why, suppose every agent strictly prefers earlier arrivals to later ones. No old good can change owners: the first good has no better replacement, then neither does the second, and so on. On these instances, every ordinal algorithm sees the same rankings, follows a fixed sequence, and gives each agent exactly her rank value. Insertion with that sequence gives each agent at least her rank value on every instance. Appendix~\ref{app:ordinal-characterization} proves this claim and gives a finite optimization approximating $\alpha_n^{\mathrm{ord}}$ to any desired accuracy. This characterizes $\alpha_n^{\mathrm{ord}}$; whether $\alpha_n^{\mathrm{ord}}=\alpha_n$ for $n\ge34$ remains open.

\subsection{Limits of Item-Level Upgrades}
Allocations achieving the stronger MMS guarantees need not be EF1, and this is unavoidable: Theorem~\ref{thm:goods-ef1-mms-tradeoff} shows that requiring EF1 reduces the optimal MMS factor to $1/n$. EF1 spreads early small goods nearly evenly. If large goods arrive later, their recipients cannot give up enough small goods, so an agent without a large good gets only about a $1/n$ fraction of her MMS. Thus the choice of picking turns reflects a genuine conflict between the two fairness notions, not a limitation of insertion. Appendix~\ref{app:ef1-mms} gives the full tradeoff for approximate EF1, and Appendix~\ref{sec:limits} gives further impossibility results for EFX and efficiency.

\section{Bundle-Level Upgrade Recourse}
\label{sec:value-monotone}
The preceding upper bounds rely on an important restriction: receiving one valuable good does not allow an agent to release several less valuable goods, even when her total value would increase. Bundle-level upgrades remove this restriction, permitting exchanges that leave more goods for agents whose MMS targets are not yet met.

Several familiar fair-division operations already preserve each recipient's value. Envy-cycle elimination does so by giving each participating agent a preferred bundle~\citep{lipton2004approximately}. Applying the partial-allocation procedures of \citet{bu2025welfare} after each arrival yields polynomial-time bundle-level algorithms for EF1 with any number of additive agents and EFX with two (Appendix~\ref{app:partial-allocations}). These procedures do not settle MMS: they ensure that no agent's value decreases, but an agent's MMS may increase more than her value does.

\subsection{Exact MMS for Two Agents}
Can we meet each agent's new MMS while preserving her old value? EFX alone does not ensure this: for two identical agents and goods worth $3,3,2,2$, assigning equal-valued goods together is EFX, but one agent receives $4$ when her MMS is $5$. Nevertheless, exact MMS, EFX, and PO can coexist after every arrival (Appendix~\ref{app:two-agent-exact}).
\begin{theorem}
\label{thm:two-agent-value-monotone-exact}
    For two agents with nonnegative additive valuations, there is a deterministic online algorithm with bundle-level upgrade recourse whose allocation after every arrival is exact MMS, EFX, and PO. 
\end{theorem}

\subsection{More Than Two Agents}
For more agents, recomputing a fair allocation need not preserve the value of someone who already exceeds her new MMS target. We retain the bundles of agents who meet that target and use exchanges to give the others enough value. The next theorem shows that this permits more than half of MMS, simultaneously with EF1 and approximate EFX.
\begin{theorem}
\label{thm:three-agent-value-monotone}
For every $n\ge3$ and nonnegative additive valuations, there is a deterministic online algorithm with bundle-level upgrade recourse whose allocation after every arrival is EF1, $1/2$-EFX, and $(n-1)/(2n-3)$-MMS. For fixed $n$ and integer singleton values, the algorithm runs in pseudo-polynomial time.
\end{theorem}

Consider an agent still below her updated target. EF1 and a maintained bound on singleton values show that each retained bundle with at least two goods consists of one good and a remainder, each worth less than her target. Retained singletons are handled by MMS deletion. We then show that the unallocated goods can be partitioned into enough bundles meeting her target. We adapt the procedure of \citet{akrami2025simultaneous} to assign minimal such bundles, allowing earlier recipients to exchange their bundles only for ones they prefer. Every remaining agent values each newly assigned bundle below her target, so the argument can continue. Appendix~\ref{app:three-agent-value-monotone} gives the construction and proof. A simpler polynomial-time algorithm achieves $n/(2n-1)$-MMS with the same envy guarantees; Appendix~\ref{app:general-valuations} presents it and its extensions to nonadditive valuations.

These MMS factors exceed the deterministic item-level upper bound for every $n\ge2$ (Corollary~\ref{cor:goods-recourse-comparison}). Nevertheless, bundle-level upgrade recourse does not permit exact MMS in general.
\begin{theorem}
\label{thm:value-monotone-mms-upper}
For every $n\ge3$ and every $\alpha\in(10/11,1]$, no deterministic online algorithm with bundle-level upgrade recourse can guarantee $\alpha$-MMS. This holds for identical additive valuations with singleton values in $[0,1]$. The final instance admits an allocation that is exact MMS, EFX, and PO.
\end{theorem}

In our integer-valued construction, the MMS requirement before the final arrival forces some agent to receive value at least $72$, whereas after the final arrival every $\alpha$-MMS allocation with $\alpha>10/11$ gives her value exactly $66$. Appendix~\ref{app:value-upper} gives the goods and the argument for every $n\ge3$.

\paragraph{Identical valuations.}
Agreement about values does not remove the $10/11$ upper bound, but it permits a stronger positive guarantee. A polynomial-time algorithm maintains $(3n-1)/(4n-2)$-MMS, EFX, and PO by recomputing the longest-processing-time-first allocation and matching bundles to agents in order of their previous values. The sorted bundle values cannot decrease when a good is inserted~\citep{galvez2020symmetry}. Appendix~\ref{app:identical-additive} proves the guarantee; Theorem~\ref{thm:identical-value-monotone} also covers EFX for identical monotone valuations.

\section{Predictions of Future Values}
\label{sec:predictions}
So far, the algorithm has had no information about future goods. In applications, a prediction of future goods may suggest a better allocation: a team expecting valuable devices later can receive fewer early choices, leaving more for other teams. The difficulty is deciding how much to rely on that prediction. When the expected goods do not arrive, earlier assignments may be difficult to revise without making someone worse off.

This is the central concern of learning-augmented online algorithms: exploit accurate predictions while quantifying the loss when they are inaccurate~\citep{lykouris2018caching,banerjee2022online,cohen2023general}. Predictions specify each agent's ranked values, equivalently her multiset of singleton values, as in the value-frequency predictions of \citet{neoh2026online}. Such information describes the anticipated mix of goods without predicting their arrival order or linking different agents' predicted values to the same good. We make no statistical assumption about how the prediction is obtained.

Our positive guarantees use insertion with a fixed picking sequence, so they hold under both recourse rules. The rules differ in the guarantees possible for arbitrary predictions.

\subsection{From Ranked Values to Picking Turns}
Correct frequency predictions already transfer offline share guarantees to irrevocable online allocation~\citep{neoh2026online}. We revisit this transfer to obtain an upgrade algorithm whose performance can be compared directly with the realized ranked values when predictions are wrong. Given an allocation of the predicted ordered instance, we give picking turn $k$ to the agent who receives the good of rank $k$. By the rank-value bound in Section~\ref{sec:preliminaries}, each agent then receives at least the true values of her assigned ranks. Since ordering preserves MMS, an offline $\alpha$-MMS allocation of the ordered instance yields an $\alpha$-MMS allocation online when the predictions are correct (Theorem~\ref{thm:ordered-instance-transfer}, proved in Appendix~\ref{app:correct-predictions}).

In particular, correct ranked values permit exact MMS for two agents, because every two-agent additive instance admits an MMS allocation (Corollary~\ref{cor:two-agent-perfect-order}). They also permit exact MMS for any number of agents with identical valuations, by using an MMS partition of the predicted instance. More generally, since the ordered instance preserves each agent's multiset of singleton values, correct ranked values permit exact MMS in any valuation domain determined by these multisets in which MMS allocations always exist offline, such as the domain of \citet{teh2026lexicographic}, which generalizes weakly lexicographic and several few-valued valuations.

\subsection{Robustness to Errors}
Suppose the prediction gives a horizon $\widehat m\ge0$ and ranked values $\widehat x_1^i\ge\cdots\ge\widehat x_{\widehat m}^i\ge0$ for each agent $i$. Before arrivals begin, apply any offline $\alpha$-MMS algorithm to the predicted ordered instance. For each $k\le\widehat m$, give picking turn $k$ to the agent who receives $e_k$ in its allocation $\mathcal B$, and extend the sequence arbitrarily thereafter. For the error analysis, set $x_k^i(G)=0$ for $k>m$ and $\widehat x_k^i=0$ for $k>\widehat m$.
\begin{theorem}
\label{thm:ranked-value-prediction}
Let $E_i:=\sum_{k\ge1}|x_k^i(G)-\widehat x_k^i|$. If $\mathcal B$ is $\alpha$-MMS in the predicted ordered instance, then insertion with the induced picking sequence satisfies item-level upgrade recourse and, for every $i\in[n]$,
\[
 v_i(A_i)
 \ge \alpha\MMS_i(G)
 -\alpha\sum_{k\ge1}[x_k^i(G)-\widehat x_k^i]_+
 -\sum_{k:e_k\in B_i}[\widehat x_k^i-x_k^i(G)]_+ \ge \alpha\MMS_i(G)-E_i.
\]
Thus, zero error gives $\alpha$-MMS. If $E_i\le\varepsilon_i\MMS_i(G)$, then $v_i(A_i)\ge(\alpha-\varepsilon_i)\MMS_i(G)$.
\end{theorem}

The first inequality distinguishes two effects. Underestimated values can lower the predicted MMS and hence the value that $\mathcal B$ guarantees to an agent. Overestimated values matter directly only at ranks assigned to that agent, so counting all errors in $E_i$ can be pessimistic. Since underestimation and overestimation concern different ranks, their losses together are at most $E_i$, rather than twice that amount. The coefficient of $E_i$ cannot be reduced for arbitrary choices of $\mathcal B$, already for two agents and $\alpha=1$; Appendix~\ref{app:prediction-proofs} gives the proof and example.

Appendix~\ref{app:relative-prediction-error} gives a guarantee under pointwise relative errors, which exclude missing or additional goods, and Appendix~\ref{app:prediction-proofs} bounds the effect of changing early picking turns.

\subsection{Correct Predictions versus Arbitrary Predictions}
An error-dependent additive guarantee need not give a positive MMS factor for every prediction. The guarantee for correct predictions is usually called \emph{consistency}, and the guarantee for arbitrary predictions is called \emph{robustness}~\citep{lykouris2018caching,cohen2023general}. The next theorem bounds their tradeoff under item-level upgrades.

\begin{theorem}
\label{thm:prediction-mms-tradeoff}
    Fix $n\ge2$ and $\alpha,\beta\in[0,1]$. Suppose a deterministic algorithm with item-level upgrade recourse receives the correct horizon and predicted ranked values. If it guarantees $\alpha$-MMS whenever the ranked values are predicted correctly and $\beta$-MMS for every prediction, then $\beta\le n(1-\alpha)/(n-1)$.
    This holds even for identical additive valuations with strictly positive singleton values in $[0,1]$. In particular, guaranteeing exact MMS for correct predictions rules out every positive MMS guarantee for arbitrary predictions.
\end{theorem}

\subsection{Bundle-Level Upgrades}
The item-level impossibility relies on nondecreasing bundle cardinalities, not on prediction errors alone. For two agents, Theorem~\ref{thm:two-agent-value-monotone-exact} achieves exact MMS without predictions, so consistency and robustness can both equal one.

For $n\ge3$ and identical valuations, correct ranked values permit exact MMS (Theorem~\ref{thm:ordered-instance-transfer}), whereas without predictions no factor above $10/11$ is possible (Theorem~\ref{thm:value-monotone-mms-upper}). Can knowing the horizon overcome this upper bound when predictions may be wrong? The next result shows that it cannot.
\begin{theorem}
\label{thm:bundle-prediction-upper}
Fix $n\ge3$. If a deterministic algorithm with bundle-level upgrade recourse receives the correct horizon and predicted ranked values and guarantees $\beta$-MMS for arbitrary predictions, then $\beta\le10/11$. This holds even for identical additive valuations with strictly positive singleton values in $[0,1]$.
\end{theorem}

The construction needs only one incorrectly predicted value per agent, not widespread errors; Appendix~\ref{app:bundle-prediction-upper} gives the proof. Unlike the item-level result, it does not rule out positive robustness with exact consistency for $n\ge3$, so it does not determine the optimal tradeoff.

\section*{Declaration of Generative AI Use}

The authors used GPT-6 as a research assistant when exploring proof ideas and revising the exposition. They
checked the mathematical arguments and take full responsibility for the contents of the paper.

\bibliographystyle{plainnat}
\bibliography{arxiv}
\clearpage
\appendix
\section{Further Related Work} \label{app:further_related_work}

We expand here on the related work discussed in the introduction.

\paragraph{Online fair division.}
Online fair division allocates goods under uncertainty about future arrivals~\citep{AleksandrovEtAl2015,AleksandrovWalsh2020Survey}. When every assignment is irrevocable, share-based guarantees are difficult to obtain. For general additive valuations, no positive approximation to MMS or EF1 is possible~\citep{neoh2026closinggaps}. Even when each agent's valuation is normalized, no positive MMS guarantee is possible for three agents~\citep{zhou2023multi}. Randomization against an oblivious adversary, which fixes the input in advance, does not ensure these guarantees with high probability: for every $\alpha>0$, every randomized algorithm has an input on which it returns an $\alpha$-MMS allocation with probability at most $5/6$, and for every $\alpha\in(0,1]$ the corresponding probability for $\alpha$-EF1 is at most $e^{-\Omega(n)}$~\citep{cohen2026oblivious}. Under budget constraints, where fairness is evaluated only against budget-feasible subsets of each recipient's bundle, no deterministic online algorithm guarantees any fixed approximation to feasible envy-freeness without further structure~\citep{cohen2026budget}. In the related model in which agents arrive instead of goods, no positive MMS guarantee is possible even for two agents without suitable information about the valuations of future agents~\citep{kulkarni2025online}. Previous work avoids these impossibilities in several ways: by weakening the fairness requirement, for example to envy that grows sublinearly in the number of goods~\citep{benade2018envy,zeng2020fairness}; by providing information about future goods~\citep{neoh2026online}; or by delaying assignments using a reordering buffer~\citep{amanatidis2026buffers}. Temporal fair division gives the algorithm complete knowledge of future goods and evaluates the allocation after every round; with irrevocable assignments, an allocation of goods that is EF1 after every round still need not exist, and deciding whether one exists is NP-hard~\citep{elkind2025temporal}. \citet{choi2026temporal} identify tractable cases, including for mixed manna, and show that an allocation of goods that is EF1 after every round gives each agent at least a $1/n$ fraction of her MMS after every round, which is tight; under item-level upgrades, maintaining EF1 likewise limits the MMS factor to $1/n$ (Theorem~\ref{thm:goods-ef1-mms-tradeoff}). A related line of work minimizes the envy that accumulates as a sequence of items is allocated~\citep{goldberg2026cumulative}, or, when the same items are matched to the agents in every round, maximizes the minimum utility after every round~\citep{lim2026repeated}. We instead still assign each good on arrival, but allow earlier assignments to change as long as no agent is made worse off.

\paragraph{Reassignment in dynamic fair division.}
Dynamic fair division usually measures the cost of reassignment by the number of changes. \citet{friedman2015dynamic} consider agents who arrive over time to share divisible resources, and bound the number of disruptions to earlier allocations. For indivisible goods arriving online, \citet{HeProcacciaPsomasZeng2019} maintain EF1 after every arrival by reassigning goods. Without information about future goods, $\Theta(m)$ reassignments are necessary and sufficient for two agents, and $O(m^{3/2})$ suffice for three or more agents. In online house allocation, each agent holds exactly one house, so accepting a newly arrived house may require reassigning earlier ones; envy-freeability with subsidies can always be maintained using bounded recourse, although reassignment chains of length linear in the number of agents are unavoidable in the worst case~\citep{teh2026house}. Such bounds limit how often earlier assignments change, but not how much a single change lowers an agent's value. Upgrade recourse does the reverse: it places no bound on the number of reassignments, but no agent may lose value, and under item-level upgrades no agent may lose a good without receiving a distinct weakly better replacement. Insertion reassigns at most $m(m-1)/2$ goods in total (Section~\ref{sec:goods-item}); we leave bounds on the number of reassignments under upgrade recourse to future work.

\paragraph{Recourse in online algorithms.}
Beyond fair division, online algorithms with \emph{recourse} may revise a limited number of earlier decisions. This model goes back to the dynamic Steiner tree problem~\citep{imase1991dynamic}; for online Steiner tree and minimum spanning tree, a constant competitive ratio can be maintained while changing only a few edges per arrival~\citep{megow2016power,gu2016power}. For online matching and related assignment problems, recourse counts the number of reassignments~\citep{gupta2014maintaining,bernstein2019online,angelopoulos2020online}. In scheduling and bin packing, the \emph{migration factor} bounds the total size of the jobs or items moved after an arrival relative to the size of the arriving one~\citep{sanders2004online,epstein2009robust,galvez2020symmetry}. \citet{boyar2022relaxing} instead restrict the type of change: an algorithm may accept an item after its arrival, or later reject an accepted item, but each such decision is final. In these problems, a solution is evaluated by a single objective, such as its total cost or the size of a matching. In fair division, each reassignment moves value between particular agents, so a limit on the number or total size of changes does not prevent one agent from losing most of her value. Upgrade recourse therefore restricts the type of change, as \citet{boyar2022relaxing} do, but separately for each agent.

\paragraph{Resource monotonicity and pairwise dominance.}
Bundle-level upgrades impose \emph{resource monotonicity} along the arrival sequence: adding resources must leave every agent weakly better off. This desideratum also appears in multi-resource cluster sharing~\citep{ghodsi2011drf}, although its divisible-resource model differs from ours. For indivisible goods, maximum weighted Nash welfare is resource-monotone for binary valuations~\citep{suksompong2022binary}, and weighted additive welfarist rules with concave functions remain so for matroid-rank valuations~\citep{suksompong2023weighted}; for two agents with binary valuations, resource monotonicity also appears in an axiomatic characterization of maximum Nash welfare~\citep{brandl2026binary}. Resource monotonicity can, however, conflict with other desiderata: in connected cake cutting, no Pareto-optimal proportional rule is resource-monotone~\citep{segalhalevi2018resource}. Item-level upgrades use the \emph{pairwise dominance} comparison of bundles studied in allocation by drafts~\citep{coreno2026draft}. Monotonicity of fixed picking sequences~\citep{chakraborty2021picking} provides one way to satisfy this requirement, whereas algorithms in our model may also depend on the arrival history. In particular, they need to preserve values only along the realized arrival order, rather than between every pair of nested sets of goods.

\paragraph{Offline MMS and simultaneous guarantees.}
The maximin share was introduced by \citet{budish2011combinatorial}; \citet{amanatidis2023fair} survey MMS, EF1, EFX, and related notions. Offline, MMS allocations always exist for two agents and for identical valuations, but they may fail to exist for three or more agents~\citep{kurokawa2018fair,feige2021tight}; the best known approximation factor exceeds $3/4$~\citep{akrami2024breaking}. Exact MMS allocations do exist for any number of agents in restricted valuation domains, such as one that generalizes weakly lexicographic and several few-valued valuations~\citep{teh2026lexicographic}. Approximate MMS can be obtained together with EF1 or EFX~\citep{amanatidis2018relaxations,akrami2025simultaneous}, and $1/2$-EFX allocations exist even for subadditive valuations~\citep{plaut2020almost}. Our upper bounds of $1/H_n$ under item-level upgrades and, for $n\ge3$, $10/11$ under bundle-level upgrades hold for identical valuations, for which exact MMS allocations exist offline. These bounds therefore come from the requirement that every update leave each agent at least as well off, not from the offline nonexistence of MMS allocations. For two agents, bundle-level upgrades attain exact MMS, as is possible offline. For $n\ge3$, our factor $(n-1)/(2n-3)$ is below the best known offline factor, but it holds after every arrival, together with EF1 and $1/2$-EFX, without decreasing any agent's value.

\paragraph{Extending partial fair allocations.}
\citet{bu2025welfare} study offline welfare maximization within fair allocations, without online arrivals or recourse. For additive valuations, their analysis gives procedures that complete partial EF1 allocations, or partial EFX allocations for two agents, without decreasing any agent's value. Applying these procedures after each arrival yields corresponding bundle-level guarantees in our model; Appendix~\ref{app:partial-allocations} explains this deduction. \citet{lim2026cost} study generalized-mean welfare within EFX allocations when the goods outnumber the agents by at most three, including the complexity of also requiring Pareto optimality. Our contribution is to formulate upgrade recourse and determine how its two forms affect MMS, simultaneous fairness and efficiency, and the value of predictions.

\paragraph{Ordinal allocation and scheduling.}
With identical additive valuations, a fixed picking sequence assigns sorted job ranks to machines, as in ordinal machine covering~\citep{he2002ordinal,tan2004ordinal,epstein2026ordinal}. In fair division, ordinal MMS guarantees were studied by \citet{amanatidis2016truthfulMMS}; \citet{halpern2021partial} obtain $1/(2H_n)$-MMS using deadlines for picking turns. Our guarantees approach $1/H_n$ within every inverse polynomial, and we show that fixed picking is optimal among all ordinal algorithms under item-level upgrades. Several small-$n$ achievable factors have antecedents in ordinal scheduling. Our matching upper bounds, however, apply to algorithms that use revealed cardinal values, and our characterization concerns one sequence that must work for every horizon. For identical valuations, our bundle-level algorithm uses the approximation and insertion-monotonicity properties of LPT~\citep{csirik1992lpt,galvez2020symmetry}.

\paragraph{Learning-augmented allocation.}
Predictions have been used for online Nash social welfare and more general allocation objectives~\citep{banerjee2022online,cohen2023general}. In indivisible fair division, \citet{neoh2026online} transfer offline share guarantees using value-frequency information and also analyze inaccurate predictions; other work studies proportionality up to one good~\citep{choo2026approximate}, EFX~\citep{melissourgos2026onlineefx}, envy-freeness under budget constraints~\citep{cohen2026budget}, and subsidy minimization in online house allocation~\citep{teh2026house}, whose offline counterpart is studied by \citet{choo2024envy}. Under budget constraints, predictions concern joint value-size types, since separate predictions of the value and size marginals do not recover strong fairness guarantees~\citep{cohen2026budget}; for MMS without budgets, our predictions need not link different agents' values to the same good. We combine prediction-based MMS guarantees with upgrades, bound loss directly through ranked-value error, and prove a tradeoff between guarantees for correct and arbitrary predictions. Unlike bandit fair division~\citep{procaccia2024honor}, we observe the arriving good's values before allocation: the uncertainty concerns future goods. Fairness concerns realized allocations, without averaging over arrivals or predictions.

\paragraph{Scheduling and predictions.}
The MMS benchmark has an exact scheduling interpretation. When all agents share an additive valuation $v$, each agent is a machine and each good $g$ is a job of size $v(g)$. The common MMS is then the offline optimum for \emph{machine covering}, which maximizes the minimum load~\citep{galvez2020symmetry}. Thus $\alpha$-MMS is exactly an $\alpha$-approximation to machine covering. Advance information about the jobs gives a \emph{semi-online} model~\citep{tan2007semi}; predictions that can be wrong give a learning-augmented model. In the accelerator example, procurement plans and workload profiles could predict the mix and values of incoming devices. We use predictions of each agent's ranked values, without requiring arrival order or item identities. Predictions supply information about future opportunities; recourse determines how earlier commitments may change. The heterogeneous case is related to the Santa Claus problem, but replaces its common minimum-utility objective by each agent's own MMS target~\citep{bansal2006santa,hajiaghayi2022santa}.
\paragraph{Other sequential collective decision-making work.}
Fairness across a sequence of collective decisions is also the focus of temporal (or perpetual) voting, in which a fixed group of voters makes a decision in each of several rounds and rules may take earlier decisions into account~\citep{lackner2020perpetual}; \citet{elkind2024temporal} propose a unified framework for these models. Earlier dynamic models choose an alternative in each round under evolving preferences or to maximize long-term Nash welfare~\citep{parkes2013dynamic,freeman2017dynamic}. In public decision making, several independent issues are decided at once~\citep{conitzer2017fair,skowron2022proportional,masarik2024generalised}, a model that generalizes the allocation of indivisible goods~\citep{conitzer2017fair}. Much of the temporal literature concerns proportionality. Perpetual voting rules have been analyzed axiomatically~\citep{lackner2020perpetual,lackner2023proportional}, and justified representation and its variants have been adapted to the temporal setting~\citep{bulteau2021justified,chandak2024proportional} and strengthened beyond EJR, for example to EJR+ and FJR, which remain satisfiable~\citep{phillips2026strengthening}. Verifying such axioms is harder than in multiwinner voting~\citep{elkind2025verifying}, although a strengthening of FJR can be both computed and verified in polynomial time~\citep{teh2026fjr}. Other work considers ranked preferences~\citep{hazon2026ranked}, bounds how often a voter is dissatisfied~\citep{kozachinskiy2025optimal}, and adapts the maximin share to binary perpetual decisions~\citep{gafni2024beyond}. For welfare, maximizing utilitarian welfare is easy, whereas optimizing egalitarian welfare is intractable even in restricted cases, including over public chores~\citep{elkind2024welfare,elkind2025chores}, and enforcing proportionality can incur a welfare loss that grows sublinearly with the number of voters or rounds~\citep{teh2026price}. Related models assign projects to time slots~\citep{elkind2022slot} or elect sequences of committees that bound how often candidates serve~\citep{bredereck2020successive}, how much consecutive committees differ~\citep{bredereck2022votes}, or how often each agent is satisfied~\citep{deltl2023algorithmics}; \citet{zech2024aversion} seek a second-stage winning committee that overlaps as much as possible with the first. As in our setting, it matters how much of the future is known. 
Among proportional rules for sequential decisions, Sequential Phragm\'en decides online, whereas the Method of Equal Shares and Proportional Approval Voting obtain stronger guarantees semi-online and offline, respectively~\citep{chandak2024proportional}; with ranked preferences, some guarantees require all future preferences, while others need only the number of rounds or no information about the future~\citep{hazon2026ranked}. Two features distinguish temporal voting from our model. Each decision there is a public outcome shared by all voters, whereas each good here belongs to one agent. Moreover, earlier decisions are typically not revised, or changes between consecutive outcomes are measured by their size, whereas upgrade recourse permits any change that leaves each agent at least as well off.

\section{Omitted Proofs in Section~\ref{sec:goods-item}}
\label{app:item-proofs}
\subsection{Proof of Lemma~\ref{lem:mms-deletion} (MMS Deletion)}
\label{app:mms-deletion}
An agent may take her first picking turn after several goods have been taken, so her MMS comparison must account for goods that are no longer available. The following lemma handles this, and it also lets us set aside singleton bundles.
\begin{lemma}
\label{lem:mms-deletion}
    Let $v_i$ be normalized and monotone. For every finite set $S$, every $g\in S$, and every $k\ge2$, $\MMS_i^{k-1}(S\setminus\{g\}) \ge \MMS_i^k(S)$.
\end{lemma}

\begin{proof}
Take a $k$-partition $(P_1,\dots,P_k)$ attaining $\MMS_i^k(S)$, and suppose $g\in P_k$. Remove $P_k$, and distribute the goods in $P_k\setminus\{g\}$ arbitrarily among $P_1,\dots,P_{k-1}$. Monotonicity implies that each resulting bundle has value at least $\MMS_i^k(S)$, proving the claim.
\end{proof}

\subsection{Pseudocode of Insertion and Proof of Lemma~\ref{lem:serial-insertion}}
\label{app:serial-insertion}
The following procedure maintains the picking allocation without repeating every choice from scratch. Position $k$ holds the good currently chosen at turn $k$.
\begin{algorithm}[H]
\caption{Insertion with a Fixed Picking Sequence}\label{alg:serial-insertion}
\begin{algorithmic}[1]
\Require An infinite picking sequence $\sigma$; initially no positions contain goods.
\For{each arriving good $g_t$}
  \State Observe $v_i(g_t)$ for all $i\in[n]$ and set $b\gets g_t$.
  \For{$k=1,\dots,t-1$}
    \State Let $o$ be the good at position $k$.
    \If{$b\succ^{\mathrm{tb}}_{\sigma_k}o$}
      \State Put $b$ at position $k$ and set $b\gets o$.
    \EndIf
  \EndFor
  \State Put $b$ at the new position $t$.
  \State For each $i$, let $A_i^{(t)}$ contain the goods at positions in $P_i^\sigma(t)$.
\EndFor
\end{algorithmic}
\end{algorithm}

Fix an arbitrary finite arrival sequence. We first prove that every update satisfies item-level upgrade recourse.
Consider the update from $\mathcal A^{(t-1)}$ to $\mathcal A^{(t)}$. During this update, each step leaves a position unchanged or exchanges its good with $b$. The final good $b$ is placed in the new position $t$. Hence the positions contain
distinct goods and partition $G^{(t)}$ at the end of the update.

Consider an old position $k<t$, and let $i=\sigma_k$. When position $k$ is
processed, it still contains the good it contained after arrival $t-1$; call this
good $o$. After position $k$ is processed, the good in that position is either $o$
itself or the unallocated good $b$ that was compared with $o$. In the second
case, the replacement occurs only because
$b \succ^{\mathrm{tb}}_i o$, which implies $v_i(b)\ge v_i(o)$. Therefore,
for each agent $i$, mapping the good formerly in each old position
$k<t$ with $\sigma_k=i$ to the good occupying that same position after the
update gives an injection from $A^{(t-1)}_i$ to $A^{(t)}_i$, and every old
good is mapped to a good that agent $i$ values at least as much. Hence
$A^{(t)}_i \succeq^{\mathrm{up}}_i A^{(t-1)}_i$ for every $i\in[n]$.

It remains to prove the picking claim. We use induction on $t$. The
case $t=0$ is immediate. Assume the claim holds after $t-1$ arrivals, and
consider the arrival of $g_t$. At the start of the update, the currently unallocated good is
$g_t$.

While processing old positions $1,\dots,t-1$, maintain the following
invariant. Immediately before old position $k<t$ is processed, positions
$1,\dots,k-1$ contain exactly the first $k-1$ choices of the picking
process on $G^{(t)}$; the unprocessed old positions $k,\dots,t-1$ contain
exactly the goods that were remaining before position $k$ in the picking
process on $G^{(t-1)}$; and the only additional remaining good is the
current unallocated good.

The invariant is immediate before position $1$ is processed. Suppose it holds
before old position $k<t$ is processed. Let $o$ be the old good in position $k$, and
let $i=\sigma_k$. By the induction hypothesis, $o$ is the
$\succ^{\mathrm{tb}}_i$-maximal good among the old goods remaining at this
point. Since the tie-breaking rule is consistent under restriction, the
$\succ^{\mathrm{tb}}_i$-maximal good among all remaining goods in
$G^{(t)}$ is either $o$ or the current unallocated good $b$. Algorithm~\ref{alg:serial-insertion} puts exactly
this good in position $k$, and leaves the other one as the currently unallocated good. Thus the
invariant is preserved.

After positions $1,\dots,t-1$ have been processed, only the currently unallocated good remains, and
Algorithm~\ref{alg:serial-insertion} places it in position $t$. Therefore, after arrival $t$, the maintained
allocation is exactly the picking allocation induced by
$(\sigma_1,\dots,\sigma_t)$ on $G^{(t)}$.

\subsection{Proof of Theorem~\ref{thm:balanced-serial-insertion}}

By Lemma~\ref{lem:serial-insertion}, the algorithm satisfies item-level upgrade recourse, and the final allocation $\mathcal A$ is the picking allocation induced by $(\sigma_1,\dots,\sigma_m)$ on $G$. Fix distinct agents $i,j\in[n]$. Let
$g^i_1,\dots,g^i_a$ be the goods in $A_i$, listed in the order of
agent $i$'s turns in $(\sigma_1,\dots,\sigma_m)$, and let
$g^j_1,\dots,g^j_b$ be the goods in $A_j$, listed in the order of
agent $j$'s turns. If $b=0$, then $A_j=\varnothing$, so there is
nothing to prove. If $b=1$, then removing $g^j_1$ leaves the empty bundle,
and $v_i(A_i)\ge0=v_i(A_j\setminus\{g^j_1\})$.
Assume therefore that $b\ge 2$.

Since $\sigma$ is balanced, $b\le a+1$. For every $\ell=1,\dots,b-1$,
agent $i$'s $\ell$th turn occurs before agent $j$'s $(\ell+1)$st turn.
Otherwise, immediately after agent $j$'s $(\ell+1)$st turn, agent $j$
would have at least $\ell+1$ turns while agent $i$ would have at most
$\ell-1$ turns, contradicting balancedness.

Thus, for every $\ell=1,\dots,b-1$, when agent $i$ picked $g^i_\ell$, the
good $g^j_{\ell+1}$ was still available. Since picking chooses a
$\succ^{\mathrm{tb}}_i$-maximal remaining good,
$v_i(g^i_\ell)\ge v_i(g^j_{\ell+1})$.
Summing over $\ell=1,\dots,b-1$, and using $b\le a+1$, additivity, and
nonnegativity, gives
\[
v_i(A_i)
=
\sum_{\ell=1}^{a} v_i(g^i_\ell)
\ge
\sum_{\ell=1}^{b-1} v_i(g^i_\ell)
\ge
\sum_{\ell=1}^{b-1} v_i(g^j_{\ell+1})
=
v_i(A_j\setminus\{g^j_1\}).
\]
Therefore EF1 holds for the ordered pair $(i,j)$. Since $i$ and $j$ were arbitrary, $\mathcal A$ is EF1.

The standard implication from EF1 to $1/n$-MMS for additive valuations~\citep{amanatidis2018relaxations} gives the remaining guarantee.

\subsection{Proof of Theorem~\ref{thm:tight-two-agent-mms} (Exact Factor for Two Agents)}
\label{app:two-agent-tight}
\begin{theorem}\label{thm:tight-two-agent-mms}
For $n=2$, Algorithm~\ref{alg:serial-insertion} run with $\sigma=(1,2,2,1,2,2,\dots)$ satisfies item-level upgrade recourse, and its final allocation is $2/3$-MMS. Moreover, the factor $2/3$ is best possible for deterministic algorithms with item-level upgrade recourse, even for identical additive valuations with singleton values in $[0,1]$.
\end{theorem}

By Lemma~\ref{lem:serial-insertion}, the algorithm has item-level upgrade recourse and the final allocation $\mathcal A$ is the picking allocation induced by $(\sigma_1,\dots,\sigma_m)$ on $G$. It is therefore enough to prove the $2/3$-MMS bound for an arbitrary finite set of $m$ arrived goods under this picking allocation.

Recall the ranked values $x_k^i(G)$ from the preliminaries, and set $x_k^i(G):=0$ for $k>m$. For each agent $i\in[2]$, we use the two upper bounds $\MMS_i(G)\le v_i(G)/2$ and $\MMS_i(G)\le v_i(G)-x_1^i(G)$.
The first bound holds because no two-bundle partition can make both
bundles worth more than half the total value. The second bound holds
because, in any two-bundle partition, the bundle not containing the
highest-valued good has value at most $v_i(G)-x_1^i(G)$, and the minimum of the
two bundle values is at most the value of that bundle.

We first consider agent $1$. Her turns are $3r-2$ for $r\ge1$, so the rank-value bound in Section~\ref{sec:preliminaries} gives $v_1(A_1)
\ge
\sum_{r\ge 1} x_{3r-2}^1(G)$.
Since the sequence $(x_k^1(G))_{k\ge1}$ is nonincreasing, $x_{3r-2}^1(G)\ge\frac13(x_{3r-2}^1(G)+x_{3r-1}^1(G)+x_{3r}^1(G))$ for every $r\ge1$.
Summing over all $r$ gives
\[
v_1(A_1)
\ge
\sum_{r\ge1} x_{3r-2}^1(G)
\ge
\frac13\sum_{r\ge1}
\left(x_{3r-2}^1(G)+x_{3r-1}^1(G)+x_{3r}^1(G)\right)
=
\frac{v_1(G)}{3}.
\]
Using $\MMS_1(G)\le v_1(G)/2$, we obtain $v_1(A_1)\ge v_1(G)/3\ge\frac23\MMS_1(G)$.

We now consider agent $2$. Her turns are $3r-1$ and $3r$ for $r\ge1$, so the rank-value bound gives $v_2(A_2)\ge\sum_{r\ge1}(x_{3r-1}^2(G)+x_{3r}^2(G))$. For every $r\ge1$, monotonicity of the ranked values gives $x_{3r-1}^2(G)+x_{3r}^2(G)\ge\frac23(x_{3r-1}^2(G)+x_{3r}^2(G)+x_{3r+1}^2(G))$.
The triples on the right partition all ranks except rank $1$. Consequently,
\[
    v_2(A_2)
    \ge
    \sum_{r\ge1}\left(x_{3r-1}^2(G)+x_{3r}^2(G)\right)
    \ge
    \frac23(v_2(G)-x_1^2(G))
    \ge
    \frac23\MMS_2(G).
\]

We have shown that both agents receive value at least two thirds of their
maximin share. Thus the algorithm is $2/3$-MMS.

Finally, the best-possible part follows from Theorem~\ref{thm:mms-impossibility}: for
$n=2$, the upper bound there is $1/H_2=2/3$.
\subsection{Deadline Sequences and Proof of Theorem~\ref{thm:deadline-serial-insertion-mms}}
\label{app:deadline-analysis}
For $\alpha\in(0,1)$, agent $i$ has turn requests $T_{i,j}$, $j=0,1,\ldots$, where $T_{i,j}$ asks for her $(j+1)$st turn by the deadline $i+\lfloor j(n-i+1)/\alpha\rfloor$.
Sort all requests by deadline, breaking ties by $(i,j)$, and denote the resulting sequence by $\tau(\alpha)$. The first $n$ turns belong to agents $1,\ldots,n$ in that order. Maintaining the next request of each agent takes $O(\log n)$ time per turn. We first check when all deadlines can be met, then derive the MMS guarantee. This refines the picking-time argument of \citet{amanatidis2016truthfulMMS} and the deadline construction of \citet{halpern2021partial}.
\begin{lemma}
\label{lem:deadline-feasibility}
Every turn request in $\tau(\alpha)$ appears no later than its deadline if either $\alpha H_n+\alpha^2/(1-\alpha)\le1$, or $n\ge256$ and $\alpha H_n+54\alpha^2/n\le1$.
\end{lemma}

\begin{proof}
Let $D(t)$ be the number of turn requests with deadline at most $t$. It is enough to prove $D(t)\le t$ for every positive integer $t$, because a request with deadline $t$ then has position at most $D(t)\le t$.

For $t\le n$, the only requests due by time $t$ are $T_{1,0},\dots,T_{t,0}$, so $D(t)=t$. Now let $t=n+x$ for an integer $x>0$, and re-index agents by $r=n-i+1$. The request $T_{i,j}$ is due by time $n+x$ exactly when $n-r+1+ \lfloor jr/\alpha\rfloor\le n+x$, which is equivalent to $j<\alpha(1+x/r)$. Hence $D(n+x)-n=\sum_{r=1}^n(\lceil\alpha+\alpha x/r\rceil-1)$.
For the first condition, a summand is positive only when $r<\alpha x/(1-\alpha)$.
Using $\lceil y\rceil-1\le y$ for these summands, we obtain
\[
    D(n+x)-n
     \le
    \sum_{r\in[n]:\,r<\alpha x/(1-\alpha)}\left(\alpha+\frac{\alpha x}{r}\right) \le
    \frac{\alpha^2x}{1-\alpha}+\alpha xH_n
    \le x,
\]
where the last inequality is the first condition in the lemma.

For the second condition, $\alpha<1/H_n<1/6$, since $H_{256}>6$. Appendix~\ref{app:deadline-rounding} proves the counting estimate
\begin{equation}\label{eq:deadline-rounding}
    \sum_{r=1}^n\left\lfloor\alpha+\frac{\alpha x}{r}\right\rfloor
    \le\alpha xH_n
    \quad\left(0\le\alpha x\le\frac{n^2}{54}\right).
\end{equation}
Since $\lceil y\rceil-1\le\lfloor y\rfloor$, this gives $D(n+x)-n\le x$ whenever $\alpha x\le n^2/54$. For $\alpha x\ge n^2/54$, use $\lceil y\rceil-1<y$ to obtain $D(n+x)-n<\alpha xH_n+\alpha n\le x(\alpha H_n+54\alpha^2/n)\le x$.
Thus $D(n+x)\le n+x$ under either condition.
\end{proof}

Meeting the deadlines gives each agent sufficiently many turns in every prefix starting at her first turn. The next lemma converts this count into a value guarantee: for nonincreasing values, selecting a given fraction of every prefix captures at least that fraction of the total value, even when most value is concentrated in the first few goods.

\begin{lemma}
\label{lem:prefix-count}
Let $z_1\ge\cdots\ge z_s\ge0$, let $P\subseteq[s]$, and let $\lambda\in[0,1]$. If $|P\cap[k]|\ge\lambda k$ for every $k \in [s]$, then $\sum_{k\in P}z_k\ge\lambda\sum_{k=1}^s z_k$.
\end{lemma}

\begin{proof}
Set $z_{s+1}:=0$ and $N(k):=|P\cap[k]|$. Summation by parts, followed by $N(k)\ge\lambda k$ and $z_k-z_{k+1}\ge0$, gives
\[
    \sum_{k\in P}z_k
    =
    \sum_{k=1}^s N(k)(z_k-z_{k+1})
    \ge
    \lambda\sum_{k=1}^s k(z_k-z_{k+1})
    =
    \lambda\sum_{k=1}^s z_k. \qedhere
\]
\end{proof}

\begin{theorem}
\label{thm:deadline-serial-insertion-mms}
Suppose every request in $\tau(\alpha)$ appears no later than its deadline. Then Algorithm~\ref{alg:serial-insertion} run with $\tau(\alpha)$ satisfies item-level upgrade recourse, and its final allocation is $\alpha$-MMS.
\end{theorem}

\begin{proof}
Item-level upgrade recourse follows from Lemma~\ref{lem:serial-insertion}. Fix an agent $i$. If $m<i$, then $m<n$, so $\MMS_i(G)=0$ and the claim is immediate. Assume $m\ge i$.

Before agent $i$'s first turn, exactly $i-1$ goods have been selected. Let $S_i$ be the set of goods remaining at that point. Repeatedly applying Lemma~\ref{lem:mms-deletion} to the $i-1$ selected goods gives $\MMS_i(G)\le\MMS_i^{n-i+1}(S_i)\le v_i(S_i)/(n-i+1)$.

Let $s:=|S_i|$, and list the singleton values in $S_i$ as $z_1\ge z_2\ge\cdots\ge z_s$.
Let $P\subseteq[s]$ be the set of relative positions $k$ such that turn $i-1+k$ belongs to agent $i$. For every $k\in[s]$, set $q:=\lceil\alpha k/(n-i+1)\rceil-1$.
Since $q<\alpha k/(n-i+1)$, the deadline of $T_{i,q}$ is at most $i+k-1$. Hence agent $i$ owns at least $q+1$ of the turns $i,\dots,i+k-1$, and therefore, $|P\cap[k]|\ge\lceil\alpha k/(n-i+1)\rceil\ge\alpha k/(n-i+1)$.

At relative position $k$, at most $k-1$ goods of $S_i$ have already been selected, so a turn of agent $i$ at that position gives her a good worth at least $z_k$. Lemma~\ref{lem:prefix-count}, with $\lambda=\alpha/(n-i+1)$, now gives $v_i(A_i)\ge\sum_{k\in P}z_k\ge\alpha v_i(S_i)/(n-i+1)\ge\alpha\MMS_i(G)$.
Since $i$ was arbitrary, the final allocation is $\alpha$-MMS.
\end{proof}
\subsection{Proof of Corollary~\ref{cor:general-deadline-mms} (Explicit Guarantee for Every \texorpdfstring{$n$}{n})}
\begin{corollary}
\label{cor:general-deadline-mms}
For every $n\ge2$, let
\[
\alpha_n^*:=
\begin{cases}
\dfrac{H_n+1-\sqrt{H_n^2-2H_n+5}}{2(H_n-1)},&2\le n<256,\\[6pt]
\dfrac{2}{H_n+\sqrt{H_n^2+216/n}},&n\ge256.
\end{cases}
\]
Algorithm~\ref{alg:serial-insertion} run with $\tau(\alpha_n^*)$ satisfies item-level upgrade recourse and returns an $\alpha_n^*$-MMS allocation. Moreover, $1/(H_n+1)\le\alpha_n^*<1/H_n$, and, for $n\ge256$, $1/H_n-\alpha_n^*=54(\alpha_n^*)^2/(nH_n)\le54/(nH_n^3)$.
\end{corollary}
\begin{proof}
In each case, $\alpha_n^*$ is the unique solution in $(0,1)$ to the corresponding equality in Lemma~\ref{lem:deadline-feasibility}. The lemma and Theorem~\ref{thm:deadline-serial-insertion-mms} therefore give the MMS guarantee.

The left-hand sides of both conditions are strictly increasing in $\alpha\in(0,1)$. At $\alpha=1/(H_n+1)$, the first equals $H_n/(H_n+1)+1/(H_n(H_n+1))\le1$.
For $n\ge256$, the second left-hand side is smaller than the first for every $\alpha\in(0,1)$, because $54/n<1<1/(1-\alpha)$. Hence $\alpha_n^*\ge1/(H_n+1)$ in both cases, and the second formula strictly improves on the first when $n\ge256$. Each defining equality also implies $\alpha_n^*H_n<1$. Finally, for $n\ge256$, rearranging $\alpha_n^*H_n+54(\alpha_n^*)^2/n=1$ gives the stated expression for $1/H_n-\alpha_n^*$, and $\alpha_n^*<1/H_n$ gives its upper bound.
\end{proof}

\subsection{Proof of the Counting Estimate\texorpdfstring{~\eqref{eq:deadline-rounding}}{}}
\label{app:deadline-rounding}

Fix $n\ge256$ and $0<\alpha\le1/6$. We prove \eqref{eq:deadline-rounding} for every real $x\ge0$ with $\alpha x\le n^2/54$. The estimate uses the total amount removed by rounding down. We first lower-bound a corresponding integral, and then compare it with the sum over agents.

Write $\{y\}:=y-\lfloor y\rfloor$ for the fractional part of $y$. For $c\ge1/2$, define
\[
    I_\alpha(c):=c\int_c^\infty\frac{\{u+\alpha\}}{u^2}\,du.
\]
We show that $I_\alpha(c)\ge10/27$. The integral converges since $0\le\{u+\alpha\}<1$. Between consecutive points $j-\alpha$, differentiation gives $I_\alpha''(c)=-1/c<0$. The function is continuous, so it suffices to check $c=1/2$ and the points $b=j-\alpha$ for positive integers $j$.

At such a point,
\[
    I_\alpha(b)
    =b\sum_{k=0}^\infty\int_0^1\frac{s}{(b+k+s)^2}\,ds
    \ge\frac b2\sum_{k=0}^\infty\frac1{(b+k+2/3)^2}.
\]
The inequality is Jensen's inequality for the convex function $s\mapsto(b+k+s)^{-2}$ under the probability density $2s$ on $[0,1]$, whose mean is $2/3$. For any $a>0$, summing the trapezoidal inequalities for the convex function $t\mapsto(a+t)^{-2}$ gives
\[
    \sum_{k=0}^\infty\frac1{(a+k)^2}
    \ge\int_0^\infty\frac{dt}{(a+t)^2}+\frac1{2a^2}
    =\frac1a+\frac1{2a^2}.
\]
Consequently,
\[
    I_\alpha(b)
    \ge\frac{b}{2(b+2/3)}+\frac{b}{4(b+2/3)^2}
    \ge\frac{10}{27}.
\]
For the last inequality, the middle expression has derivative $(b/12+7/18)/(b+2/3)^3>0$, and its value at $b=5/6$ is $10/27$; here $b\ge1-\alpha\ge5/6$.

For the remaining endpoint, put $b=1-\alpha\in[5/6,1]$. On $[1/2,5/6)$, $\{u+\alpha\}=u+\alpha$. Discarding a nonnegative part of the integral, we obtain
\[
    I_\alpha(1/2)
    \ge\frac12\int_{1/2}^{5/6}\frac{du}{u}
       +\frac{I_\alpha(b)}{2b}
    \ge\frac15+\frac5{27}
    >\frac{10}{27}.
\]
The integral bound uses $1/u\ge6/5$ on $[1/2,5/6]$. Concavity between consecutive checked points now proves $I_\alpha(c)\ge10/27$ for all $c\ge1/2$.

We next prove \eqref{eq:deadline-rounding}. If $\alpha x\le n/2$, set $k=\lfloor\alpha x/(1-\alpha)\rfloor$. All summands with $r>k$ vanish, and $k\le3n/5$. The result is immediate when $k=0$. Otherwise,
\[
    \sum_{r=1}^n\left\lfloor\alpha+\frac{\alpha x}{r}\right\rfloor
    \le\alpha xH_k+\alpha k
    \le\alpha x\left(H_k+\frac15\right)
    \le\alpha xH_n.
\]
The second inequality uses $\alpha k\le\alpha^2x/(1-\alpha)\le\alpha x/5$, and the last follows from $H_n-H_k\ge(n-k)/n\ge2/5$.

It remains to consider $n/2<\alpha x\le n^2/54$. Let $R=\lceil\sqrt{\alpha x/2}\rceil$, so $1\le R\le n$. Since $t\mapsto\lfloor\alpha+\alpha x/t\rfloor$ is nonincreasing,
\[
    \sum_{r=1}^n\left\lfloor\alpha+\frac{\alpha x}{r}\right\rfloor
    \le\alpha xH_R+\alpha R+
       \int_R^n\left\lfloor\alpha+\frac{\alpha x}{t}\right\rfloor\,dt.
\]
The change of variable $u=\alpha x/t$ gives $\int_0^n\{\alpha+\alpha x/t\}\,dt=nI_\alpha(\alpha x/n)\ge10n/27$, since $\alpha x/n>1/2$.
The same integral over $(0,R)$ is at most $R$. Expanding the floor in the preceding bound therefore gives
\[
    \sum_{r=1}^n\left\lfloor\alpha+\frac{\alpha x}{r}\right\rfloor
    \le\alpha x(H_R+\log(n/R))
       +\alpha n-\frac{10n}{27}+R.
\]
Convexity of $1/t$ and the trapezoidal inequalities give $H_R+\log(n/R)-H_n\le1/(2R)-1/(2n)$.
Using $\alpha\le1/6$ and omitting the nonpositive term $-\alpha x/(2n)$, we obtain
\begin{align*}
    \sum_{r=1}^n\left\lfloor\alpha+\frac{\alpha x}{r}\right\rfloor
    &\le\alpha xH_n-\frac{11n}{54}+R+\frac{\alpha x}{2R}\\
    &\le\alpha xH_n-\frac{11n}{54}+\sqrt{2\alpha x}+1\\
    &\le\alpha xH_n-\frac n{108}+1
    \le\alpha xH_n.
\end{align*}
Here $R\ge\sqrt{\alpha x/2}$ and $R<\sqrt{\alpha x/2}+1$ imply the second inequality. The third uses $\sqrt{2\alpha x}\le n/\sqrt{27}\le7n/36$, and the last uses $n\ge256$. This completes the proof of \eqref{eq:deadline-rounding}.
\subsection{Proof of Theorem~\ref{thm:ordinal-optimality} (Optimality of Fixed Picking Sequences)}
\label{app:ordinal-characterization}
\begin{theorem}\label{thm:ordinal-optimality}
For every deterministic ordinal algorithm with item-level upgrade recourse, insertion with some fixed picking sequence guarantees at least the same MMS factor. For every $\varepsilon>0$, a finite mixed-integer linear optimization over the first $n\lceil1/\varepsilon\rceil$ turns produces a sequence attaining at least $\alpha_n^{\mathrm{ord}}-\varepsilon$, whose turns repeat periodically after a finite initial segment.
\end{theorem}

The theorem reduces the search to fixed picking sequences. To approximate their optimal MMS factor, we compare an agent's early picks with weighted averages of the parts of any MMS partition. A bound on each part's weight determines the fraction of later turns the agent needs; these fractions must sum to at most one. This gives a finite linear optimization with integer turn assignments. We show below that its values decrease to $\alpha_n^{\mathrm{ord}}$, with error at most $1/\lfloor m/n\rfloor$ when the optimization covers the first $m$ turns. The computation need not be polynomial-time.

\paragraph{A partition comparison.}
Fix $r\ge1$ and values $z_1\ge\cdots\ge z_m\ge0$. For this proof only, allow an additional total value $T\ge0$ to be divided among the $r$ parts of a partition, and let
\[
 \widetilde{\MMS}^{\,r}(z,T)
 =\max_{(P_1,\ldots,P_r)\in\Pi_r([m])}
   \max_{\substack{u_j\ge0\\\sum_j u_j=T}}
       \min_{j\in[r]}\left(\sum_{k\in P_j}z_k+u_j\right).
\]
This is an upper bound on MMS when the additional goods are indivisible and have total value $T$. It is also the limit obtained by adding increasingly many equal goods. Indeed, replacing the divisible value by $M$ goods of value $T/M$ changes the optimum by at most $T/M$: round down the number assigned to each part and distribute any remaining goods arbitrarily. At $T=0$, the expression is the usual $r$-maximin share of the first $m$ goods. The allocation model itself remains indivisible.

Let $C(k)\ge0$ be a proposed lower bound on the number of selected ranks through position $k$, and set $z_{m+1}=0$. Summation by parts makes $\sum_{k=1}^m C(k)(z_k-z_{k+1})$ a lower bound on their value. The following lemma characterizes when those early ranks and a fraction $\lambda$ of the additional value suffice.
\begin{lemma}\label{lem:partition-weights}
For $a,\lambda\ge0$, the inequality $\sum_{k=1}^m C(k)(z_k-z_{k+1})+\lambda T\ge a\widetilde{\MMS}^{\,r}(z,T)$ holds for every nonincreasing nonnegative $z$ and every $T\ge0$ if and only if, for every $P\in\Pi_r([m])$, there are weights $w_1,\ldots,w_r$ satisfying $\sum_{j=1}^r w_j=a$, $0\le w_j\le\lambda$, and $\sum_{j=1}^r |P_j\cap[k]|w_j\le C(k)$ for every $k\in[m]$.
The same weights must satisfy the inequalities for all $k$.
\end{lemma}
\begin{proof}
Fix $P$ and put $d_k=z_k-z_{k+1}\ge0$. Part $P_j$ has value $\sum_k|P_j\cap[k]|d_k$ before receiving additional value. The minimum left-hand side of the value inequality in Lemma~\ref{lem:partition-weights} over choices making every part worth at least one is
\[
 \min_{d,u\ge0}\left\{
 \sum_{k=1}^m C(k)d_k+\lambda\sum_{j=1}^r u_j:
 \sum_{k=1}^m |P_j\cap[k]|d_k+u_j\ge1\quad\text{for }j\in[r] 
 \right\}.
\]
Its linear-programming dual maximizes $\sum_j w_j$ subject to the last two weight conditions in Lemma~\ref{lem:partition-weights}. Strong duality applies: the primal is feasible and its objective is nonnegative. Thus its value is at least $a$ exactly when the dual has a feasible vector of sum at least $a$, which can be scaled down to sum $a$. Scaling the part values proves the assertion for this partition. Taking all partitions proves the lemma; when MMS is zero, the value inequality is automatic.
\end{proof}

The upper bound on each weight expresses the needed fraction of later value. The other inequalities require the agent's number of early picks to be at least the weighted count of the same parts at every rank. This is why checking only the total number of early picks would not suffice.

\paragraph{The finite optimization.}
For $m\ge n$, define $\beta_{n,m}$ as the maximum $a\in[0,1]$ over a finite picking sequence $\sigma\in[n]^m$ and numbers $\lambda_i\ge0$ with $\sum_i\lambda_i\le1$, subject to the following condition: for every $i\in[n]$ and $P\in\Pi_n([m])$, there are weights $w_1,\ldots,w_n$ satisfying $\sum_{j=1}^n w_j=a$, $0\le w_j\le\lambda_i$, and $\sum_{j=1}^n |P_j\cap[k]|w_j\le N_i^\sigma(k)$ for every $k\in[m]$.
The weights may depend on $i$ and $P$. Binary variables specify who owns each turn; all remaining constraints are linear. The sum of the $\lambda_i$ is at most one because one good is assigned at each later turn. This optimization uses only $n,m$, not the values of any particular instance.

\begin{theorem}\label{thm:finite-ordinal-formula}
For every $n\ge2$ and $m\ge n$, $\beta_{n,m}$ is attained, $\beta_{n,n}=1/H_n$, and $\beta_{n,m+1}\le\beta_{n,m}$. Moreover, $\max\{0,\beta_{n,m}-1/\lfloor m/n\rfloor\}\le\alpha_n^{\mathrm{ord}}\le\beta_{n,m}$.
This lower bound can be attained by insertion with a sequence that repeats periodically after its first $m$ turns.
Consequently, $\alpha_n^{\mathrm{ord}}=\inf_{m\ge n}\beta_{n,m}$.
\end{theorem}
\begin{proof}
For a fixed $\sigma$, the feasible set is a compact polytope: $a,\lambda_i,w_j\in[0,1]$. There are finitely many such sequences, so the maximum is attained. Restricting a feasible length-$(m+1)$ sequence to $m$ turns preserves feasibility: extend each partition of $[m]$ by putting rank $m+1$ in any part and use its weights. This proves monotonicity.

We first compare ordinal algorithms with fixed sequences. On any instance in which all agents strictly rank goods in arrival order, combine the upgrade injections after arrival $t$ into one injection $\phi:[t-1]\to[t]$. Strict ranking implies $\phi(k)\le k$. Inductively, $\phi(k)=k$ for every $k<t$, so no old good changes owners. Every ordinal algorithm sees the same ranking history on these instances and therefore follows a fixed infinite assignment sequence. Its guarantee must hold for every strictly decreasing positive value list; continuity extends it to all nonincreasing nonnegative lists. The rank-value comparison in the preliminaries then shows that insertion with the same sequence guarantees at least as much on arbitrary instances.

It therefore suffices to study a fixed sequence with guarantee $a$. We show that $a\le\beta_{n,m}$ for every such sequence, which proves the upper inequality after taking the supremum.

Set $\lambda_i=\liminf_{t\to\infty}N_i^\sigma(t)/t$. These rates sum to at most one. Fix prefix values $z_1\ge\cdots\ge z_m>0$, and append $M$ equal goods of total value $T$. For large $M$, they follow the prefix in the ranked order. Along a subsequence of values of $M$ for which agent $i$'s fraction of the appended turns tends to $\lambda_i$, her selected value tends to $\sum_{k=1}^m N_i^\sigma(k)(z_k-z_{k+1})+\lambda_iT$.
The MMS tends to $\widetilde{\MMS}^{\,n}(z,T)$. Hence the value inequality in Lemma~\ref{lem:partition-weights} holds; zero prefix values follow by continuity. Lemma~\ref{lem:partition-weights} then provides the weights required in the definition of $\beta_{n,m}$. Thus $a\le\beta_{n,m}$, and hence $\alpha_n^{\mathrm{ord}}\le\beta_{n,m}$.

For the constructive direction, let a maximizing solution have value $a>0$, and write $c_i=N_i^\sigma(m)$ and $c=\min_i c_i$. Taking all $m$ prefix values equal to one and $T=0$ in Lemma~\ref{lem:partition-weights} gives us $c_i\ge a\lfloor m/n\rfloor$.
In particular $c\ge1$. Put $\theta=1-1/c$. Append a fixed sequence in which agent $i$ receives at least $\lfloor\theta\lambda_i s\rfloor$ of the first $s$ additional turns. Such a sequence is obtained by giving her $j$th additional request the deadline $\lceil j/(\theta\lambda_i)\rceil$ and sorting all requests by deadline, omitting agents with $\lambda_i=0$ and filling unused turns arbitrarily. By time $s$, the number of requests due is at most $\sum_i\lfloor\theta\lambda_i s\rfloor\le s$.

For every $s\ge0$, the total count of agent $i$ is at least $c_i+\lfloor\theta\lambda_i s\rfloor
 \ge \theta c_i+\theta\lambda_i s$,
where for $s>0$ we use $(1-\theta)c_i\ge1$, and $s=0$ is immediate. Before turn $m$ her counts are at least $\theta N_i^\sigma(k)$. Summation by parts therefore bounds her selected value below by $\theta$ times the left-hand side of the value inequality in Lemma~\ref{lem:partition-weights}, with $T$ the total value after rank $m$. That benchmark is at least her MMS. Horizons below $m$ are covered by padding with zeros. Insertion consequently guarantees $\theta a$, losing at most $a-\theta a=a/c\le1/\lfloor m/n\rfloor$.
A maximizing solution can be chosen with rational rates. Include an extra rate equal to $1-\sum_i\theta\lambda_i$, assigning its requests to any agent. With a common denominator, the resulting deadline order repeats periodically. The case $\theta=0$ is immediate. This proves the bounds in Theorem~\ref{thm:finite-ordinal-formula} and the asserted limit.

It remains to verify $\beta_{n,n}=1/H_n$. A positive feasible factor requires each agent to appear once in the first $n$ turns, by the bound $c_i\ge a\lfloor m/n\rfloor$ above with $m=n$. Relabel them in that order. For agent $i$, take a partition whose first $i-1$ ranks are singleton parts. The zero count before her first turn forces the weights of these parts to vanish; hence $a\le(n-i+1)\lambda_i$. Summing gives $aH_n\le1$. Conversely, set $a=1/H_n$ and $\lambda_i=a/(n-i+1)$. For any partition, at least $n-i+1$ parts avoid the first $i-1$ ranks. Put weight $a/(n-i+1)$ on that many parts and zero on the others. Before rank $i$ the weighted count is zero; afterwards it is at most $a\le1$. These weights prove feasibility.
\end{proof}

Taking $m=n\lceil1/\varepsilon\rceil$ proves the computational assertion of Theorem~\ref{thm:ordinal-optimality}. The formula does not assume that an optimal infinite sequence is periodic, and it does not identify the optimum for algorithms using cardinal values. Its precise implication is that adapting to the observed rankings cannot improve the MMS guarantee, and a fixed periodic continuation suffices for any desired additive approximation.

\subsection{Proof of Theorem~\ref{thm:finite-exact} (Exact Factors for \texorpdfstring{$n\le33$}{n <= 33})}
\label{app:exact-factors}
Theorem~\ref{thm:tight-two-agent-mms} treats two agents. We first give a short proof for three agents. For the other cases, we use the same two arguments throughout: weighted MMS partitions determine sufficient picking turns, and, for the upper bounds, later goods worth less than all earlier goods prevent any earlier good from being replaced, so the earlier assignments limit the attainable MMS factor. Tables~\ref{tab:exact-factors}--\ref{tab:two-prefix-upper} specify all parameters. The remaining arithmetic is finite and is reproduced by the accompanying program \texttt{verify\_exact\_factors.py}, listed in Appendix~\ref{app:program}; the reductions below justify the guarantees at every horizon.

\begin{table}[t]
\caption{Optimal MMS factors under item-level upgrades. In every entry, $\alpha_n=\alpha_n^{\mathrm{ord}}$: insertion with a fixed sequence attains the factor, and algorithms using cardinal values cannot improve it.}
\label{tab:exact-factors}
\centering\small
\setlength{\tabcolsep}{9pt}\renewcommand{\arraystretch}{1.65}
\begin{tabular}{@{}cr cr cr cr@{}}
\toprule
$n$ & $\alpha_n$ & $n$ & $\alpha_n$ & $n$ & $\alpha_n$ & $n$ & $\alpha_n$ \\
\midrule
2 & $\dfrac{2}{3}$ & 10 & $\dfrac{439}{1286}$ & 18 & $\dfrac{102255661}{357395274}$ & 26 & $\dfrac{735}{2833}$ \\[4pt]
3 & $\dfrac{1}{2}$ & 11 & $\dfrac{540}{1631}$ & 19 & $\dfrac{4095}{14528}$ & 27 & $\dfrac{493920}{1922069}$ \\[4pt]
4 & $\dfrac{10}{21}$ & 12 & $\dfrac{20159}{62560}$ & 20 & $\dfrac{16380}{58931}$ & 28 & $\dfrac{493920}{1939709}$ \\[4pt]
5 & $\dfrac{24}{55}$ & 13 & $\dfrac{72}{229}$ & 21 & $\dfrac{16380}{59711}$ & 29 & $\dfrac{1024380}{4058239}$ \\[4pt]
6 & $\dfrac{15}{37}$ & 14 & $\dfrac{504}{1639}$ & 22 & $\dfrac{7371}{27205}$ & 30 & $\dfrac{204876}{818477}$ \\[4pt]
7 & $\dfrac{42}{109}$ & 15 & $\dfrac{5171009}{17158995}$ & 23 & $\dfrac{27440}{102469}$ & 31 & $\dfrac{493920}{1989137}$ \\[4pt]
8 & $\dfrac{18}{49}$ & 16 & $\dfrac{812054744}{2745351561}$ & 24 & $\dfrac{82320}{310837}$ & 32 & $\dfrac{123480}{501143}$ \\[4pt]
9 & $\dfrac{6}{17}$ & 17 & $\dfrac{145645483}{500957470}$ & 25 & $\dfrac{493920}{1884779}$ & 33 & $\dfrac{8350020}{34141547}$ \\[4pt]
\bottomrule
\end{tabular}
\end{table}

\subsubsection{Proof for Three Agents}
\label{app:three-agent-tight}
\begin{proof}[Proof of Theorem~\ref{thm:finite-exact} for $n=3$]
For the lower bound, use the periodic sequence $(1,2,3,3,2,3)^\infty$. Consider the goods remaining before agent $i$'s first turn, and put $r=4-i$. MMS deletion bounds her original MMS by the $r$-maximin share of these goods, and hence by their total value divided by $r$. The first agent selects relative ranks $1,7,13,\ldots$, obtaining at least one sixth of their total value. The second selects relative ranks $1,4,7,\ldots$, obtaining at least one third. The third selects relative ranks $1,2,4$ in each successive block of six. For nonincreasing values in a block, $z_1+z_2+z_4\ge z_3+z_5+z_6$.
Padding a final block with zeros proves a one-half fraction for her. The resulting MMS factors are at least $1/2,2/3,1/2$, respectively. Insertion gives the required recourse.

For the upper bound, suppose an algorithm guarantees $a>1/2$, and choose $1/(2a)<q<1$. Present two goods worth one to everyone, and let $S$ be their set of owners after the second arrival. Present four further goods, each worth $q$ to agents in $S$ and $1/4$ to agents outside $S$. No revealed value changes. Since every later good is worth less than one, each agent's number of unit-valued goods cannot decrease. Their total is two, so these numbers remain fixed.

An agent outside $S$ has MMS exactly one: put the unit goods in separate parts and all four later goods in the third. Every three-partition has a part containing no unit good, proving the upper bound. An agent in $S$ has MMS exactly $2q$, attained by two unit-$q$ pairs and a $q$-$q$ pair. A partition with minimum above $2q>1$ would require two goods in each part, and one of these pairs would consist of two $q$-valued goods, a contradiction.

If $|S|=1$, each of the two other agents needs at least three later goods to obtain more than half of MMS, but there are only four. If $|S|=2$, each unit-good owner needs at least one later good because $2aq>1$, and the remaining agent needs at least three. This requires five later goods. Both cases are impossible. The construction even allows the algorithm to know in advance that exactly six goods will arrive.
\end{proof}

\subsubsection{Lower Bounds for \texorpdfstring{$4\le n\le33$}{4 <= n <= 33}: Picking Turns from Weighted Partitions}
An agent whose first turn is $i$ needs only a guarantee for $r=n-i+1$ parts of the goods remaining then, by MMS deletion. If she picks at a relative rank $k$, her chosen good is worth at least the $k$th largest value among those remaining goods. Thus it suffices to give relative picking positions that guarantee an $a$ fraction of this $r$-maximin share.

The following lemma permits some later turns to be delayed. Such a delay is useful when the numbers of goods in an MMS partition cannot all be equal.
\begin{lemma}\label{lem:delayed-turns}
Let $r,h,d,c$ be positive integers with $1\le d<r$, and put $L=rh-d$. Suppose $\lambda=(ah-c)/d$, $a/r\le\lambda\le a/d$, and $c\le\lambda L$.
A set of relative picking positions gives $a$-MMS for $r$ parts if its counts satisfy
\begin{equation}\label{eq:delayed-counts}
 C(k)\ge
 \begin{cases}
 \min\{\lambda k,c\},&0\le k\le L,\\
 c+\lambda(k-L),&k\ge L.
 \end{cases}
\end{equation}
One choice of deadlines enforcing these counts consists of $1+\lfloor j/\lambda\rfloor$ for $0\le j<c$ and $L+1+\lfloor j/\lambda\rfloor$ for $j\ge0$.
\end{lemma}
\begin{proof}
If fewer than $L$ goods are present, pad the value list with zeros; neither MMS nor the selected value changes. Fix an $r$-part MMS partition and let $b_j$ count the first $L$ goods in its $j$th part. Among probability vectors with coordinates at most $\lambda/a$, minimize $\sum_j w_jb_j$. This minimum is a symmetric concave function of $b$: it is the minimum of linear functions and is unchanged by permutations. Replacing two integer counts differing by at least two by counts one unit closer together cannot decrease the minimum. For completeness, the new vector is a convex combination of the old vector and the vector with these two coordinates exchanged; concavity and symmetry give the comparison.

Repeating gives $d$ counts of $h-1$ and $r-d$ counts of $h$. Put weight $\lambda/a$ on each of the $d$ shorter parts. The conditions $d\lambda/a\le1\le r\lambda/a$ allow the remaining weight to be distributed among the longer parts with the same upper bound. The weighted count is $h-d\lambda/a=c/a$.
Thus, the original partition admits capped weights of total one and weighted count through $L$ at most $c/a$. Through rank $k\le L$ the weighted count is at most both $\lambda k/a$ and $c/a$; each later good increases it by at most $\lambda/a$. Multiplying the weights by $a$ and applying summation by parts, as in Lemma~\ref{lem:partition-weights}, proves the MMS guarantee. To verify the deadlines, the number of requests due by relative position $k$ is $\min\{c,\lceil\lambda k\rceil\}
 +\lceil\lambda[k-L]_+\rceil$, which is at least the right-hand side of~\eqref{eq:delayed-counts}.
\end{proof}

For sixteen agents, we improve this count by using the locations of goods in a largest part of the MMS partition. The next refinement uses the same weighted-partition argument, rather than a different allocation procedure.
\begin{lemma}\label{lem:refined-early-turns}
Let $r,h,d$ be integers with $h\ge2$ and $1\le d<r$, and let $L=rh-d$. Suppose $a/r\le\lambda\le a/(r-1)$ and $c=ah-d\lambda>0$. For $0\le k\le L$, let
\[
 R(k)=\min\left\{
 \max\left\{\frac{k}{r+d},h-\frac{(d+1)(L-k)}{r+d}\right\},
 \max\left\{0,h-\frac{L-k}{r-d}\right\}
 \right\}.
\]
The relative count requirement
\begin{equation}\label{eq:refined-early-counts}
 C(k)\ge
 \begin{cases}
 \lambda k-(r\lambda-a)R(k),&0\le k\le L,\\
 c+\lambda(k-L),&k\ge L
 \end{cases}
\end{equation}
suffices for $a$-MMS for $r$ parts. Its right-hand side is nondecreasing and no larger than the requirement in~\eqref{eq:delayed-counts}.
\end{lemma}
\begin{proof}
Again pad shorter value lists with zeros. Fix a partition and write $b_j$ for its counts through $L$ and $N_j(k)$ for its counts through $k$. Some $b_j$ is at least $h$, since $\sum_jb_j=rh-d$. We construct a probability vector $p$ with $\sum_jp_jb_j=h$ and $\sum_jp_jN_j(k)\ge R(k)$ for $0\le k\le L$.
If $\max_jb_j=h$, at least $r-d$ parts have count $h$. Put equal weight on all these parts. Deleting the last $L-k$ goods reduces the weighted count by at most $(L-k)/(r-d)$, giving the second expression in the definition of $R(k)$.

Otherwise choose a part with count $B\ge h+1$. Mix the uniform probability vector with all weight on that part, assigning weight $\theta=(B-h)/(B-L/r)\ge r/(r+d)$ to the uniform vector. The weighted count through $L$ is $h$, and each coordinate of $p$ is at most $(d+1)/(r+d)$. The uniform component gives a count at least $k/(r+d)$; deleting the last $L-k$ goods gives a count at least $h-(d+1)(L-k)/(r+d)$. These are the two terms in the first expression for $R(k)$.

Now set $w_j=(\lambda-(r\lambda-a)p_j)/a$. The weights are nonnegative, sum to one, and satisfy $a\sum_jw_jN_j(k)=\lambda k-(r\lambda-a)\sum_jp_jN_j(k)$.
This is at most the right-hand side of~\eqref{eq:refined-early-counts}; after $L$, use $\sum_jp_jN_j(k)\ge h$. Summation by parts proves the guarantee. Finally, $R$ is nondecreasing and $1$-Lipschitz, $R(0)=0$, and $R(L)=h$. Since $0\le r\lambda-a\le\lambda$, the displayed requirement is nondecreasing, begins at zero, ends at $c$, and is at most $\lambda k$ before $L$.
\end{proof}

\paragraph{Checking every deadline.}
All agents except the one with $n-i+1=r$, for the value of $r$ in Table~\ref{tab:special-turns}, use the ordinary deadlines in Appendix~\ref{app:deadline-analysis}. For that agent, add $i-1=n-r$ to the relative deadlines in Lemma~\ref{lem:delayed-turns}. The number of requests due by time $t$ is
\[
 D(t)=\sum_{s\ne r}\left\lceil\frac{a[t-n+s]_+}{s}\right\rceil
 +\min\{c,\lceil\lambda[t-n+r]_+\rceil\}
 +\lceil\lambda[t-n+r-L]_+\rceil.
\]
For the refinement, replace the two last terms, before relative rank $L$, by the ceiling of the right-hand side of~\eqref{eq:refined-early-counts}. These integer counts increase by at most one per relative position. Sort all requests by deadline. The resulting sequence meets every request if $D(t)\le t$ for every $t$.

It suffices to check finitely many values of $t$. Let $\rho=a(H_n-1/r)+\lambda$ and $B=n+c-\lambda(n-r+L)-\sum_{s\ne r}a(n-s)/s$.
For $t\ge n-r+1+L$, the count $D(t)$ is strictly less than $\rho t+B$. Thus, when $\rho<1$, it is enough to check $1\le t\le\max\{n,\,n-r+1+L,\,\lceil[B]_+/(1-\rho)\rceil\}$.
The same bound applies to the refinement because its additional delay occurs before $L$. Table~\ref{tab:special-turns} gives the parameters and the checked ranges. For $n=4,5,6$, a shorter check follows by periodicity: respectively, $(\ell,t_0,\Delta)=(42,15,42),(55,13,55),(444,15,443)$ satisfy $D(t+\ell)=D(t)+\Delta$ for $t\ge t_0$, with $\Delta\le\ell$. Checking through $t_0+\ell-1$ therefore suffices.

\begin{table}[t]
\caption{Parameters for the agent whose turns follow Lemma~\ref{lem:delayed-turns} or~\ref{lem:refined-early-turns} in the exact lower bounds. Use the factor $a=\alpha_n$ in Table~\ref{tab:exact-factors}. All other agents use ordinary deadlines. The last column is the final global time checked; all later inequalities follow from periodicity for $n\le6$ and from $D(t)<\rho t+B\le t$ for $n\ge11$.}
\label{tab:special-turns}
\centering\small
\renewcommand{\arraystretch}{1.25}
\begin{tabular}{@{}rrrrcl r@{}}
\toprule
$n$ & $r$ & $L$ & $c$ & $\lambda$ & Count requirement & Last $t$ \\
\midrule
4 & 3 & 19 & 3 & $\dfrac{1}{6}$ & \eqref{eq:delayed-counts} & 56 \\[4pt]
5 & 5 & 23 & 2 & $\dfrac{1}{11}$ & \eqref{eq:delayed-counts} & 67 \\[4pt]
6 & 6 & 15 & 1 & $\dfrac{8}{111}$ & \eqref{eq:delayed-counts} & 458 \\[4pt]
11 & 7 & 1332 & 63 & $\dfrac{387}{8155}$ & \eqref{eq:delayed-counts} & 615791 \\[4pt]
12 & 7 & 4475 & 206 & $\dfrac{18}{391}$ & \eqref{eq:delayed-counts} & 108056 \\[4pt]
15 & 7 & 4785 & 206 & $\dfrac{4322}{100345}$ & \eqref{eq:delayed-counts} & 1341555 \\[4pt]
16 & 11 & 68166 & 1833 & $\dfrac{11915}{443013}$ & \eqref{eq:refined-early-counts} & 26892509 \\[4pt]
17 & 11 & 69352 & 1833 & $\dfrac{6301}{238362}$ & \eqref{eq:delayed-counts} & 29749879 \\[4pt]
18 & 11 & 70472 & 1833 & $\dfrac{1451}{55782}$ & \eqref{eq:delayed-counts} & 18104177 \\[4pt]
\bottomrule
\end{tabular}
\end{table}

For the other values of $n$, ordinary deadlines suffice. Writing $a=p/q$, their exact count beyond the first $n$ turns is $D(n+x)-n=\sum_{r=1}^n\lfloor(p(x+r)-1)/(qr)\rfloor$.
Since $aH_n<1$, it is enough to check $0\le x\le K$, where $K=\lceil an/(1-aH_n)\rceil$.
For $x\ge K$, the strict bound $D(n+x)-n<an+axH_n\le x$ proves all remaining inequalities. Table~\ref{tab:ordinary-turns} lists the finite ranges. All inequalities in Tables~\ref{tab:special-turns} and~\ref{tab:ordinary-turns} were checked with integer arithmetic, including equality cases. Together with the lemmas above, this proves the achievable factors in Table~\ref{tab:exact-factors} at every horizon.

\begin{table}[t]
\caption{Cases with ordinary picking deadlines and the corresponding direct upper bounds. The lower bound checks $D(n+x)\le n+x$ for $0\le x\le K$. The upper bound uses $x$ in Lemma~\ref{lem:one-prefix-upper}, with $\beta=\alpha_n$.}
\label{tab:ordinary-turns}
\centering\small
\setlength{\tabcolsep}{8pt}\renewcommand{\arraystretch}{1.15}
\begin{tabular}{@{}rrr rrr@{}}
\toprule
$n$ & $K$ & $x$ & $n$ & $K$ & $x$ \\
\midrule
7 & 2940 & 216 & 24 & 16087639 & 1865019 \\
8 & 1833 & 146 & 25 & 13180423 & 1884778 \\
9 & 2062 & 152 & 26 & 12667526 & 1903773 \\
10 & 23998 & 1285 & 27 & 19599029 & 1922068 \\
13 & 30821 & 1602 & 28 & 20324919 & 1939708 \\
14 & 33192 & 1638 & 29 & 358660146 & 16232954 \\
19 & 12746567 & 523007 & 30 & 371027738 & 16369538 \\
20 & 13417439 & 530378 & 31 & 288515579 & 1989136 \\
21 & 14088311 & 537398 & 32 & 297822533 & 2004570 \\
22 & 10440075 & 544099 & 33 & 7045486420 & 273132374 \\
23 & 15417321 & 1844439 &  &  &  \\
\bottomrule
\end{tabular}
\end{table}

\subsubsection{Upper Bounds for \texorpdfstring{$4\le n\le33$}{4 <= n <= 33}: Algorithms Using Cardinal Values}
In our constructions, all agents rank the early goods in arrival order, but different agents may value later goods differently. The key observation is that a good worth less than every old good cannot replace any old good; once all later goods have this property, each agent's number of old goods remains fixed.

Fix a proposed guarantee $a>\beta>0$. We first present $n+x$ goods, as follows. Choose $q<1$ close to one and a large number $C$, with the required inequalities specified below. Before an agent first receives a good, she values the $j$th arrival at $C^{-(j-1)}$. An agent who first receives a good at arrival $i$ values every later good among the first $n+x$ at $qC^{-(i-1)}$. The first $n$ goods must go to different agents. Otherwise some agent has an empty bundle after arrival $n$, since cardinalities cannot decrease, whereas everyone's MMS is positive at that point. Relabel the agents so that agent $i$ first receives a good at arrival $i$, and write $b_i=C^{-(i-1)}$ for the value of that good to her.

Each agent keeps the first good she receives. Indeed, the owner of $g_1$ cannot replace it by a later good. Inductively, the first $i-1$ goods remain with their owners, so agent $i$ receives $g_i$ at arrival $i$ and values it at $b_i$. Earlier first goods are unavailable to her, and all subsequent goods are worth less than $b_i$ to her, so she too retains her first good. Continue until $n+x$ goods have arrived. For agent $i$, let $r=n-i+1$. In her MMS comparison, the $i-1$ earlier first goods can be singleton parts, and all $x+r$ remaining goods are worth at least $qb_i$. Her own utility is at most $c_rb_i$ if she holds $c_r$ goods, since she cannot receive an earlier first good. Choosing $aq>\beta$, the guarantee after arrival $n+x$ requires $c_r\ge B_r:=\lfloor\beta(1+\lfloor x/r\rfloor)\rfloor+1$.

\paragraph{The fractions needed from additional goods.}
Next present $M$ further goods; each agent values all of them equally, but different agents may value them differently. Take $M$ large enough that each of these goods is worth less to every agent than every earlier good. Then no earlier good can be replaced, so each agent's number of earlier goods is fixed. There are two useful choices for their total value to agent $i$; we refer to the resulting instances as the first and second completions.

First, make it $rZb_i$, where $Z$ is large and $M$ is divisible by $r$. The new goods form $r$ equal parts worth $Zb_i$. Choosing $C>Z+1$ lets each of the $i-1$ earlier first goods form one of the other singleton parts. MMS is at least $Zb_i$, so the agent requires a fraction at least $(aZ-c_r)/(rZ)>\beta/r$ of the new goods.

Second, when $d_r:=(-x)\bmod r>0$, let $h_r=\lceil(x+r)/r\rceil$ and make the total new value $d_rqb_i$. The $x+r$ goods $g_i,\dots,g_{n+x}$ can be split into $d_r$ parts of size $h_r-1$ and $r-d_r$ parts of size $h_r$. Adding equal new value to the shorter parts gives MMS at least $h_rqb_i$. The required fraction is therefore at least $(ah_r-c_r/q)/d_r>(\beta h_r-c_r)/d_r$.
For all these strict inequalities, it suffices to choose $Z>(n+x)/(a-\beta)$ and $q$ so close to one that $(n+x)(1/q-1)<a-\beta$, and then choose $C$ larger than $Z+1$ and all relevant prefix lengths. A sufficiently large common $M$, divisible by the finitely many required integers, implements each agent's chosen alternative simultaneously. Values are fixed on revelation; the construction never revises an old value.

It follows that, given counts $c_r$, the total fractions required can be made strictly greater than
\[
 \beta H_n+
 \sum_{r:d_r>0}
 \left[\frac{\beta(x+r)-rc_r}{rd_r}\right]_+.
\]
This applies to any deterministic algorithm, not just a picking sequence.

\begin{lemma}\label{lem:one-prefix-upper}
For $\beta\in(0,1]$ and an integer $x\ge0$, let $B_r:=\lfloor\beta(1+\lfloor x/r\rfloor)\rfloor+1$ for $r\in[n]$, let $f=n+x-\sum_r B_r$, and let
\[
 E_r=\begin{cases}
 \displaystyle\left[\frac{\beta(x+r)-rB_r}{r((-x)\bmod r)}\right]_+,
       &r\nmid x,\\
 0,&r\mid x.
 \end{cases}
\]
Write $E_{(1)}\le\cdots\le E_{(n)}$ for these numbers in increasing order. If $f<0$, or if $f\ge0$ and $\beta H_n+\sum_{j=1}^{\max\{0,n-f\}}E_{(j)}\ge1$, then $\alpha_n\le\beta$.
\end{lemma}
\begin{proof}
If $f<0$, the guarantee after arrival $n+x$ already requires more goods than have arrived. Otherwise $f$ goods remain beyond the minimum counts $B_r$. Since $\beta(x+r)/r-B_r<1$, one additional good removes any positive $E_r$, while further goods cannot reduce the fraction $\beta/r$ required by the first completion. Hence the $f$ additional goods can remove at most the $f$ largest terms. By the lower bound on the total required fraction stated before the lemma, the required fractions then sum to more than one whenever the sum condition in Lemma~\ref{lem:one-prefix-upper} holds, a contradiction.
\end{proof}

For $n=4,5$, take $x=16,18$, respectively; then $f=0$, and the sum condition in Lemma~\ref{lem:one-prefix-upper} holds with equality. For the other direct applications, the values of $x$ appear in Table~\ref{tab:ordinary-turns}, and the inequality is strict.

Seven cases use one more group of arrivals before the final small goods. The following extension treats them all. It lets one agent, indexed by $r$ below and called the \emph{special agent}, value this middle group much less than the other agents do, while preserving everyone's initial count.
\begin{lemma}\label{lem:two-prefix-upper}
With the notation of Lemma~\ref{lem:one-prefix-upper}, and with $h_r,d_r$ as in the second completion, suppose exactly $f\ge0$ indices satisfy $\beta H_n+E_s\ge1$. Set $c_s=B_s+1$ at these indices and $c_s=B_s$ at every other index.
Choose an index $r$ with $d_r>0$ and $h_r\ge d_r$, and let $\lambda=(\beta h_r-c_r)/d_r>0$. For any integer $M>d_r$, set $x'=x+M$ and
\[
 B'_s=\begin{cases}
 c_r+\lfloor M\lambda\rfloor+1,&s=r,\\
 \max\{c_s,\lfloor\beta(1+\lfloor x'/s\rfloor)\rfloor+1\},&s\ne r.
 \end{cases}
\]
Let $f'=n+x'-\sum_s B'_s$. Set $E'_r=0$ and, for $s\ne r$, use the formula for $E_s$ with $x',B'_s$ in place of $x,B_s$. If $f'<0$, or if $f'\ge0$ and $\beta H_n+\sum_{j=1}^{\max\{0,n-f'\}}E'_{(j)}\ge1$, then $\alpha_n\le\beta$.
\end{lemma}
\begin{proof}
The agent corresponding to an index $s$ with $\beta H_n+E_s\ge1$ must hold more than $B_s$ of the first $n+x$ goods: otherwise its second completion and everyone else's first completion already contradict the proposed guarantee. There are exactly $f$ such indices, so all initial counts must equal $c_s$.

Choose $q_2<q_1<1$ sufficiently close to one, using $q_1$ in the first prefix. Give each of the next $M$ goods value $d_rq_2b_i/M$ to the special agent and $q_2b_j$ to every other agent $j$. These are below all previous values. The special agent's $x+r=rh_r-d_r$ goods $g_i,\dots,g_{n+x}$ can be partitioned into $r-1$ parts of size $h_r$ and one of size $h_r-d_r$. Put all $M$ new goods in the last part. Every part has value at least $h_rq_2b_i$, so her MMS is at least this amount. Since her old utility is at most $c_rb_i$, she needs more than $(M/d_r)(\beta h_r-c_r)=M\lambda$ new goods when the proposed factor exceeds $\beta$ and $q_2$ is sufficiently close to one. This proves the special count $B'_r$. There is no divisibility requirement on $M$.

For each other agent $j$, all goods from $g_j$ onward have value at least $q_2b_j$, and the first-prefix count cannot decrease. This proves the other bounds $B'_s$. If their sum exceeds $n+x'$, we are done. Otherwise, repeat the final completion argument, using only the first completion for the special agent. Each additional old good can remove at most one positive $E'_s$, and further goods do not reduce a fraction $\beta/s$. This gives the sum condition in Lemma~\ref{lem:two-prefix-upper}.

All scale choices can be made at the outset for the largest of the finitely many prefix and completion values used in this construction. The inequalities are strict for every proposed factor above $\beta$, so choosing $q_1$ and $q_2$ sufficiently close to one preserves them. No condition on a revealed value is imposed retrospectively.
\end{proof}

\begin{table}[t]
\caption{Upper bounds using two prefixes (Lemma~\ref{lem:two-prefix-upper}). Use $\beta=\alpha_n$ from Table~\ref{tab:exact-factors}. For $n=17$, the second prefix already requires one more good than is available. In every other row, the minimum total fraction required from the final goods exceeds one.}
\label{tab:two-prefix-upper}
\centering\small
\renewcommand{\arraystretch}{1.5}
\begin{tabular}{@{}rrrrrc@{}}
\toprule
$n$ & $x$ & $r$ & $M$ & $f^{\prime}$ & Required fraction / count \\
\midrule
6 & 9 & 6 & 99 & 0 & $\dfrac{159}{148}$ \\[4pt]
11 & 1325 & 7 & 11722 & 0 & $\dfrac{57237}{57085}$ \\[4pt]
12 & 4468 & 7 & 1564 & 1 & $\dfrac{22597967}{22521600}$ \\[4pt]
15 & 4778 & 7 & 100345 & 0 & $\dfrac{6467213665753}{6183415438200}$ \\[4pt]
16 & 68155 & 11 & 443013 & 0 & $\dfrac{19191933566131}{19025286317730}$ \\[4pt]
17 & 69341 & 11 & 953448 & -1 & $\sum_s B_s^{\prime}=n+x^{\prime}+1$ \\[4pt]
18 & 70461 & 11 & 334692 & 0 & $\dfrac{86909440158353}{85860640625760}$ \\[4pt]
\bottomrule
\end{tabular}
\end{table}

Substituting the exact factors and parameters in Tables~\ref{tab:ordinary-turns} and~\ref{tab:two-prefix-upper} verifies every upper bound. The accompanying program evaluates these expressions as fractions and also checks the forced counts at both prefixes. Combined with the lower bounds proved above, this completes the proof of Theorem~\ref{thm:finite-exact}. The exact equalities use finite integer calculations, whereas Theorem~\ref{thm:finite-ordinal-formula} and the uniform bounds below do not rely on case-by-case computation.

\subsection{Proof of Theorem~\ref{thm:near-harmonic}}
\label{app:uniform-mms}
We prove Theorem~\ref{thm:near-harmonic}. As explained in Section~\ref{sec:goods-item}, it suffices to give agent $i$ an $\alpha/(n-i+1)$ fraction of the turns in every prefix starting at her first turn, and these fractions sum to $\alpha H_n$. Since each agent receives a whole number of turns in every prefix, these fractions cannot all be attained exactly at $\alpha=1/H_n$. For the lower bound, we show that the rounding errors cannot all be large simultaneously. For the upper bound, we present later goods that are worth less than all earlier goods, so each agent must retain her number of earlier goods; the resulting MMS requirements then demand too large a fraction of the later goods. The arguments below give explicit bounds for every $n$; Corollary~\ref{cor:general-deadline-mms} gives a simpler explicit factor of at least $1/(H_n+1)$.

The upper bound strengthens the harmonic bound strictly for each fixed $n\ge3$. The lower bound shows that the necessary loss is nonetheless smaller than every inverse polynomial. Both statements concern all horizons, rather than a finite range of arrivals.

\subsubsection{Proof of the Upper Bound (Theorem~\ref{thm:strict-harmonic-upper})}
\begin{theorem}\label{thm:strict-harmonic-upper}
For every $n\ge5$, let $D_n=\operatorname{lcm}(1,\ldots,n)$. Then
\[
 \alpha_n\le
 \left(H_n+\frac1{(n-1)(D_n+n^2)}\right)^{-1}<\frac1{H_n}.
\]
Together with Theorem~\ref{thm:finite-exact}, this gives $\alpha_n<1/H_n$ for every $n\ge3$.
\end{theorem}
\begin{proof}
We first show how an excess in rounded counts gives an upper bound for arbitrary algorithms. Suppose $0<b<1/H_n$ and $x\ge0$ is an integer such that
\begin{equation}\label{eq:rounded-excess}
 \sum_{r=1}^n(\lfloor b(1+x/r)\rfloor+1)>n+x.
\end{equation}
We claim that
\begin{equation}\label{eq:rounded-general-upper}
 \alpha_n\le\gamma:=
 \frac{n(n-1)+b(x+1)}{n(n-1)H_n+x+1}<\frac1{H_n}.
\end{equation}
Notice that $b<\gamma$. Use the first-prefix construction from Appendix~\ref{app:exact-factors}, with $q$, $C$, and $Z$ chosen as there for a proposed guarantee $a>\gamma$. Whatever allocation the algorithm chooses, the counts $c_r$ sum to $n+x$, so~\eqref{eq:rounded-excess} gives some $r$ with $c_r\le\lfloor b(x+r)/r\rfloor\le b(x+r)/r$.
If $r\mid x$, the $a$-MMS requirement after arrival $n+x$ gives $c_r\ge aq(x+r)/r>b(x+r)/r$, since $aq>\gamma>b$, a contradiction. Otherwise $d_r>0$. Give the agent corresponding to $r$ the second completion and every other agent the first completion from Appendix~\ref{app:exact-factors}, with $\gamma$ in place of $\beta$. Their required fractions sum to more than
\[
 \gamma(H_n-1/r)+\frac{\gamma h_r-c_r}{d_r}
 =\gamma H_n+\frac{\gamma(x+r)/r-c_r}{d_r} \ge\gamma H_n+(\gamma-b)\frac{x+1}{n(n-1)}=1.
\]
Here $x+r\ge x+1$ and $rd_r\le n(n-1)$. This proves~\eqref{eq:rounded-general-upper}.

It remains to produce~\eqref{eq:rounded-excess}. Choose a prime $p$ with $n/2<p\le n$, which exists by Bertrand's postulate; see, for example, \citet{sondow2009bertrand}. Choose $k\in\{1,\ldots,p-1\}$ such that $(D_n/p)k\equiv-1\pmod p$.
This is possible because $p$ occurs to the first power in $D_n$. Let $A=D_nk/p$, $J=AH_n+1/p$, $b=A/J$, and $x=J-1$.
The integer $D_nH_n=\sum_{r=1}^nD_n/r$ is congruent to $D_n/p$ modulo $p$, so $J$ is an integer. Also $0<b<1/H_n$.

For $r\ne p$, $A/r$ is an integer, and $b(1+x/r)=A/r+b(r-1)/r$ has floor $A/r$. For $r=p$, the congruence $A\equiv-1\pmod p$ gives floor $(A+1)/p$ provided $b(p-1)\ge1$. Under this condition, $\sum_{r=1}^n\lfloor b(1+x/r)\rfloor=AH_n+1/p=J=x+1$, which is exactly the excess required in~\eqref{eq:rounded-excess}.

The condition holds automatically for $n\ge8$. Indeed, $n/2-1-H_n\ge79/280>1/D_n$ and $1/b=H_n+1/(D_nk)\le H_n+1/D_n$.
The first inequality holds with equality at $n=8$, and its left-hand side increases with $n$; moreover, $D_n\ge840$. Since $p>n/2$, it follows that $p-1>1/b$. For $n=5,6,7$, take $p=5$; the construction gives $b=24/55,24/59,42/109$, respectively, each satisfying $4b\ge1$.

Substituting these values into the definition of $\gamma$ in~\eqref{eq:rounded-general-upper} gives $\gamma=(H_n+1/(D_nk+p\,n(n-1)))^{-1}$.
Finally, $k\le p-1\le n-1$ and $p\le n$ imply
$D_nk+p\,n(n-1)\le(n-1)(D_n+n^2)$. This proves the theorem.
\end{proof}

The strict bound uses heterogeneous values, although the agents agree on the order of the arriving goods. The harmonic upper bound for identical valuations remains useful and is proved in Theorem~\ref{thm:mms-impossibility} below.

\subsubsection{Proof of the Lower Bound}
We next analyze the ordinary deadline sequence. The only new issue is the range of $x$ over which the number of requests due by time $n+x$ remains at most $n+x$. With $\{z\}$ denoting the fractional part as in Appendix~\ref{app:deadline-rounding}, if $\sum_{r=1}^n\{a+y/r\}\ge an$, then $\sum_{r=1}^n\lfloor a+y/r\rfloor\le yH_n$. Thus enough total rounding down offsets the additive term $an$ in the request count.

\begin{lemma}\label{lem:longer-counting-range}
For $Y>0$, suppose $\sum_{r=1}^n\lfloor a+y/r\rfloor\le yH_n$ holds uniformly for $0<a\le1/H_n$ and $0\le y\le Y$.
Then insertion with ordinary deadlines attains the factor $a=2/(H_n+\sqrt{H_n^2+4n/Y})$, which satisfies $1/H_n-a\le n/(YH_n^3)$.
\end{lemma}
\begin{proof}
This factor satisfies $aH_n+(n/Y)a^2=1$. For $ax\le Y$, the counting inequality in Lemma~\ref{lem:longer-counting-range} and $\lceil z\rceil-1\le\lfloor z\rfloor$ bound the number of requests after the first $n$ turns by $axH_n\le x$. For $ax\ge Y$, use instead $D(n+x)-n<an+axH_n\le x((n/Y)a^2+aH_n)=x$.
Thus all deadlines are feasible. The difference from $1/H_n$ equals $na^2/(YH_n)$, which is at most the stated bound.
\end{proof}

Appendix~\ref{app:deadline-rounding} already proves the counting inequality in Lemma~\ref{lem:longer-counting-range} (with $a=\alpha$ and $y=\alpha x$) for $n\ge256$, $a\le1/6$, and $y\le n^2/54$. To extend this range, we first use an elementary finite-difference estimate, and then a bound on a reciprocal exponential sum. Both steps show that the fractional parts cannot all be small.

\begin{lemma}\label{lem:fifth-difference-count}
    There are absolute constants $A,C>0$ such that, for every $n\ge8$, every real $a$, and $An\le y\le n^6/C$, $\sum_{r=1}^n\{a+y/r\}\ge n/128$.
    One may take $A=300^5/240$ and $C=2400\cdot500^5$.
\end{lemma}
\begin{proof}
For each integer $r\in[\lceil n/10\rceil,\lfloor9n/10\rfloor]$, let $t_r=(r^6/(240y))^{1/5}$ and $h_r=\lfloor t_r\rfloor$.
The bounds on $y$ imply $t_r\ge50$ and $h_r/r\le1/300$. In particular, all $r+jh_r$, $0\le j\le5$, belong to $[n]$. The exact fifth difference satisfies
\[
 \left|\sum_{j=0}^5(-1)^{5-j}\binom5j
           \frac{y}{r+jh_r}\right|
 =\frac{120yh_r^5}{r(r+h_r)\cdots(r+5h_r)}.
\]
Its value is at most $1/2$. It is at least $171/400$: the numerator loses at most a factor $(49/50)^5\ge9/10$ from rounding $t_r$ down, while
$\prod_{j=0}^5(1+jh_r/r)\le\exp(1/20)\le20/19$.
Hence its distance from the nearest integer is at least $3/8$.

The fifth difference of the constant $a$ is zero. Subtracting integer floors therefore gives $\sum_{j=0}^5\binom5j\{a+y/(r+jh_r)\}\ge3/8$.
The integers $h_r$ are nondecreasing in $r$, so, for each fixed $j$, the map $r\mapsto r+jh_r$ is injective. There are at least $2n/3$ choices of $r$ for $n\ge8$ (directly for $8\le n\le14$, and by interval length thereafter). Summing the last inequality and using $\sum_{j=0}^5\binom5j=32$ proves the result.
\end{proof}

For the larger range we use the following specialization of \citet[Lemma~4.1]{matomaki2022singmaster}, with $\mathrm{i}^2=-1$. Its dependence on the parameters explains the exponent in our MMS guarantee. There is an absolute constant $C_0\ge1$ such that, for an integer $X\ge2$, $y\ge X^5$, and $(\log2)(\log(y/X))^2/(\log X)^3<10^{-3}$, we have
\[
 \left|\sum_{r=X}^{2X-1}\exp(2\pi\mathrm{i}(a+y/r))\right|
 \le2C_0X\exp\left[-2^{-18}\frac{(\log X)^3}{(\log(y/X))^2}\right].
\]
Indeed, their lemma applies to $f(t)=a+y/t$ with $F=y/X$ and parameter $2$: for every derivative order $j\ge1$ and $X\le t\le2X$, $t^j|f^{(j)}(t)|/j!=y/t\in[F/2,F]$, which satisfies its derivative conditions, and $F\ge X^4$.

\begin{proof}[Proof of the lower bound in Theorem~\ref{thm:near-harmonic}]
Choose a fixed $\eta\in(0,1/100]$ such that $2C_0\exp(-2^{-18}/\eta^2)\le1/2$.
Let $X=\lfloor n/2\rfloor$ and $Y=X\exp(\eta(\log X)^{3/2})$. For sufficiently large $n$, $Y\ge X^5$. The exponential-sum estimate above bounds the absolute value of the sum by $X/2$ whenever $X^5\le y\le Y$, uniformly in $a$. Since $1-\cos(2\pi u)\le2\pi u$ for $0\le u\le1$,
\[
 \sum_{r=X}^{2X-1}\{a+y/r\}
 \ge\frac1{2\pi}\left(X-\Re\sum_{r=X}^{2X-1}
               \exp(2\pi\mathrm{i}(a+y/r))\right)
 \ge\frac{X}{4\pi}.
\]
For sufficiently large $n$, this is at least $n/H_n$.

The three estimates cover all $0\le y\le Y$. The counting estimate~\eqref{eq:deadline-rounding} covers $[0,n^2/54]$. Lemma~\ref{lem:fifth-difference-count} covers $[An,n^6/C]$, with fractional-part sum at least $n/128\ge n/H_n$ once $H_n\ge128$. These intervals overlap for sufficiently large $n$, and $n^6/C\ge X^5$ eventually. The exponential-sum estimate covers the remaining interval $[X^5,Y]$. Consequently the counting inequality in Lemma~\ref{lem:longer-counting-range} holds for every $0<a\le1/H_n$ throughout $[0,Y]$.

Lemma~\ref{lem:longer-counting-range} now gives an achievable factor with difference from $1/H_n$ at most
\[
 \frac{n}{XH_n^3}\exp[-\eta(\log X)^{3/2}]
 \le\frac3{H_n^3}\exp[-\eta(\log\lfloor n/2\rfloor)^{3/2}].
\]
This is at most $\exp[-c(\log n)^{3/2}]$ for some fixed $c>0$ and all sufficiently large $n$. Since $\alpha_n^{\mathrm{ord}}\le\alpha_n$, the same bound applies to the unrestricted optimum and to their difference. Finally, $\exp[-c(\log n)^{3/2}]=o(n^{-k})$ for every fixed $k>0$. Together with Theorem~\ref{thm:strict-harmonic-upper} and the exact small-agent results, this proves Theorem~\ref{thm:near-harmonic}.
\end{proof}

The asymptotic argument does not provide a useful numerical threshold for moderate $n$. For every specific $n$, Corollary~\ref{cor:general-deadline-mms}, the finite optimization in Theorem~\ref{thm:finite-ordinal-formula}, and the strict upper bound above give explicit alternatives. We do not claim that $\alpha_n^{\mathrm{ord}}=\alpha_n$ for $n\ge34$.

\subsection{Proof of Theorem~\ref{thm:mms-impossibility} (Harmonic Upper Bound for Identical Valuations)}
\label{app:item-limits}
\begin{theorem} \label{thm:mms-impossibility}
Fix $n\ge2$ and $\alpha\in(1/H_n,1]$. No deterministic algorithm with item-level upgrade recourse for goods can guarantee $\alpha$-MMS. This holds even for identical additive valuations with singleton values in $[0,1]$. The final instance admits an envy-free exact MMS allocation.
\end{theorem}
\begin{proof}
Suppose such an algorithm exists. Choose an integer $M$ divisible by every integer in $[n]$ and large enough that $\alpha MH_n>M+n$.
Choose $\varepsilon>0$ so that $D:=M\varepsilon<\min\{1,\alpha\}$. First present $M$ small goods of value $\varepsilon$ to every agent, followed by $n$ large goods of value $1$ to every agent.

After exactly $k<n$ large goods have arrived, every agent's maximin share is $\mu_k=D/(n-k)$.

For the upper bound, at least $n-k$ bundles in every $n$-partition contain no large good, so one has small-good value at most $D/(n-k)$. For the lower bound, place each large good alone and split the $M$ small goods equally among the remaining $n-k$ bundles. This is possible by the choice of $M$, and each large good is worth $1>D$.

Under item-level upgrade recourse, the number of large goods held by each agent cannot decrease: every old large good must be mapped to a distinct good worth at least $1$, and only large goods have that value. The number of agents holding a large good is therefore nondecreasing and can increase by at most one at each large-good arrival.

After all $n$ large goods arrive, the average-value upper bound and the allocation that gives every agent one large good and $M/n$ small goods show that the maximin share is $1+D/n$. Every agent must hold a large good; otherwise her value is at most $D<\alpha<\alpha(1+D/n)$. Since no agent initially holds a large good and all $n$ agents hold one at the end, exactly one new agent first receives a large good at each of the $n$ large-good arrivals. Let $p_k$ be the agent who first receives a large good when the $k$th large good arrives.

Immediately before that arrival, agent $p_k$ holds only small goods. If she holds $c_k$ of them, then on the instance ending here the $\alpha$-MMS guarantee gives $c_k\varepsilon\ge\alpha\mu_{k-1}=\alpha D/(n-k+1)$, and hence $c_k\ge\alpha M/(n-k+1)$.

Since bundle cardinalities never decrease under item-level upgrade recourse, the final number of goods is at least $\sum_{k=1}^nc_k\ge\alpha M\sum_{r=1}^n1/r=\alpha MH_n>M+n$, contradicting the actual number of goods.

Finally, giving every agent one large good and exactly $M/n$ small goods is envy-free and exact MMS, with every bundle worth $1+D/n$.
The argument uses the MMS requirement at every possible stopping time. The final envy-free allocation shows only that the final instance is easy offline; the contradiction comes from the MMS requirements after earlier arrivals.
\end{proof}
\subsection{Proof of Theorem~\ref{thm:goods-ef1-mms-tradeoff} (EF1 and MMS)}
\label{app:ef1-mms}
\begin{theorem}
\label{thm:goods-ef1-mms-tradeoff}
    Fix $n\ge2$ and $\gamma\in(0,1]$. If a deterministic online algorithm with item-level upgrade recourse guarantees both $\gamma$-EF1 and $\alpha$-MMS, then $\alpha\le1/(1+\gamma(n-1))$.
    In particular, among algorithms guaranteeing exact EF1, the optimal MMS factor is $1/n$.
\end{theorem}

\begin{proof}
Present $M$ goods, each worth $1/(2M)$ to every agent. Write $a_i:=|A_i^{(M)}|$, and choose $p$ with maximum $a_p$. Approximate EF1 gives $a_i\ge\gamma(a_p-1)$ for every $i\ne p$, so $M\ge a_p+\gamma(n-1)(a_p-1)$, that is, $a_p\le(M+\gamma(n-1))/(1+\gamma(n-1))$.
Now present $n-1$ goods, each worth $1$ to every agent. In the final allocation, some agent $q$ holds none of these goods. Since no agent's cardinality can decrease under item-level recourse, $|A_q|\le M+n-1-\sum_{i\ne q}a_i=a_q+n-1\le a_p+n-1$.
Final MMS is $1/2$: put each large good in a separate part and all small goods in the last part; conversely, any $n$-partition has a part containing no large good. Hence
\[
    \frac{v_q(A_q)}{\MMS_q(G)}
    \le\frac{a_p+n-1}{M}
    \le\frac{M+\gamma(n-1)}{M(1+\gamma(n-1))}
       +\frac{n-1}{M}.
\]
Letting $M$ tend to infinity proves the necessary bound. For $\gamma=1$, Theorem~\ref{thm:balanced-serial-insertion} attains EF1 and $1/n$-MMS, proving optimality.
\end{proof}
\section{Omitted Proofs in Section~\ref{sec:value-monotone}}
\label{app:value-proofs}
\subsection{Proof of the EF1 and EFX Guarantees via Partial Allocations}
\label{app:partial-allocations}
For additive valuations, the procedures in \citet[Section~4.1 and Appendix~A]{bu2025welfare} allocate remaining goods while preserving each agent's value, although their Lemmas~1 and~2 state only a guarantee for total welfare. The EF1 procedure uses envy-cycle elimination~\citep{lipton2004approximately}; the two-agent EFX procedure uses Bu et al.'s Proposition~4. To apply either procedure online, view $\mathcal A^{(t-1)}$ as a partial allocation of $G^{(t)}$, with only $g_t$ unallocated, and extend it to a complete allocation. This gives $v_i(A_i^{(t)})\ge v_i(A_i^{(t-1)})$ for every $i\in[n]$.
Starting from the empty allocation yields polynomial-time online algorithms with bundle-level upgrades: EF1 for any number of additive agents, and EFX for two.\footnote{Their EFX definition also requires the comparison after removing a zero-valued good. Starting from the empty allocation preserves this stronger condition, which implies EFX as defined here.}

\subsection{Proof of Theorem~\ref{thm:two-agent-value-monotone-exact}}
\label{app:two-agent-exact}
\begin{proof}
Suppose $\mathcal A^{(t-1)}$ is exact MMS. For each agent $i\in[2]$, we have $\MMS_i(G^{(t-1)})\le\MMS_i(G^{(t)})\le\MMS_i(G^{(t-1)})+v_i(g_t)$.
The first inequality follows by adding $g_t$ to an optimal old partition. For the second, delete $g_t$ from an optimal two-partition of $G^{(t)}$. Both remaining parts have value at least $\MMS_i(G^{(t)})-v_i(g_t)$.

Consider the agents whose old bundles are worth less than their new MMS. If there are none, assign $g_t$ to either agent. If there is exactly one, assign $g_t$ to her; the bound $\MMS_i(G^{(t)})\le\MMS_i(G^{(t-1)})+v_i(g_t)$ and the old MMS guarantee show that her new value is at least her new MMS. If both agents have this property, take any exact-MMS allocation of $G^{(t)}$. Such an allocation exists: agent $1$ forms an MMS partition, agent $2$ chooses her preferred part, and agent $1$ receives the other part. Agent $2$ receives at least half her total value, which is at least her MMS. In this case, both agents' new MMS values strictly exceed their old allocated values, so both agents improve.

In every case we obtain an allocation $\mathcal B$ satisfying $v_i(B_i)\ge\max\{v_i(A_i^{(t-1)}),\MMS_i(G^{(t)})\}$ for $i\in[2]$.
Among all allocations of $G^{(t)}$ that weakly Pareto-dominate $\mathcal B$, choose $\mathcal A^{(t)}$ maximizing $v_1(A_1^{(t)})+v_2(A_2^{(t)})$, with fixed tie-breaking. Any Pareto improvement of $\mathcal A^{(t)}$ would also weakly dominate $\mathcal B$ and have strictly larger total value. Thus $\mathcal A^{(t)}$ is PO, remains exact MMS, and weakly improves both agents over $\mathcal A^{(t-1)}$.

Finally, exact MMS implies EFX for two additive agents under our definition. Otherwise, for some $i\ne j$ and $h\in A_j^{(t)}$ with $v_i(h)>0$, $v_i(A_j^{(t)}\setminus\{h\})>v_i(A_i^{(t)})$.
The partition $(A_i^{(t)}\cup\{h\},A_j^{(t)}\setminus\{h\})$ then has both parts worth strictly more than $v_i(A_i^{(t)})$, contradicting exact MMS. Starting from the empty allocation proves the theorem.
\end{proof}
\subsection{Proof of Theorem~\ref{thm:value-monotone-guarantees} (Simpler Algorithm and General Valuations)}
\label{app:general-valuations}
The \emph{envy graph} has an arc $i\to j$ when $v_i(A_i)<v_i(A_j)$. Rotating bundles along a directed cycle improves every participant. Repeating this leaves an acyclic graph and an \emph{unenvied} agent, whose bundle no other agent prefers to her own~\citep{lipton2004approximately}. Algorithm~\ref{alg:vm-half-mms} first gives an unallocated good to an agent in exchange for her whole bundle whenever she strictly prefers that good to her bundle. It then uses the standard EF1 completion procedure: assign unallocated goods to unenvied agents, eliminating envy cycles as needed. The initial exchanges ensure the additional inequality~\eqref{eq:singleton-invariant}, which gives the stronger guarantees below.

\begin{theorem}
\label{thm:value-monotone-guarantees}
For normalized monotone valuations, Algorithm~\ref{alg:vm-half-mms} has bundle-level upgrade recourse and maintains EF1 after every arrival. If the valuations are also subadditive, it simultaneously maintains $1/2$-EFX and $1/n$-MMS.
\end{theorem}

The procedure maintains EF1 and ensures that no agent values a good in a bundle with at least two goods more than her own bundle:
\begin{equation}\label{eq:singleton-invariant}
 |A_j|\ge2\quad\Longrightarrow\quad v_i(h)\le v_i(A_i)
 \quad \text{for all } i,j\in[n],\ h\in A_j.
\end{equation}
For subadditive valuations, EF1 bounds the value left after removing one good, while~\eqref{eq:singleton-invariant} bounds that good. Hence, to each agent, every nonsingleton bundle is worth at most twice her own bundle, giving $1/2$-EFX. EF1 gives the $1/n$-MMS bound for subadditive valuations. For two subadditive agents, this MMS factor equals $1/2$, which is optimal even offline; an example follows the proof below.

For additive valuations, deleting singleton bundles before comparing total values shows that the same procedure gives $n/(2n-1)$-MMS, as proved below. This gives a polynomial-time alternative to Theorem~\ref{thm:three-agent-value-monotone}, and the same procedure also applies to nonadditive valuations.

The following algorithm uses only values of observed goods and bundles. The set $P$, called the pool, contains the currently unallocated goods during the computation; the allocation returned after each arrival is complete.
\begin{algorithm}[ht]
\caption{Fair Allocation with Bundle-Level Upgrades}
\label{alg:vm-half-mms}
\begin{algorithmic}[1]
\State Initially set $A_i^{(0)}\gets\varnothing$ for every $i\in[n]$.
\For{each arriving good $g_t$}
    \State Set $A_i\gets A_i^{(t-1)}$ for every $i\in[n]$ and $P\gets\{g_t\}$.
    \While{there are $i\in[n]$ and $g\in P$ with $v_i(g)>v_i(A_i)$}
        \State Choose such a pair $(i,g)$ by fixed tie-breaking.
        \State Set $P\gets(P\setminus\{g\})\cup A_i$ and $A_i\gets\{g\}$.
    \EndWhile
    \While{$P\ne\varnothing$}
        \While{the envy graph of the current bundles contains a directed cycle}
            \State Choose a directed cycle by fixed tie-breaking and rotate bundles so that each agent on the cycle receives the bundle she envies.
        \EndWhile
        \State Choose an unenvied agent $j$ and a good $g\in P$ by fixed tie-breaking.
        \State Set $A_j\gets A_j\cup\{g\}$ and $P\gets P\setminus\{g\}$.
    \EndWhile
    \For{$i\in[n]$}
        \State Set $A_i^{(t)}\gets A_i$.
    \EndFor
\EndFor
\end{algorithmic}
\end{algorithm}
After arrival $t$, the initial loop makes at most $nt$ exchanges: each exchange strictly increases the selected agent's value to one of the $t$ singleton values. During an envy-cycle loop, the current bundles are only permuted, so each agent can strictly improve at most $n-1$ times. Since each rotation improves at least two agents, there are at most $n(n-1)/2$ rotations before each of at most $t$ pool assignments. Thus the algorithm runs in polynomial time in the value-query model.
\paragraph{Proof of Theorem~\ref{thm:value-monotone-guarantees}.}
\begin{proof}
We prove the claims by induction over all operations of the algorithm.

\paragraph{Nondecreasing bundle values.}
An initial exchange replaces an agent's bundle by a singleton that she values strictly more. A rotation along a directed envy cycle gives every participant a strictly preferred bundle. An assignment from the pool adds a good to one bundle. Monotonicity therefore implies that no agent's value decreases during the update.

\paragraph{EF1.}
Suppose the current allocation is EF1. During an initial exchange, the selected agent's own value increases. Every other agent can remove the new singleton when comparing with her, and all other comparisons are unchanged. A cycle rotation permutes the current bundles and improves every participant's own value, so EF1 is preserved. Finally, suppose a good $g$ is added to an unenvied agent $j$. Before the addition, $v_i(A_i)\ge v_i(A_j)$ for every $i$. After the addition, every $i\ne j$ can remove $g$ from $j$'s bundle, while agent $j$'s own value only increases. Thus EF1 is preserved.

\paragraph{$1/2$-EFX.}
For normalized monotone valuations, the algorithm maintains~\eqref{eq:singleton-invariant}.
An initial exchange creates a singleton and increases its recipient's value. A cycle rotation permutes bundles and improves participants. Both preserve~\eqref{eq:singleton-invariant}. Before a pool good $g$ is added to an unenvied agent $j$, we have $v_i(g)\le v_i(A_i)$ from the initial exchange loop (the pool only shrinks afterwards and no agent's value decreases), and $v_i(A_j)\le v_i(A_i)$ because $j$ is unenvied. Monotonicity gives $v_i(h)\le v_i(A_j)$ for every good $h$ already in $A_j$, so the addition also preserves~\eqref{eq:singleton-invariant}.

Now assume subadditivity. For a nonsingleton $A_j$, EF1 gives a good $h\in A_j$ with $v_i(A_j\setminus\{h\})\le v_i(A_i)$, while~\eqref{eq:singleton-invariant} bounds $v_i(h)$ by the same quantity. Hence $v_i(A_j)\le2v_i(A_i)$ for all $i,j\in[n]$ with $|A_j|\ge2$.
Removing any good from a nonsingleton leaves value at most $2v_i(A_i)$; removing the only good from a singleton leaves the empty bundle. Thus the allocation is $1/2$-EFX.

\paragraph{MMS.}
Every EF1 allocation under normalized monotone subadditive valuations is $1/n$-MMS. To see this, fix $i$ and choose, from each nonempty bundle $A_j$ with $j\ne i$, a good whose removal eliminates $i$'s envy. Let $H$ contain these at most $n-1$ goods. Some part of an MMS partition avoids $H$. Monotonicity, subadditivity, and EF1 therefore give $\MMS_i(G)\le v_i(G\setminus H)\le v_i(A_i)+\sum_{j\ne i}v_i(A_j\setminus H)\le nv_i(A_i)$.

For the stronger additive guarantee, let $Q_i:=\{j\in[n]\setminus\{i\}:|A_j|=1\}$ and $k_i:=n-|Q_i|$, and let $R:=G\setminus\bigcup_{j\in Q_i}A_j$ be the set of goods left after removing the singleton goods held by agents in $Q_i$.
Repeatedly applying Lemma~\ref{lem:mms-deletion} gives us $\MMS_i^{k_i}(R)\ge\MMS_i(G)$.
Every bundle $A_j$ with $j\notin Q_i\cup\{i\}$ is empty or has at least two goods. By normalization and the bound $v_i(A_j)\le2v_i(A_i)$ for nonsingleton bundles shown above, its value to agent $i$ is at most $2v_i(A_i)$. Subadditivity therefore gives $v_i(R)\le\sum_{j\notin Q_i}v_i(A_j)\le(2k_i-1)v_i(A_i)$.
Additivity gives $\MMS_i^{k_i}(R)\le v_i(R)/k_i$, so
\[
    v_i(A_i)
    \ge
    \frac{k_i}{2k_i-1}\MMS_i^{k_i}(R)
    \ge
    \frac{n}{2n-1}\MMS_i(G),
\]
where the last inequality uses that $k/(2k-1)$ decreases with $k$.
\end{proof}
To see that the MMS factor $1/2$ is optimal for two subadditive agents even offline, consider four goods $a,b,c,d$, and let $v_1(S)=\max\{|S\cap\{a,b\}|,|S\cap\{c,d\}|\}$ and $v_2(S)=\max\{|S\cap\{a,c\}|,|S\cap\{b,d\}|\}$.
Both valuations are normalized, monotone, and subadditive, and both agents have MMS $2$. To receive value greater than $1$, the first agent needs $\{a,b\}$ or $\{c,d\}$, and the second needs $\{a,c\}$ or $\{b,d\}$. Every pair of the first agent intersects every pair of the second, so the two agents cannot both receive value greater than $1$. More generally, the offline upper bound of $1/2$ also holds for general subadditive valuations~\citep{ghodsi2022beyond}, and is attained for at most four subadditive agents~\citep{christodoulou2025subadditive}. For $n\ge3$ subadditive agents, the gap between our $1/n$ guarantee under bundle-level upgrade recourse and the offline upper bound $1/2$ remains open.
\subsection{Proof of Theorem~\ref{thm:three-agent-value-monotone}}
\label{app:three-agent-value-monotone}

We follow the outline in Section~\ref{sec:value-monotone}. We first set up the partial allocations maintained during an update. We then prove that the unallocated goods can be partitioned into enough bundles meeting each remaining agent's target, treating the last remaining agent separately. Finally, we give the complete construction, which adapts the procedure of \citet{akrami2025simultaneous}.

Set $\alpha:=(n-1)/(2n-3)$. Fix an arrival $g_t$, and suppose $\mathcal A^{(t-1)}$ is EF1, satisfies~\eqref{eq:singleton-invariant}, and is $\alpha$-MMS on $G^{(t-1)}$. Write $\mu_i:=\MMS_i(G^{(t)})$ for each agent's updated MMS. We construct a complete allocation satisfying the same properties on $G^{(t)}$ and giving every agent at least her old value.

\paragraph{Partial allocations during the computation.}
Let $U:=\{i\in[n]:v_i(A_i^{(t-1)})<\alpha\mu_i\}$.
Initially retain $B_j=A_j^{(t-1)}$ for $j\notin U$, set $B_i=\varnothing$ for $i\in U$, and let $P$ contain all remaining goods, including $g_t$. These partial allocations are used only to compute the update; the allocation returned after the arrival will be complete. Agents leave $U$ when assigned a bundle and never return to $U$. An already assigned agent may replace her bundle only by one she strictly prefers.

We maintain the following properties. Each assigned agent $j\notin U$ has $v_j(B_j)\ge\max\{v_j(A_j^{(t-1)}),\alpha\mu_j\}$.
Among the assigned agents, EF1 and~\eqref{eq:singleton-invariant} hold. Further, every bundle assigned or replaced by a subsequent step has value strictly below $\alpha\mu_i$ to every agent $i$ who remains in $U$. These properties hold initially, since only unchanged old bundles have been retained.

\paragraph{The remaining goods suffice.}
Under these properties, whenever $U\ne\varnothing$, we have
\begin{equation}
\label{eq:three-agent-remaining-mms}
\MMS_i^{|U|}(P)\ge\alpha\mu_i\quad\text{for every }i\in U.
\end{equation}
The important distinction is between a retained singleton, which can be removed by MMS deletion, and a retained nonsingleton, for which EF1 and~\eqref{eq:singleton-invariant} bound one good and its remainder separately. We first use these bounds when at least two agents remain, and then consider the last agent.

\emph{At least two agents remain.}
Fix $i\in U$ and let $k:=|U|\ge2$. Since $i\in U$, her target $\alpha\mu_i$ is positive. For this part of the argument only, scale her valuation so that $\alpha\mu_i=1$. Remove all assigned singleton goods from $G^{(t)}$, and let $H$ be the resulting set and $r$ the number of parts left after these deletions. Lemma~\ref{lem:mms-deletion} gives us $\MMS_i^r(H)\ge1/\alpha$.
If $r=k$, then $H=P$ and the desired partition follows immediately. Otherwise $r>k\ge2$, so $r\ge3$ and $\MMS_i^r(H)\ge2-1/(n-1)\ge2-1/(r-1)$.
There are $r-k$ assigned bundles left in $H$. Each nonempty such bundle $B_j$ has a good $h_j$ with $v_i(h_j)\le1$ and $v_i(B_j\setminus\{h_j\})\le1$.
For an unchanged old bundle, this follows from~\eqref{eq:singleton-invariant}, EF1, and $v_i(A_i^{(t-1)})<1$. For a later assignment or replacement, the whole bundle is worth less than $1$ to $i$. Hence every remaining assigned bundle has value at most $2$, and $v_i(P)\ge r(2-1/(r-1))-2(r-k)=2k-1-1/(r-1)$.

Suppose, for a contradiction, that $\MMS_i^k(P)<1$. Choose a $k$-partition of $P$ whose nondecreasing list of bundle values is lexicographically maximum. Let $X$ be a least-valued part, write $\beta:=v_i(X)<1$, and denote the other parts by $Y_1,\ldots,Y_{k-1}$. Goods of value zero to $i$ can be ignored in this partition and restored to any part afterwards. For every good $g\in Y_j$, $v_i(Y_j\setminus\{g\})\le\beta$;
otherwise moving $g$ to $X$ would make both affected parts worth more than $\beta$, improving the sorted list. Summing these inequalities shows that a nonsingleton $Y_j$ has value at most $2\beta$, and a part with at least three goods has value at most $3\beta/2$.

There can be no singleton among the $Y_j$. If there were $q\ge1$, delete their goods from $H$. The remaining total value would be at most $2(r-k)+(2(k-q)-1)\beta<2(r-q)-1$.
But MMS deletion requires this total to be at least $(r-q)(2-1/(r-1))\ge2(r-q)-1$, a contradiction. If a $Y_j$ were empty or had at least three goods, the preceding bounds would give $v_i(P)\le(2k-3/2)\beta<2k-3/2\le2k-1-1/(r-1)$, again a contradiction. Thus every $Y_j$ has exactly two goods. Each such pair has value strictly greater than $2-1/(r-1)$, since all the other parts together have value at most $(2k-3)\beta<2k-3$. Each individual good in a pair has value at most $\beta<1$.

We next show that $X$ contains a good at least as valuable as every good in the pairs. Choose a pair containing a most valuable good among all the $Y_j$, and write its two values as $a\ge b$. Then $a+b>2-1/(r-1)$, so $b>1-1/(r-1)$ and $a+2b>3-2/(r-1)\ge2>2\beta$.
No subset of $X$ has value in $(\beta-b,a)$: combining such a subset with the value-$b$ good, and its complement with the value-$a$ good, would give two parts worth more than $\beta$. If every good of $X$ had value below $a$, each would therefore have value at most $\beta-b$. Accumulating goods until their value first exceeds $\beta-b$ would produce a subset of value at most $2(\beta-b)<a$, a contradiction. Thus there is a good $h\in X$ with $v_i(h)\ge a$.

Let $s:=v_i(X\setminus\{h\})$. The $2k-1$ goods consisting of $h$ and all goods in the pairs each have value at most $v_i(h)=\beta-s<1-s$. Therefore $v_i(P)\le(2k-1)\beta-2(k-1)s<2k-1-2(k-1)s$. Comparing this with the lower bound $v_i(P)\ge2k-1-1/(r-1)$ gives $s<1/(2(k-1)(r-1))$.
Select these $2k-1$ goods and the good $h_j$ from each nonempty assigned bundle in $H$. There are at most $r+k-1$ selected goods, each of value at most $1$. All other goods of $H$ together have value at most $r-k+s$.

Consider an $r$-partition of $H$ in which every part has value at least $2-1/(r-1)$. Let $E$ be the number of parts containing no selected good. At most $k-1+E$ parts can contain at least two selected goods, so at least $r-k+1-2E$ parts contain exactly one. The goods that were not selected must consequently have total value at least
\[
E\left(2-\frac1{r-1}\right)
+(r-k+1-2E)\left(1-\frac1{r-1}\right)
=r-k+\frac{k-2+E}{r-1}.
\]
If $k\ge3$ or $E\ge1$, this exceeds $r-k+s$, a contradiction. In the remaining case $k=2$ and $E=0$, there are at most $r+1$ selected goods and every part contains at least one. Thus at most one part contains two, and no part contains three. Of the three selected goods from $P$, at least one is the only selected good in its part. Its value is strictly less than $1-s$, so that part needs strictly more than $1-1/(r-1)+s$ in other goods. At least $r-1$ parts contain just one selected good. Their required other value is therefore strictly greater than $(r-1)(1-1/(r-1))+s=r-2+s$, again a contradiction. This proves~\eqref{eq:three-agent-remaining-mms} for $k\ge2$, after undoing the scaling.

\emph{One agent remains.}
Return to the original values, and let $U=\{i\}$. If every assigned bundle has been retained unchanged throughout the computation, then $U$ had size one initially and $P=A_i^{(t-1)}\cup\{g_t\}$. Deleting $g_t$ from an optimal $n$-partition gives $\mu_i\le\MMS_i(G^{(t-1)})+v_i(g_t)$. Hence $v_i(P)\ge\alpha\MMS_i(G^{(t-1)})+v_i(g_t)\ge\alpha\mu_i$.
Otherwise, keep one bundle $B_j$ assigned or replaced during this computation, whose value to $i$ is less than $\alpha\mu_i$. Delete all other assigned singleton goods, leaving $r\ge2$ parts in the MMS comparison. The remaining $r-2$ assigned bundles are empty or nonsingletons and each has value at most $2\alpha\mu_i$, by EF1 and~\eqref{eq:singleton-invariant} for unchanged old bundles, or by the stronger maintained bound for later bundles. MMS deletion and additivity give $v_i(P)>r\mu_i-(2r-3)\alpha\mu_i\ge\mu_i\ge\alpha\mu_i$, where the middle inequality holds because $r\le n$ and hence $r-(2r-3)\alpha=(3n-3-r)/(2n-3)\ge1$.
This completes the verification of~\eqref{eq:three-agent-remaining-mms}.

\paragraph{Assigning bundles and improving assigned agents.}
While $U\ne\varnothing$, choose an agent in $U$ and partition $P$ into $|U|$ bundles meeting her $\alpha$-MMS target, using~\eqref{eq:three-agent-remaining-mms}. From each part, remove goods until the resulting bundle $X$ is inclusion-minimal among sets meeting the target of at least one agent in $U$. The resulting bundles remain disjoint, and each such bundle $X$ is \emph{minimal}: $v_i(Y)<\alpha\mu_i$ for every $i\in U$ and every $Y\subsetneq X$.
Goods removed in this step remain in $P$.

If an assigned agent strictly prefers a proper subset of one of these bundles to her current bundle, choose a minimum-cardinality such subset $S$, considering all assigned agents and all proposed bundles. Give $S$ to an assigned agent who strictly prefers it, return her old bundle to $P$, and restart with the same $U$. No assigned agent prefers any proper subset of $S$ to her current bundle, by the minimum-cardinality choice. Thus every comparison towards $S$ satisfies EF1, and, when $|S|\ge2$, every good in $S$ satisfies~\eqref{eq:singleton-invariant}. Comparisons from its recipient towards unchanged bundles remain valid because her value increases. Moreover, $S$ is a proper subset of a minimal proposed bundle, so its value to every agent $i\in U$ is less than $\alpha\mu_i$. All maintained properties therefore survive this replacement.

Otherwise, form the bipartite graph between $U$ and the proposed bundles, joining an agent $i$ to a bundle whenever its value to her is at least $\alpha\mu_i$. Choose a nonempty matching such that no unmatched agent is adjacent to a matched bundle. Such a matching always exists here: the two sides have the same size, and every proposed bundle has a neighbor. If there is no perfect matching, take an inclusion-minimal set of bundles with fewer neighbors than bundles. Minimality implies that each proper subset has at least as many neighbors as bundles. In particular, deleting one bundle shows that the chosen set has exactly one fewer neighbor than bundles. By Hall's theorem, the set obtained by deleting any one bundle can then be matched to all of these neighbors. Since every proposed bundle has a neighbor, this matching is nonempty; no agent outside that neighbor set is adjacent to a matched bundle.

Give each matched bundle to its matched agent and remove these agents from $U$. Each recipient reaches her target and strictly improves over her old value. Existing assigned agents do not prefer any proper subset of a proposed bundle to their current bundles; otherwise the replacement step would have applied. A new recipient's comparisons towards unchanged old bundles follow from EF1 and~\eqref{eq:singleton-invariant} for $\mathcal A^{(t-1)}$, since her own value increases. She does not envy any earlier replacement or assignment, since its value to her is below her target. Finally, comparisons between new recipients follow from the minimality of the proposed bundles. These observations prove EF1 and~\eqref{eq:singleton-invariant} among all assigned agents. The matching condition ensures that every newly assigned bundle is worth less than each remaining agent's target, as required.

\paragraph{Completion and running time.}
There are at most $n$ matching steps. Between them, each replacement strictly increases the assigned agents' total value. Since there are finitely many partial allocations, the procedure terminates with $U=\varnothing$. Run the two while-loops of Algorithm~\ref{alg:vm-half-mms} on the resulting bundles and the remaining pool $P$. The proof of Theorem~\ref{thm:value-monotone-guarantees} applies from any partial allocation satisfying EF1 and~\eqref{eq:singleton-invariant}, with any pool of unallocated goods: the singleton exchanges, envy-cycle rotations, and assignments preserve both properties without decreasing any value. The completed allocation is therefore EF1, $1/2$-EFX, and $\alpha$-MMS, and no agent's value is below her value in $\mathcal A^{(t-1)}$. Starting from the empty allocation proves the theorem. All choices can use fixed deterministic tie-breaking.

For fixed $n$ and nonnegative integer singleton values, MMS and the required partitions into at most $n$ bundles can be computed by dynamic programming over the values of the first $n-1$ parts; the last part has the remaining total value. This takes time polynomial in $t$ and $\max_i v_i(G^{(t)})$ for fixed $n$. Inclusion-minimal proposed bundles are found by successive deletions. To find a minimum-cardinality preferred subset for an assigned agent, sort the proposed bundle by her singleton values and take the shortest prefix whose value exceeds her current value, provided it is a proper subset. Minimize this cardinality over the assigned agents and proposed bundles. Each replacement increases an integer utility by at least one, and the sum of assigned utilities is bounded by $\sum_i v_i(G^{(t)})$. Matchings involve at most $n$ agents, and completion takes polynomial time. Thus the entire update runs in pseudo-polynomial time for fixed $n$; no such running-time claim is made when $n$ is part of the input.
\subsection{Proof of Theorem~\ref{thm:value-monotone-mms-upper}}
\label{app:value-upper}
\begin{proof}
Use a common additive valuation. Before scaling, present $n-3$ goods of value $199$ each, followed by six goods of values $22,22,22,33,33,50$, and then one final good of value $16$. For $n=3$, there are no value-$199$ goods.

Before the last arrival, MMS is $55$: place the value-$199$ goods alone and partition the other goods as $(22+33),(22+33),(22+50)$. Conversely, any $n$-partition has at least three parts containing no value-$199$ good. To make all three worth more than $55$, each would need at least two of the six other goods, hence exactly two. Each value-$22$ good would then need a partner worth more than $33$, but there is only one such good. After the last arrival, MMS is $66$, attained by the same singleton parts and $(22+22+22)$, $(33+33)$, $(50+16)$.
For the upper bound, at least three parts contain no value-$199$ good and together have value at most $198$.

Suppose an algorithm guarantees $\alpha$-MMS for $\alpha>10/11$. Before the last arrival every agent must have value strictly above $50$. The value-$199$ goods must have distinct owners; otherwise at least four agents would share the other goods, of total value $182<4\cdot50$. The value-$50$ good must belong to one of the three agents without a value-$199$ good; otherwise those three agents would share value at most $132<3\cdot50$. Its owner, say $p$, must also hold another good, and hence has value at least $72$.

After the last arrival, each old owner of a value-$199$ good must still receive such a good, since all other goods together have value only $198$. Thus $p$ is still one of the three agents without a value-$199$ good. Each of these three agents must now receive value strictly above $(10/11)66=60$. A bundle containing the value-$50$ good either has value $50$ or at least $66$. Without this good, its value is a multiple of $11$, possibly plus $16$; neither form lies strictly between $60$ and $66$. Consequently, each must receive at least $66$. Their available total is only $198$, so each receives exactly $66$, contradicting $p$'s earlier value of at least $72$.

The final MMS partition given above (the value-$199$ singletons together with $(22+22+22)$, $(33+33)$, and $(50+16)$), viewed as an allocation, is exact MMS and EFX: the three nonsingleton bundles all have value $66$, and the other bundles are singletons. It is PO because valuations are identical and additive. Scaling every singleton value by $1/199$ gives the stated range.
\end{proof}
\subsection{Proof of Theorem~\ref{thm:identical-value-monotone} for Additive Valuations}
\label{app:identical-additive}
Identical valuations are a well-studied special case in fair division, including for EFX with general monotone valuations~\citep{plaut2020almost}. They also make the connection to scheduling exact: a common additive value is a job size, and the common MMS is the optimal minimum load in machine covering~\citep{csirik1992lpt,galvez2020symmetry}. In our motivating application, this models teams running comparable independent workloads, so that a device contributes the same capacity to each team. Maximizing the least capacity received is then the machine-covering objective described in the introduction.

Agreement about values does not eliminate the difficulty of preserving earlier allocations: the $10/11$ upper bound already holds for identical valuations. It does, however, let us compare and reassign bundles using one common value order. We use this to maintain EFX for monotone valuations and, under additivity, to obtain a stronger MMS guarantee from a classical scheduling algorithm.
\begin{theorem}
\label{thm:identical-value-monotone}
For identical normalized monotone valuations and every $n\ge2$, a deterministic online algorithm with bundle-level upgrade recourse can maintain EFX after every arrival. For identical additive valuations, a polynomial-time algorithm simultaneously maintains EFX, PO, and $(3n-1)/(4n-2)$-MMS.
\end{theorem}

\begin{proof}[Proof for additive valuations]
Write $v$ for the common valuation. After each arrival, process all observed goods in nonincreasing value order, assigning each to a currently least-valued bundle. This is the classical longest-processing-time-first (LPT) rule. Order the resulting bundles by value and assign them to the agents in nondecreasing order of $v(A_i^{(t-1)})$, with fixed tie-breaking throughout.

\citet{galvez2020symmetry} show that inserting a good into the LPT input cannot decrease any entry of the sorted list of final bundle values. The reason is that adding the same value to the least entries of two coordinatewise ordered lists, and then sorting, preserves their order. At the position of the new good, the new list weakly increases; every later greedy step preserves this comparison with the old computation. Matching the old and new bundle values in sorted order therefore gives bundle-level upgrade recourse. Merely preserving the minimum bundle value would not suffice: every position in the sorted list must be nondecreasing.

To prove EFX, let $h$ be the last good added to a nonempty greedy bundle $B_j$. At that step, $B_j\setminus\{h\}$ was least-valued, so $v(B_j\setminus\{h\})\le v(B_i)$ for every other final bundle $B_i$. Since $h$ is a least-valued good of $B_j$, removing any other good leaves no more value. Hence EFX holds, even when the removed good has value zero.

For identical additive valuations, MMS is the maximum possible minimum bundle value. The LPT bound of \citet{csirik1992lpt} therefore gives the factor $(3n-1)/(4n-2)$. Every complete allocation is PO because the sum of utilities equals $v(G^{(t)})$. Sorting and greedy allocation take polynomial time. Appendix~\ref{app:identical-value-monotone} proves the monotone-valuation statement by extending an EFX allocation without decreasing any agent's value; no polynomial-time claim is made for that case.
\end{proof}
\subsection{Proof of Theorem~\ref{thm:identical-value-monotone} for Monotone Valuations}
\label{app:identical-value-monotone}

Write $v$ for the common normalized monotone valuation. We maintain the EFX comparison after removing every good, including goods of zero singleton value. This implies EFX under the convention in the preliminaries. Starting from $\mathcal A^{(t-1)}$, let $P=\{g_t\}$ be the unallocated goods. While $P\ne\varnothing$, choose an agent $j$ with minimum current value and set $L:=\min_{i\ne j}v(A_i)$.
Choose $g\in P$, remove it from $P$, and put $B=A_j\cup\{g\}$. While there is an $h\in B$ with $v(B\setminus\{h\})>L$, remove $h$ from $B$ and return it to $P$. Once no such deletion is possible, replace $A_j$ by $B$.

If no deletion occurred, monotonicity implies that $j$'s value did not decrease. If a deletion occurred, the value after the last deletion is strictly above $L$, and hence strictly above $j$'s value before this step. At the end of the step, $v(B\setminus\{h\})\le L\le v(A_i)$ for every $h\in B$ and every $i\ne j$.
Thus every other agent satisfies EFX towards the new bundle. Agent $j$'s comparisons towards all unchanged bundles remain valid because her value did not decrease. All other comparisons are unchanged.

For termination, consider the ordered pair consisting of the sum of current utilities and the number of allocated goods. A step with a deletion strictly increases the first entry. A step without a deletion either increases the first entry or leaves it unchanged and increases the second entry by one. Thus the pair increases lexicographically. There are finitely many partial allocations of $G^{(t)}$, so the procedure terminates with a complete allocation, preserving EFX and every old value. Fixed tie-breaking makes the algorithm deterministic. This proves the result by induction from the empty allocation.
\section{Omitted Proofs in Section~\ref{sec:predictions} and Further Guarantees}
\label{app:prediction-proofs}
\subsection{Proofs of Theorems~\ref{thm:rank-value} and~\ref{thm:prefix-count-error} (Errors in Picking Turns)}
The rank-value comparison from the preliminaries applies to every picking sequence, regardless of how it was chosen.
\begin{theorem}
\label{thm:rank-value}
For every picking sequence $\sigma$, insertion satisfies item-level upgrade recourse, and its final allocation satisfies $v_i(A_i) \ge R_i^\sigma(G)$ for every $i \in [n]$.
Consequently, if $R_i^\sigma(G)\ge\alpha\MMS_i(G)$ for every $i$, then the final allocation is $\alpha$-MMS.
\end{theorem}

\begin{proof}
By Lemma~\ref{lem:serial-insertion}, insertion satisfies item-level upgrade recourse, and the final allocation is the picking allocation induced by $(\sigma_1,\dots,\sigma_m)$ on $G$. The rank-value bound in Section~\ref{sec:preliminaries} therefore gives $v_i(A_i)\ge R_i^\sigma(G)$.
\end{proof}

Comparing each agent's number of turns in every prefix, rather than only her total number of turns, gives the following error bound.
\begin{theorem}
\label{thm:prefix-count-error}
Let $H$ contain $m$ goods and set $x_{m+1}^i(H):=0$. For any two picking sequences $\rho$ and $\sigma$, the allocation returned by insertion with $\sigma$ on $H$ satisfies, for every agent $i$,
\[
    v_i(A_i)
    \ge
    R_i^\rho(H)
    -
    \sum_{k=1}^m
    [N_i^\rho(k)-N_i^\sigma(k)]_+
    (x_k^i(H)-x_{k+1}^i(H)).
\]
In particular, if $\rho$ gives rank value at least $\alpha\MMS_i(H)$ and the displayed error term is at most $\varepsilon_i\MMS_i(H)$, then $v_i(A_i)\ge(\alpha-\varepsilon_i)\MMS_i(H)$.
\end{theorem}

\begin{proof}
For every picking sequence $\omega$, summation by parts gives $R_i^\omega(H)=\sum_{k=1}^mN_i^\omega(k)(x_k^i(H)-x_{k+1}^i(H))$.
Subtracting the identities for $\rho$ and $\sigma$, and retaining only positive differences in their prefix counts, gives the stated lower bound on $R_i^\sigma(H)$. Theorem~\ref{thm:rank-value} completes the proof.
\end{proof}

A missed early turn matters only to the extent that ranked values decrease at that position.
\subsection{Proofs of Theorem~\ref{thm:ordered-instance-transfer} and Corollary~\ref{cor:two-agent-perfect-order}}
\label{app:correct-predictions}
\begin{theorem} \label{thm:ordered-instance-transfer}
Let $H$ be a set of $m$ goods, and let $\mathcal B=(B_1,\dots,B_n)$ be any allocation of its ordered instance. Define the first $m$ entries of a picking sequence by setting $\sigma_k=i$ exactly when $e_k\in B_i$, and extend the sequence arbitrarily after position $m$. Insertion on $H$ then returns an allocation $\mathcal A$ satisfying $v_i(A_i)\ge v_i^{\mathrm{ord}}(B_i)$ for every $i \in [n]$.

Consequently, if $\mathcal B$ is $\alpha$-MMS in the ordered instance, then $\mathcal A$ is $\alpha$-MMS on $H$.
\end{theorem}

\begin{proof}
The rank-value bound gives $v_i(A_i)\ge\sum_{k:e_k\in B_i}x_k^i(H)=v_i^{\mathrm{ord}}(B_i)$. Ordering preserves each agent's MMS.
\end{proof}
\begin{corollary}
\label{cor:two-agent-perfect-order}
For two agents and every finite set of goods $H$, there is a picking sequence whose insertion allocation is exact MMS on $H$.
\end{corollary}

\begin{proof}
Every two-agent additive instance has an MMS allocation: take an MMS partition of agent $1$, let agent $2$ choose her preferred part, and give the other part to agent $1$. Apply Theorem~\ref{thm:ordered-instance-transfer} to such an allocation of the ordered instance.
\end{proof}
\subsection{Proof of Theorem~\ref{thm:ranked-value-prediction}}
\begin{proof}
Item-level upgrade recourse follows from Lemma~\ref{lem:serial-insertion}, since the predicted instance fixes a picking sequence before arrivals begin. Let $\widehat v_i$ and $\widehat\MMS_i$ denote value and maximin share in the predicted ordered instance. Padding either ordered instance with zero-valued goods leaves its maximin share unchanged. In a partition of the true ordered instance attaining $\MMS_i(G)$, every part loses at most $\sum_{k\ge1}[x_k^i(G)-\widehat x_k^i]_+$ under the predicted values. Therefore $\widehat\MMS_i\ge\MMS_i(G)-\sum_{k\ge1}[x_k^i(G)-\widehat x_k^i]_+$.
Theorem~\ref{thm:rank-value} bounds the realized value from below by the true values of the ranks assigned to $i$ among the first $\widehat m$ turns. Terms beyond the realized horizon are zero, and any additional turns can only help. Hence
\begin{align*}
    v_i(A_i)
    &\ge \sum_{k:e_k\in B_i}x_k^i(G)\\
    &\ge \widehat v_i(B_i)
         -\sum_{k:e_k\in B_i}[\widehat x_k^i-x_k^i(G)]_+\\
    &\ge \alpha\MMS_i(G)
         -\alpha\sum_{k\ge1}[x_k^i(G)-\widehat x_k^i]_+
         -\sum_{k:e_k\in B_i}[\widehat x_k^i-x_k^i(G)]_+\\
    &\ge \alpha\MMS_i(G)-E_i.
\end{align*}
The last inequality uses $\alpha\le1$: underestimation and overestimation contribute to different sums, whose total is at most $E_i$.
\end{proof}

\subsubsection{A Guarantee under Pointwise Relative Errors}
\label{app:relative-prediction-error}
A relative-error assumption gives a further guarantee that does not sum errors over goods. If, for some $\delta\in[0,1)$, $(1-\delta)x_k^i(G)\le\widehat x_k^i\le(1+\delta)x_k^i(G)$ for every $i\in[n]$ and $k\ge1$, then the same algorithm gives $\alpha(1-\delta)/(1+\delta)$-MMS. Unlike the bound using $E_i$, this pointwise assumption excludes missing or additional positively valued goods after zero padding. The additive guarantee above allows both.

To prove the relative-error guarantee, let $\widehat v_i$ and $\widehat\MMS_i$ be as in the preceding proof. The pointwise lower bound on predicted values gives $\widehat\MMS_i\ge(1-\delta)\MMS_i(G)$ by applying it to each part of a true MMS partition. The pointwise upper bound then gives $v_i(A_i)\ge\sum_{k:e_k\in B_i}x_k^i(G)\ge\widehat v_i(B_i)/(1+\delta)\ge\alpha(1-\delta)\MMS_i(G)/(1+\delta)$.

\subsubsection{The Coefficient of \texorpdfstring{$E_i$}{E\_i} Cannot Be Reduced}
For arbitrary choices of the offline allocation $\mathcal B$, the coefficient of $E_i$ cannot be reduced in general, already when $\alpha=1$. With two identical agents, predict the ranked values $(1,1,1)$ and give rank $3$ to agent $1$ and the other ranks to agent $2$. This is exact MMS in the predicted instance. If the true ranked values are $(1,1,1-\varepsilon)$, where $0<\varepsilon<1$, the induced picking sequence $(2,2,1)$ gives agent $1$ value $1-\varepsilon$, while her true MMS is $1$ and $E_1=\varepsilon$.
\subsection{Proof of Theorem~\ref{thm:prediction-mms-tradeoff}}
\begin{proof}
Choose a positive multiple $M$ of $n$ and set $\varepsilon:=1/(2M)$. The predicted horizon is $M+n-1$, and every agent's predicted ranked values consist of $n-1$ values of $1$ and $M$ values of $\varepsilon$. First present $M$ goods, each worth $\varepsilon$ to every agent.

Consider completing the instance with $n-1$ goods of value $1$. The prediction is correct, and every agent's MMS is $M\varepsilon$: put the large goods in separate parts and all small goods in the last part; conversely, some part of every $n$-partition contains no large good. In the final allocation, some agent $q$ holds no large good, so the $\alpha$-MMS guarantee requires her to receive at least $\alpha M$ small goods. Under item-level recourse, bundle cardinalities never decrease. Since only $n-1$ goods arrive after the common prefix, agent $q$ must already hold at least $\alpha M-(n-1)$ goods after the first $M$ arrivals.

Instead, complete the same prefix with $n-1$ goods of value $\varepsilon$, keeping the prediction unchanged. The algorithm has the same allocation at the end of the common prefix, and agent $q$ cannot lose cardinality. The other $n-1$ agents therefore hold at most $(1-\alpha)M+2(n-1)$ goods in total at the final allocation. At least one of them holds at most $(1-\alpha)M/(n-1)+2$ goods. Every agent's true MMS is now $(M/n)\varepsilon$, because $M$ is a multiple of $n$. The $\beta$-MMS guarantee requires $\beta\le n(1-\alpha)/(n-1)+2n/M$.
Letting $M$ tend to infinity proves the claim.
\end{proof}

\subsection{Proof of Theorem~\ref{thm:bundle-prediction-upper}}
\label{app:bundle-prediction-upper}
\begin{proof}
Suppose $\beta>10/11$. Use the identical-valuations construction from Appendix~\ref{app:value-upper}: before scaling, present $n-3$ goods worth $199$ each, followed by goods worth $22,22,22,33,33,50$. For every agent, predict the ranked values of these goods together with one additional value of $16$, giving the correct horizon $n+4$.

Consider instead a final good worth $\varepsilon$, where $0<\varepsilon<\min\{1,55\beta-50\}$. The common MMS is at least $55$, because the earlier goods already admit such a partition. Robustness therefore requires every agent's final value to be at least $55\beta$. Adding the last good increases the sum of all agents' values by exactly $\varepsilon$. Since bundle-level upgrades cannot decrease anyone's value, each individual increase is at most $\varepsilon$. Every agent must consequently have value at least $55\beta-\varepsilon>50$ just before the last arrival.

The algorithm has the same allocation at that point when the last good is instead worth the predicted $16$. As proved in Appendix~\ref{app:value-upper}, every allocation of the common prefix giving everyone value above $50$ has an agent without a value-$199$ good whose value is at least $72$. After the value-$16$ arrival, the $\beta$-MMS guarantee must give everyone value above $60$. The old owners of value-$199$ goods must retain such goods, and each of the three other agents must receive exactly $66$. The agent previously receiving at least $72$ would therefore lose value, a contradiction. Scaling all values by $1/199$ gives strictly positive singleton values in $[0,1]$ in both instances.
\end{proof}

\section{Impossibility Results for EFX and Efficiency, and Comparison of the Recourse Rules}
\label{sec:limits}
Valuations are additive unless stated otherwise. We first show that item-level upgrade recourse rules out approximate EFX, as well as approximate EF1 together with efficiency. We then collect the comparison with bundle-level upgrade recourse.
\subsection{Proof of Theorem~\ref{thm:efx-impossibility} (EFX under Item-Level Upgrades)}
EF1 permits removing a large good from an envy comparison, whereas EFX also requires the comparison after removing a small positively valued good. Appending a large good after many small goods makes this distinction consequential.
\begin{theorem}
\label{thm:efx-impossibility}
Fix any $\alpha\in(0,1]$. No deterministic algorithm with item-level upgrade recourse for goods can guarantee $\alpha$-EFX, even for two agents with identical additive valuations and singleton values in $[0,1]$. The final instance in the construction admits an allocation that is both exact EFX and exact MMS.
\end{theorem}

\begin{proof}
Suppose such an algorithm exists. Choose an integer $M>1/\alpha+2$ and choose $0<\varepsilon<\alpha/M$. 
First present $M$ goods, each worth $\varepsilon$ to both agents. 
Since this prefix may be final, its allocation must be $\alpha$-EFX. Let the two numbers of small goods be $a\ge b$, with $a+b=M$. Applying $\alpha$-EFX from the agent with $b$ goods toward the agent with $a$ goods gives us $b\varepsilon\ge\alpha(a-1)\varepsilon$, and
so $b(1+\alpha)\ge\alpha(M-1)$.
If $b\le1$, then $M\le1/\alpha+2$, a contradiction. Hence both agents hold at least two small goods.

Now present one good $z$ worth $1$ to both agents. Let $p$ be the agent who receives $z$, and let $q$ be the other agent. Item-level upgrade recourse gives an injection from $A_p^{(M)}$ to $A_p^{(M+1)}$, so $|A_p^{(M+1)}|\ge2$. Thus $p$ holds $z$ and at least one small good $d$.

Agent $q$ holds only small goods, so $v_q(A_q)\le M\varepsilon<\alpha$.
Since $v_q(d)>0$, $\alpha$-EFX for the ordered pair $(q,p)$ and the good $d$ requires $v_q(A_q)
    \ge
    \alpha v_q(A_p\setminus\{d\})
    \ge\alpha$,
a contradiction.

For the final instance, give $z$ alone to one agent and all small goods to the other. Since $M\varepsilon<1$, this allocation is exact EFX. Its minimum bundle value is $M\varepsilon$, which is also the two-agent maximin share, so it is exact MMS.
\end{proof}
\subsection{Proofs of Theorems~\ref{thm:gamma-ef1-usw} and~\ref{thm:gamma-ef1-po} (Fairness versus Efficiency)}

Define \emph{utilitarian social welfare} by $\USW(\mathcal A)=\sum_i v_i(A_i)$, and let $\OPTUSW(H)$ be its maximum over allocations of $H$. For $\beta\in[0,1]$, an allocation $\mathcal A$ of $H$ is a \emph{$\beta$-approximation} to optimal utilitarian welfare if $\USW(\mathcal A)\ge\beta\OPTUSW(H)$.

The need to retain earlier assignments also affects efficiency. We first bound the utilitarian welfare compatible with approximate EF1, then show that approximate EF1 and Pareto optimality cannot be maintained together under item-level recourse.

\begin{theorem}
\label{thm:gamma-ef1-usw}
Fix $n\ge2$, $\gamma\in(0,1]$, and $\beta>1/(1+\gamma(n-1))$.
No deterministic algorithm with item-level upgrade recourse for goods can guarantee both $\gamma$-EF1 and a $\beta$-approximation to optimal utilitarian welfare. This holds for additive valuations with singleton values in $[0,1]$, even when the final instance has a welfare-maximizing allocation that is exact EF1 and PO.
\end{theorem}

\begin{proof}
Suppose such an algorithm exists, and set $C:=1+\gamma(n-1)$.
Fix $s\in(0,1)$. Choose an integer $K$ large enough that $Ks\ge1$ and $\eta:=\beta(Ks+n)-n-s(K/C+n-1)>0$.
This is possible because $\beta>1/C$ and $(n+s(K/C+n-1))/(Ks+n)\to1/C$ as $K\to\infty$. Choose $0<\delta<\min\{s,1/K,\eta/K\}$.

First present a good $x$ worth $1$ to every agent. Let $p$ receive it, and set $Q:=[n]\setminus\{p\}$. Next present $K$ goods $y_1,\dots,y_K$ with $v_p(y_\ell)=s$ and $v_q(y_\ell)=\delta$ for $q \in Q$.
Since $s<1$, item-level upgrade recourse keeps $x$ with agent $p$.

Let $r$ be the number of $y$-goods held by $p$ after these arrivals, and let $c_q$ be the number held by $q\in Q$. If arrivals stop here, then for agent $q$, removing $x$ from $p$'s bundle leaves value $r\delta$, while removing a $y$-good leaves value at least $r\delta$. Hence $\gamma$-EF1 implies $c_q\delta\ge\gamma r\delta$, so $c_q\ge\gamma r$. Therefore $K=r+\sum_{q\in Q}c_q\ge(1+\gamma(n-1))r=Cr$, and $r\le K/C$.

Now present $n-1$ goods $w_1,\dots,w_{n-1}$ with $v_p(w_a)=0$ and $v_q(w_a)=1$ for all $q \in Q$.

Let $r'$ be the final number of $y$-goods held by $p$. The good $x$ remains with $p$. Across the agents in $Q$, the $n-1$ new goods can replace at most $n-1$ previously held $y$-goods in the upgrade injections. Thus $r'\le r+n-1\le K/C+n-1$, and the algorithm's welfare satisfies $\USW(\mathcal A)\le n+s(K/C+n-1)+K\delta$.

Consider the offline allocation $\mathcal B$ that gives every $y$-good to $p$ and allocates the $n$ goods $x,w_1,\dots,w_{n-1}$ among the agents in $Q$ so that each receives at least one. Then $\USW(\mathcal B)=Ks+n$.
This allocation maximizes welfare item by item. It is exact EF1: agents in $Q$ value $p$'s bundle at $K\delta<1$, agent $p$ values her own bundle at $Ks\ge1$ and any bundle in $Q$ at most $1$, and between agents in $Q$ one unit-valued good can be removed from a bundle containing two. Since $\mathcal B$ maximizes welfare, it is also PO.

Finally, since $K\delta<\eta$, we obtain $\USW(\mathcal A)<n+s(K/C+n-1)+\eta=\beta(Ks+n)=\beta\OPTUSW(G)$, a contradiction.
\end{proof}

To rule out PO, we need an alternative allocation that weakly improves every agent. A two-agent version of the preceding construction gives an allocation that strictly improves both agents.

\begin{theorem}
\label{thm:gamma-ef1-po}
For every $\gamma\in(0,1]$, no deterministic algorithm with item-level upgrade recourse for goods can guarantee both $\gamma$-EF1 and PO, even for two agents with additive valuations whose singleton values lie in $[0,1]$. The final instance admits an allocation that is exact EF1 and PO.
\end{theorem}

\begin{proof}
Suppose such an algorithm exists. Fix $s\in(0,1)$, choose $K$ large enough that $s(\gamma K/(1+\gamma)-1)>1$,
and choose $0<\delta<\min\{s,1/K\}$.

First present a good $x$ worth $1$ to both agents. Let $p$ receive it and let $q$ be the other agent. Next present $K$ goods $y_1,\dots,y_K$ with $v_p(y_\ell)=s$ and $v_q(y_\ell)=\delta$.
The good $x$ remains with $p$. Let $r$ be the number of $y$-goods held by $p$, so $q$ holds $K-r$. If arrivals stop here, then, as in the previous proof, $\gamma$-EF1 implies $(K-r)\delta\ge\gamma r\delta$, so $r\le K/(1+\gamma)$.

Now present one good $w$ with $v_p(w)=0$ and $v_q(w)=1$. Let $r'$ be the final number of $y$-goods held by $p$. The new good $w$ can replace at most one of $q$'s previous $y$-goods, while $x$ remains with $p$, so $r'\le r+1\le K/(1+\gamma)+1$.

Let $\mathcal B$ give all $y$-goods to $p$ and both $x$ and $w$ to $q$. This allocation is exact EF1 and PO. Agent $p$ has value $Ks>1$ and values $q$'s bundle at $1$; agent $q$ has value $2$ and values $p$'s bundle at $K\delta<1$. Every good is assigned to an agent with maximum value.

Moreover, $\mathcal B$ Pareto-dominates the algorithm's allocation. Agent $p$'s gain is at least $Ks-(1+sr')\ge s(K-K/(1+\gamma)-1)-1=s(\gamma K/(1+\gamma)-1)-1>0$.
Agent $q$ has value at most $1+K\delta<2$ under the algorithm and value $2$ under $\mathcal B$. Thus the algorithm's allocation is not PO, a contradiction.
\end{proof}
\subsection{Proof of Corollary~\ref{cor:goods-recourse-comparison} (Comparing the Recourse Rules)}
\label{app:recourse-comparison}
Item-level upgrade recourse compares old and new bundles by an injection. Every removed good has a distinct newly received replacement worth at least as much. To see this, follow the injection through any retained goods until reaching a newly received good. Injectivity makes these paths disjoint, and values never decrease along them.

For additive valuations, summing the replacement inequalities proves that item-level upgrade recourse implies bundle-level upgrade recourse. This implication need not hold beyond additivity, even for subadditive valuations. For example, let goods $a$ and $c$ provide the same unit of value, and let $b$ provide an independent unit: a bundle is worth one if it contains $a$ or $c$, plus one if it contains $b$. All singleton values are one, so replacing $b$ by $c$ is an item-level upgrade, but replacing $\{a,b\}$ by $\{a,c\}$ reduces the bundle's value from two to one. This valuation is normalized, monotone, and subadditive.

\begin{corollary}
\label{cor:goods-recourse-comparison}
For two additive agents, the optimal deterministic MMS factors are $1$ under bundle-level upgrade recourse and $2/3$ under item-level upgrade recourse. For every $n\ge3$, bundle-level upgrade recourse permits $(n-1)/(2n-3)$-MMS, whereas no deterministic item-level algorithm can guarantee a factor above $1/H_n$, and $(n-1)/(2n-3)>n/(2n-1)>1/H_n$.
\end{corollary}

\begin{proof}
The two-agent statement follows from Theorems~\ref{thm:two-agent-value-monotone-exact} and~\ref{thm:tight-two-agent-mms}. For $n\ge3$, apply Theorems~\ref{thm:three-agent-value-monotone} and~\ref{thm:mms-impossibility}. The first strict inequality follows by cross-multiplication. For the second, each term $1/r$ with $3\le r\le n$ is at least $1/n$, so $H_n\ge1+1/2+(n-2)/n=5/2-2/n>2-1/n$, which is equivalent to the second strict inequality.
\end{proof}

\section{Verification Program for Theorem~\ref{thm:finite-exact}}
\label{app:program}
The following program performs the finite arithmetic in Appendix~\ref{app:exact-factors}, using only integer and rational arithmetic. It requires Python~3 and a C++17 compiler with OpenMP.

\begingroup\footnotesize
\begin{verbatim}
from fractions import Fraction as F
from pathlib import Path
import argparse
import os
import subprocess
import tempfile

FACTORS = {
    2: F(2, 3), 3: F(1, 2), 4: F(10, 21), 5: F(24, 55),
    6: F(15, 37), 7: F(42, 109), 8: F(18, 49), 9: F(6, 17),
    10: F(439, 1286), 11: F(540, 1631), 12: F(20159, 62560),
    13: F(72, 229), 14: F(504, 1639), 15: F(5171009, 17158995),
    16: F(812054744, 2745351561), 17: F(145645483, 500957470),
    18: F(102255661, 357395274), 19: F(4095, 14528),
    20: F(16380, 58931), 21: F(16380, 59711), 22: F(7371, 27205),
    23: F(27440, 102469), 24: F(82320, 310837), 25: F(493920, 1884779),
    26: F(735, 2833), 27: F(493920, 1922069), 28: F(493920, 1939709),
    29: F(1024380, 4058239), 30: F(204876, 818477),
    31: F(493920, 1989137), 32: F(123480, 501143),
    33: F(8350020, 34141547),
}
# n: (ordinary lower-bound cutoff K, one-prefix upper-bound x)
ORDINARY = {
    7: (2940, 216), 8: (1833, 146), 9: (2062, 152),
    10: (23998, 1285), 13: (30821, 1602), 14: (33192, 1638),
    19: (12746567, 523007), 20: (13417439, 530378),
    21: (14088311, 537398), 22: (10440075, 544099),
    23: (15417321, 1844439), 24: (16087639, 1865019),
    25: (13180423, 1884778), 26: (12667526, 1903773),
    27: (19599029, 1922068), 28: (20324919, 1939708),
    29: (358660146, 16232954), 30: (371027738, 16369538),
    31: (288515579, 1989136), 32: (297822533, 2004570),
    33: (7045486420, 273132374),
}
# n: (r, L, c, lambda, last global t)
SPECIAL = {
    4: (3, 19, 3, F(1, 6), 56),
    5: (5, 23, 2, F(1, 11), 67),
    6: (6, 15, 1, F(8, 111), 458),
    11: (7, 1332, 63, F(387, 8155), 615791),
    12: (7, 4475, 206, F(18, 391), 108056),
    15: (7, 4785, 206, F(4322, 100345), 1341555),
    16: (11, 68166, 1833, F(11915, 443013), 26892509),
    17: (11, 69352, 1833, F(6301, 238362), 29749879),
    18: (11, 70472, 1833, F(1451, 55782), 18104177),
}
# n: (x, r, M, stated f', stated lower bound on required fraction)
TWO_PREFIX = {
    6: (9, 6, 99, 0, F(159, 148)),
    11: (1325, 7, 11722, 0, F(57237, 57085)),
    12: (4468, 7, 1564, 1, F(22597967, 22521600)),
    15: (4778, 7, 100345, 0, F(6467213665753, 6183415438200)),
    16: (68155, 11, 443013, 0, F(19191933566131, 19025286317730)),
    17: (69341, 11, 953448, -1, None),
    18: (70461, 11, 334692, 0, F(86909440158353, 85860640625760)),
}


def ceil(x):
    return -(-x.numerator // x.denominator)


def harmonic(n):
    return sum((F(1, r) for r in range(1, n + 1)), F(0))


def extras(n, beta, x, counts):
    return [max(F(0), (beta * (x + r) - r * counts[r - 1])
                / (r * ((-x) % r))) if x % r else F(0)
            for r in range(1, n + 1)]


def minimum_counts(n, beta, x):
    return [(beta * (1 + x // r)).__floor__() + 1
            for r in range(1, n + 1)]


def check_upper_bounds():
    for n, x in [(4, 16), (5, 18)] + [(n, v[1]) for n, v in ORDINARY.items()]:
        beta = FACTORS[n]
        counts = minimum_counts(n, beta, x)
        free = n + x - sum(counts)
        terms = extras(n, beta, x, counts)
        total = beta * harmonic(n) + sum(sorted(terms)[:max(0, n - free)])
        assert free < 0 or total >= 1, (n, "one-prefix upper", free, total)
        assert (total == 1) if n in (4, 5) else (total > 1)
        print(f"one-prefix upper n={n}: f={free}, fraction={total}", flush=True)
    for n, (x, r, middle, stated_free, stated_total) in TWO_PREFIX.items():
        beta = FACTORS[n]
        counts = minimum_counts(n, beta, x)
        free = n + x - sum(counts)
        terms = extras(n, beta, x, counts)
        forced = [beta * harmonic(n) + e >= 1 for e in terms]
        assert free >= 0 and sum(forced) == free, (n, "forced counts")
        counts = [b + forced[s] for s, b in enumerate(counts)]
        d = (-x) % r
        h = (x + r + r - 1) // r
        assert x + r == r * h - d and 1 <= d < r and h >= d
        rate = (beta * h - counts[r - 1]) / d
        assert rate > 0 and middle > d
        assert rate == SPECIAL[n][3]
        second_x = x + middle
        second_counts = minimum_counts(n, beta, second_x)
        second_counts = [max(counts[s], second_counts[s]) for s in range(n)]
        second_counts[r - 1] = counts[r - 1] + (middle * rate).__floor__() + 1
        second_free = n + second_x - sum(second_counts)
        assert second_free == stated_free, (n, "second free", second_free)
        terms = extras(n, beta, second_x, second_counts)
        terms[r - 1] = F(0)
        total = beta * harmonic(n) + sum(sorted(terms)[:max(0, n - second_free)])
        if stated_total is not None:
            assert total == stated_total and total > 1, (n, total, stated_total)
            outcome = f"fraction={total}"
        else:
            assert sum(second_counts) == n + second_x + 1
            outcome = "one more good required than is available"
        print(f"two-prefix upper n={n}: forced={sum(forced)}, "
              f"f'={second_free}, {outcome}", flush=True)


def special_count(n, t, refined=True):
    a = FACTORS[n]
    r, length, c, rate, _ = SPECIAL[n]
    k = max(0, t - n + r)
    ordinary = sum(ceil(a * max(0, t - n + s) / s)
                   for s in range(1, n + 1) if s != r)
    if refined and n == 16 and k <= length:
        d = (-length) % r
        h = (length + d) // r
        lower_count = min(max(F(k, r + d), h - F((d + 1) * (length - k), r + d)),
                          max(F(0), h - F(length - k, r - d)))
        own = ceil(rate * k - (r * rate - a) * lower_count)
    else:
        own = min(c, ceil(rate * k)) + ceil(rate * max(0, k - length))
    return ordinary + own


def check_parameters():
    assert set(ORDINARY) | set(SPECIAL) == set(range(4, 34))
    for n, (cutoff, _) in ORDINARY.items():
        a = FACTORS[n]
        assert a * harmonic(n) < 1
        assert ceil(a * n / (1 - a * harmonic(n))) == cutoff
        for t in range(1, n + 1):
            assert sum(ceil(a * max(0, t - n + r) / r)
                       for r in range(1, n + 1)) <= t
    for n, (r, length, c, rate, last) in SPECIAL.items():
        a = FACTORS[n]
        d = (-length) % r
        h = (length + d) // r
        assert 1 <= d < r and rate == (a * h - c) / d
        assert a / r <= rate <= a / d and c <= rate * length
        if n == 16:
            assert h >= 2 and rate <= a / (r - 1)
        rho = a * (harmonic(n) - F(1, r)) + rate
        bound = n + c - rate * (n - r + length) - sum(
            (a * (n - s) / s for s in range(1, n + 1) if s != r), F(0))
        if n not in (4, 5, 6):
            assert rho < 1
            assert last == max(n, n - r + 1 + length, ceil(max(F(0), bound) / (1 - rho)))
        else:
            period, start, delta = {4: (42, 15, 42), 5: (55, 13, 55),
                                     6: (444, 15, 443)}[n]
            assert last == start + period - 1 and delta <= period
            assert all((period * a / s).denominator == 1
                       for s in range(1, n + 1) if s != r)
            assert (period * rate).denominator == 1
            assert period * rho == delta
            # Beyond both transitions, integer shifts prove periodicity.
            settled = max(n, n - r + length, n - r + ceil(F(c, 1) / rate))
            for t in range(start, settled + 1):
                assert special_count(n, t + period) - special_count(n, t) == delta
        # Exact fractional evaluation of the only refined finite segment.
        if n == 16:
            for t in range(1, n - r + length + 1):
                assert special_count(n, t) <= t, (n, t, special_count(n, t))
        print(f"parameters n={n}: correct", flush=True)


CPP_HEAD = r'''
#include <algorithm>
#include <atomic>
#include <cstdint>
#include <cstdlib>
#include <iostream>
#include <vector>
#include <omp.h>
using I = int64_t;
template<I P,I Q,int N> inline I ordinary_count(I x) {
  I result=0;
  // The generated compiler receives constant divisors after unrolling.
  #pragma GCC unroll 40
  for(int r=1;r<=N;++r) result+=(P*(x+r)-1)/(Q*r);
  return result;
}
template<I P,I Q,int N> void ordinary(I last) {
  if(ordinary_count<P,Q,N>(0)>0) std::abort();
  constexpr I block=1<<20;
  I lastz=(P*last-1)/Q;
  std::atomic<I> bad(-1);
  I direct=0;
  #pragma omp parallel for schedule(dynamic) reduction(+:direct)
  for(I base=0;base<=lastz;base+=block) {
    I end=std::min(base+block-1,lastz);
    std::vector<uint8_t> increment(end-base+1,0);
    I upper=0;
    for(int r=1;r<=N;++r) {
      I shift=P*r/Q;
      upper+=(base+shift)/r;
      I first=(r-(base+1+shift)%r)%r;
      for(I j=first;j<=end-base;j+=r) ++increment[j];
    }
    for(I z=base;z<=end;++z) {
      upper+=increment[z-base];
      I first=(Q*z+P)/P;
      if(upper<=first) continue;
      I endx=std::min(last,(Q*(z+1)+P)/P-1);
      for(I x=first;x<=endx;++x) {
        ++direct;
        if(ordinary_count<P,Q,N>(x)>x) bad.store(x);
      }
    }
  }
  if(bad.load()>=0) {
    std::cerr<<"VIOLATION ordinary n="<<N<<" x="<<bad.load()<<std::endl;
    std::exit(1);
  }
  std::cout<<"ordinary lower n="<<N<<": all x=0.."<<last
           <<" checked; "<<direct<<" direct evaluations"<<std::endl;
}
template<I P,I Q,int N,int R,I L,I C,I LP,I LQ>
void special(I first,I last) {
  std::atomic<I> bad(-1);
  #pragma omp parallel for schedule(static)
  for(I t=first;t<=last;++t) {
    I total=0;
    #pragma GCC unroll 40
    for(int s=1;s<=N;++s) if(s!=R) {
      I k=std::max(I(0),t-N+s);
      total+=(P*k+Q*s-1)/(Q*s);
    }
    I k=std::max(I(0),t-N+R);
    total+=std::min(C,(LP*k+LQ-1)/LQ);
    total+=(LP*std::max(I(0),k-L)+LQ-1)/LQ;
    if(total>t) bad.store(t);
  }
  if(bad.load()>=0) {
    std::cerr<<"VIOLATION special n="<<N<<" t="<<bad.load()<<std::endl;
    std::exit(1);
  }
  std::cout<<"special lower n="<<N<<": all t="<<first<<".."<<last
           <<" checked"<<std::endl;
}
int main(int argc,char**argv) {
  omp_set_num_threads(argc>1?std::atoi(argv[1]):4);
'''


def check_deadlines(threads):
    calls = []
    for n, (r, length, c, rate, last) in SPECIAL.items():
        a = FACTORS[n]
        first = n - r + length + 1 if n == 16 else 1
        args = (a.numerator, a.denominator, n, r, length, c,
                rate.numerator, rate.denominator)
        calls.append('special<' + ','.join(map(str, args)) + f'>({first},{last});')
    for n, (last, _) in ORDINARY.items():
        a = FACTORS[n]
        calls.append(f'ordinary<{a.numerator},{a.denominator},{n}>({last}LL);')
    with tempfile.TemporaryDirectory(prefix='verify-mms-') as directory:
        source = Path(directory) / 'verify.cpp'
        executable = Path(directory) / 'verify'
        source.write_text(CPP_HEAD + '\n'.join(calls) + '\n}\n')
        subprocess.run(['g++', '-O3', '-std=c++17', '-fopenmp', str(source),
                        '-o', str(executable)], check=True)
        subprocess.run([str(executable), str(threads)], check=True)


def main():
    parser = argparse.ArgumentParser(description=__doc__)
    parser.add_argument('--threads', type=int, default=min(8, os.cpu_count() or 1))
    args = parser.parse_args()
    check_upper_bounds()
    check_parameters()
    check_deadlines(args.threads)
    print('PASS: all finite tables and deadline inequalities verified exactly.')


if __name__ == '__main__':
    main()
\end{verbatim}
\endgroup

\end{document}